\documentclass[11pt,letterpaper]{article}
\usepackage[utf8]{inputenc}
\usepackage[english]{babel}
\usepackage[margin=1in]{geometry} 
\usepackage{amsmath,amsthm,amssymb,amsfonts,bbm}
\usepackage{titling}
\usepackage{graphicx}
\usepackage{verbatim}
\usepackage{xcolor}
\usepackage{calc}
\usepackage[boxed]{algorithm2e}   
\usepackage{bbm}
\usepackage{mathrsfs}
\usepackage{tikz}
\usepackage{authblk}
\usepackage{ulem}
\usepackage{array}   

\usepackage{hyperref}
\hypersetup{colorlinks=true, linkcolor=darkred, citecolor=darkgreen, urlcolor=darkblue} 
\usepackage[round]{natbib} 

\usepackage{booktabs} 
\usepackage{arydshln} 

\usepackage{wrapfig}
\usepackage{lscape}
\usepackage{rotating}
\usepackage{epstopdf}

\usepackage{float}  

\usepackage{mathtools}
\mathtoolsset{showonlyrefs} 

\usepackage{xcolor} 

\usepackage{caption}  
\usepackage{subcaption} 

\usepackage{color}
\definecolor{darkred}{RGB}{100,0,0}
\definecolor{darkgreen}{RGB}{0,100,0}
\definecolor{darkblue}{RGB}{0,0,150}

\usepackage{accents}

\usepackage{mathtools}
\mathtoolsset{showonlyrefs}   

\newcommand{\ee}[1]{E\left[#1\right]}

\newcommand{\pp}[1]{\text{pr}\left(#1\right)}
\newcommand{\ppo}[1]{\text{pr}^*\left(#1\right)}
\newcommand{\E}{E}
\newcommand{\PP}{\text{pr}}
\newcommand{\var}{\text{var}}

\newcommand{\Sc}{\mathcal{S}}
\newcommand{\Z}{\mathcal{Z}}
\newcommand{\I}{\mathcal{I}}
\newcommand{\HH}{\mathcal{H}}
\newcommand{\A}{\mathcal{A}}
\newcommand{\B}{\mathcal{B}}
\newcommand{\D}{\mathcal{D}}

\newcommand{\convp}{\stackrel{p}{\rightarrow}}  
\newcommand{\convd}{\stackrel{d}{\rightarrow}}  

\newcommand{\TV}{\text{TV}}

\newcommand{\bigCI}{\mathrel{\text{\scalebox{1.07}{$\perp\mkern-10mu\perp$}}}}

\newtheorem{theorem}{Theorem}
\newtheorem{proposition}{Proposition}

\newtheorem{lemma}{Lemma}
\newtheorem{assumption}{Assumption}

\newtheorem{assumptionalt}{Assumption}[assumption]

\newenvironment{assumptionp}[1]{
  \renewcommand\theassumptionalt{#1}
  \assumptionalt
}{\endassumptionalt}

\newtheoremstyle{bfnote}%
 {}{}%
 {}{}%
{\bfseries}{.}%
{ }%
{\thmname{#1}\thmnumber{ #2}\thmnote{ (#3)}}

\theoremstyle{bfnote}
\newtheorem{remark}{Remark}

\title{\textbf{Doubly Robust Estimation under Covariate-Induced Dependent Left Truncation}}
\author[1]{Yuyao Wang}
\author[2]{Andrew Ying}
\author[1,3]{Ronghui Xu}
\affil[1]{Department of mathematics, University of California San Diego}
\affil[2]{Department of Statistics and Data Science, The Wharton School, University of Pennsylvania}
\affil[3]{Herbert Wertheim School of Public Health, and Halicioglu Data Science Institute, University of California San Diego}
\date{}

\begin{document}
\maketitle

\begin{abstract}
In prevalent cohort studies with follow-up, the time-to-event outcome is subject to left truncation leading to selection bias. For estimation of the distribution of time-to-event, conventional methods adjusting for left truncation tend to rely on the (quasi-)independence assumption that the truncation time and the event time are ``independent" on the observed region. This assumption is violated when there is dependence between the truncation time and the event time possibly induced by measured covariates. Inverse probability of truncation weighting leveraging covariate information can be used in this case, but it is sensitive to misspecification of the truncation model. In this work,  we apply the semiparametric theory to find the efficient influence curve of an expected  (arbitrarily transformed) survival time in the presence of covariate-induced dependent left truncation. We then use it to construct estimators that are shown to enjoy double-robustness properties. Our work represents the first attempt to construct doubly robust estimators in the presence of left truncation, which does not fall under the established framework of coarsened data where doubly robust approaches are developed. We provide technical conditions for the asymptotic properties that appear to not have been carefully examined in the literature for time-to-event data, and study the estimators via extensive simulation. We apply the estimators to two data sets from practice, with different right-censoring patterns.
\end{abstract}

\section{Introduction}

A time-to-event outcome in a prevalent cohort study is subject to left truncation because, starting from a well-defined time zero, such as the onset of a disease, often only subjects with time to events such as death greater than the enrollment times are included in the data set. Subjects with early or short event times therefore tend not to be captured in the data.
For example, in prospective cohort pregnancy studies, women typically enroll after clinical recognition of their pregnancies, so women with early pregnancy losses tend not to be included in the data \citep{xu:cham, ying2020causal}. 
As another example, in aging studies, age is the time scale of interest, so the time origin is  birth. However, in many such studies participants are recruited into the study at different ages, and participants with early event times tend not to be captured due to the study design.

In the presence of left truncation, subjects with longer event times are preferentially selected.
In the estimation of the marginal distribution of the time to event, which is the focus of this paper, 
conventional methods typically rely on the random truncation assumption that the left truncation time and event time are independent in the full data, i.e.~without truncation \citep{woodroofe1985estimating, wang1986asymptotic, wang1989semiparametric, wang1991nonparametric, gross1996weighted, gross1996nonparametric, shen2010semiparametric}. 
This can also be weakened to a quasi-independence assumption \citep{tsai1990testing}. Such independence or quasi-independence assumptions may be violated in practice.  
For instance, in the retirement center example of \citet{klein2003survival}, the life lengths of the individuals in the retirement center are left truncated because individuals must survive long enough to enter the retirement center. However, individuals' life lengths and entry times may be dependent because individuals who entered the retirement home earlier may have received better medical attention and therefore lived longer \citep{chaieb2006estimating}.

For dependent left truncation, copula models \citep{chaieb2006estimating, emura2011semi,emura2012nonparametric} and structural transformation models \citep{efron1994survival, chiou2019transformation} have been proposed to handle the dependence between the left truncation time and the event time under different assumptions. 
Copula models rely on strong modeling assumptions for the dependency, whilst structural transformation models specify a latent quasi-independent truncation time as a function of the observed dependent truncation time and the event time. 
There are also methods that incorporate the left truncation time as a covariate in the event time model \citep{mackenzie2012survival, cheng2015causal} which, as pointed out in \citet{vakulenko2022nonparametric}, is biologically unjustified if the left truncation time is study-specific and not related to the event process.
None of the above methods leverage covariate information if available;  a more detailed discussion of the above approaches can be found in \citet{vakulenko2022nonparametric}.

When the dependence between the left truncation and the event times is captured by measured covariates in a regression setting, 
\citet{cox1972regression} model for example, can properly handle the left truncation via
risk set adjustment \citep{andersen1993statistical}.
For marginal survival probabilities, 
 \citet{vakulenko2022nonparametric} recently proposed two inverse probability weighted (IPW) estimators  that account for the dependence  induced by the covariates. 
 IPW approaches \citep{rosenbaum1983central} are known to be not only sensitive to model specification but also inefficient. 
This motivated us to seek estimators that provide extra protection against model misspecification and are more efficient. To that end, we leverage semiparametric theory to find the efficient influence curve (EIC), which often suggests reasonable estimators in both the missing data  and the causal inference literature \citep{tsiatis2006semiparametric, hernan2020causal}. 
In our case, we show that the EIC leads to estimators that enjoy model double robustness and rate double robustness \citep{rotnitzky2021characterization, hou2021treatment}.
The model doubly robust estimator is consistent and asymptotically normal (CAN) when one of the two parametric/semiparametric nuisance models is correctly specified, but not necessarily both. In addition, when both models are correctly specified, the estimator achieves the semiparametric efficiency bound.
The rate doubly robust estimator is constructed from a cross-fitting procedure.
It is CAN and achieves the semiparametric efficiency bound when both nuisance parameters are consistently estimated and the product error rate of the two nuisance estimators is faster than root-$n$. This allows us to use flexible nonparametric or machine learning methods to estimate the nuisance parameters. 

We remark that there has been a growing literature on doubly robust estimators to handle missing or coarsened data including right censoring for time-to-event data, as well as non-randomized treatment assignment in causal inference \citep[for examples]{robins1995analysis, bang2005doubly, tsiatis2006semiparametric, tchetgen2010doubly, robins2005inverse, robins1997causal, robins1997marginal, robins1999association, van2003unified}. 
However, left truncation does not belong to the coarsened data framework. This is because for coarsened data, subjects in the observed data set are a random sample from the population of interest, while the data for each subject may be coarsened. For truncation problems, on the other hand,  subjects in the observed data set are biasedly sampled. Therefore, the existing techniques for coarsened data \citep{heitjan1991ignorability, tsiatis2006semiparametric} do not apply.
To the best of our knowledge, our paper represents the first attempt to apply semiparametric theory in order to explore doubly robust estimation in the presence of  left truncation, a common type of selection bias for time-to-event data. 

The rest of the paper is organized as follows. 
We introduce notation in the next section and review preliminaries on left truncation that are not widely known.  
In Section \ref{sec:AIPW}, we show that the estimand can be identified from inverse probability of truncation weighting, derive the EIC 
and construct an estimating function that is doubly robust in the sense that it has expectation zero when either the time-to-event model or the truncation model is true. 
In Section \ref{sec:estimation} we construct a model doubly robust estimator and a rate doubly robust estimator, and carefully establish their asymptotic properties. 
We extend the estimators to two different right-censoring scenarios in Section \ref{sec:censoring}. 
Extensive simulations are conducted in Section \ref{Sec:simulation}, and
Section \ref{sec:application} contains applications to two data sets under the different right-censoring scenarios. 
Section \ref{sec:discussion} concludes with discussion.

\section{Notation and Preliminaries}\label{sec:notation_preliminary}

Let $Q$, $T$, $Z$ be random variables denoting the left truncation time, the event time of interest, and  baseline covariates, respectively. 
Suppose that $Q$ and $T$  both have absolutely continuous distributions. Let $F$, $G$ and $H$ be the conditional cumulative distribution function (CDF) of $T$ given $Z$, $Q$ given $Z$, and the CDF of $Z$ in the full data, respectively; here by full data we mean  if there were no left truncation. Let $f$, $g$, and $h$ be the corresponding full data densities or probability functions in the case of discrete $Z$. 

Denote $O=(Q,T,Z)$. 
In the presence of left truncation, we observe $O$ only if $Q < T$. The observed data distribution of $O$ is the conditional full data distribution of $O$ given $Q<T$. In this paper, we will denote $\PP^*$ as the probability operator of the full data,  $\E^*$ as the expectation with respect to the full data distribution, `$\PP$' as the probability operator of the observed data, i.e.~$\PP(\cdot) = \PP^*(\cdot|Q<T)$, and $\E$ as the expectation  with respect to the observed data distribution. In addition, we will use $P$ and $p$ to denote the CDF and the density or probability function for the observed data; for example, $P_{Q|Z}$, $p_{Q|Z}$, $P_{Q,Z}$ and $p_{Q,Z}$, etc.

 In the following we characterize the distribution of $Q$ including its associated martingale, which will be useful later. 
In particular, due to the asymmetry between $Q$ and $T$, it is often useful to consider the ``reverse time" for the distribution of $Q$. In fact, without covariates  the product-limit estimate of $G$ can be derived by moving backwards in time \citep{wang1991nonparametric, gross1996weighted}.
\citet{bickel1993efficient} considered the ``reverse time" counting processes, and here we extend them to the setting with covariates. 
For $t\geq 0$, let
\begin{align}
    \bar N_Q(t) = \mathbbm{1}(t\leq Q<T), \quad \bar N_T(t) = \mathbbm{1}(t\leq T),
\end{align}
with their natural history filtration $\{\mathcal{\bar F}_t\}_{t\geq 0}$ on the reversed time scale:
\begin{align}
\mathcal{\bar F}_t 
    &= \sigma \left\{Z, \mathbbm{1}(Q<T), \mathbbm{1}(s\leq T),\mathbbm{1}(s\leq Q<T): s\geq t  \right\}.  \label{eq:F_t}
\end{align}
Note that $\bar N_Q$ is a decreasing process (as $t$ increases) that starts at one and jumps to zero immediately after $t=Q$ if $Q<T$, and is always zero if $Q\geq T$.
Also $\{\mathcal{\bar F}_t\}_{t\geq 0}$ decreases as $t$ increases. 

Let 
$\alpha$ be the reverse time hazard function of $Q$ given $Z$ for the full data; 
specifically, 
\begin{align}
\alpha(q|z) &= \lim_{h\to 0+} \frac{\ppo{q-h< Q\leq q|Q\leq q, Z = z}}{h} 
= \frac{g(q|z)}{G(q|z)}. \label{eq:alpha_def}
\end{align}
It follows immediately that $G(q|z) = \ppo{Q\leq q|Z = z} = \exp \{ - \int_q^\infty \alpha(v|z) dv \}$.
In this way either $\alpha$ or $G$ alone characterizes the full data distribution of $Q$ given $Z$. 

Define the compensator $\bar A_Q(\cdot;G)$ for $\bar N_Q$, where
\begin{align}
    \bar A_Q(t;G) 
    &= \int_t^\infty \mathbbm{1}(Q\leq s <T) \alpha(s|Z) ds = \int_t^\infty \mathbbm{1}(Q\leq s <T) \frac{dG(s|Z)}{G(s|Z)}.
\end{align}
 In the Supplementary Material we showed that  
\begin{align}
    \bar M_Q(t; G) = \bar N_Q(t) - \bar A_Q(t;G)
\end{align}
is a {\it backwards martingale} with respect to $\{\mathcal{\bar F}_t\}_{t\geq 0}$ under the observed data law, if $G$ is the true CDF of $Q$ given $Z$ in the full data.
A stochastic process $\{Y_t\}_{t\geq 0}$ is called a backwards martingale with respect to a set of decreasing $\sigma$-algebras $\{\mathcal{\bar F}_t\}_{t\geq 0}$ 
if $\E(|Y_t|) <\infty$, $Y_t$ is $\mathcal{\bar F}_t$-measurable for all $t\geq 0$, and $\E(Y_s|\mathcal{\bar F}_t) = Y_t$ for all $0\leq s\leq t$. 
If we reverse the time scale and define 
    \begin{align}
         M_Q^\tau(t; G) &= \bar M_Q(\tau-t;G), \quad \forall\ t\geq 0, \label{eq:M^tau}\\
         \mathcal{G}^\tau_t &= \mathcal{\bar F}_{\tau-t},\quad \forall\ t\geq 0\label{eq:G_t},
    \end{align}
    for some $\tau>0$,
then $\{M_Q^\tau(t;G_0): 0\leq t\leq\tau\}$ is a martingale with respect to the filtration $\{\mathcal{G}^\tau_t\}_{0\leq t\leq\tau}$.

In this paper we focus on estimating the expectation of a transformed event time in the full data:
\begin{align}
\theta = \E^*\{ \nu(T) \}, \label{eq:theta}
\end{align}
where $\nu$ is a real-valued bounded function. 
Expression \eqref{eq:theta} includes commonly considered estimands for time-to-event data. For example, when $\nu(t) = \mathbbm{1}(t>t_0)$ for some fixed $t_0>0$, $\theta = \PP^*(T>t_0)$ is the survival probability; 
when $\nu(t) = \min(t,t_0)$, $\theta$ is the restricted mean survival time (RMST). 

\section{Identification, influence curve and Double Robustness}
\label{sec:AIPW}

\subsection{Identification via inverse probability weighting }\label{sec:IPW}

We make the following conditional quasi-independence assumption.
\begin{assumption}[Conditional quasi-independence]\label{ass:quasi-indpendent}
	The observed data density for $(Q,T,Z)$ satisfies
	\begin{align}
	&p_{Q,T,Z}(q,t,z) 
	= \left\{
	\begin{array}{ll}
	f(t|z) g(q|z) h(z)/\beta, & \quad \text{if } t > q; \\
	0, & \quad \mbox{otherwise},
	\end{array} \label{eq:quasi-indpendent}
	\right.
	\end{align}
	where $\beta = \PP^*(Q<T) = \int \mathbbm{1}(q<t) f(t|z)g(q|z)h(z)\ dt\ dq\ dz$.
\end{assumption}
	Assumption \ref{ass:quasi-indpendent} is a generalization of the quasi-independence assumption in \citet{tsai1990testing} to settings where the event time and the left truncation time are dependent via covariates. It is weaker than the conditional independence assumption that $Q$ is independent of $T$ given $Z$ in the full data, which was assumed in \citet{ vakulenko2022nonparametric}, because Assumption \ref{ass:quasi-indpendent} only constraints the distribution of $(Q,T,Z)$ on the observed region $\{(q,t,z): q<t\}$.
    In addition, it does not imply that $Q$ is independent of $T$ given $Z$ in observed data since $p_{Q,T|Z} \neq p_{Q|Z} p_{T|Z}$, as shown in the Supplementary Material. 

We also need the following positivity assumption, which ensures that there are enough observed data to identify the parameter of interest.

\begin{assumption}[Positivity]\label{ass:beta*(z)>0}
	$\PP^*(Q<T|Z) > 0$ almost surely.
\end{assumption}

Under Assumptions \ref{ass:quasi-indpendent} and \ref{ass:beta*(z)>0}, the estimand $\theta$ can be identified from the observed data distribution. 
This is shown by an inverse probability weighting type of argument below, where the inverse probability weights are constructed using the truncation distribution.

\begin{lemma}[Identification]\label{lem:psi_identification}
	Under Assumptions \ref{ass:quasi-indpendent} and \ref{ass:beta*(z)>0},
	\begin{align}
	\theta = \left.\E\left\{\frac{\nu(T)}{G(T|Z)}\right\}\right/\E\left\{\frac{1}{G(T|Z)}\right\},\label{eq:identify_psi}
	\end{align}
where 
$G(q|z) = \exp \{ - \int_q^\infty \alpha(v|z) dv \}$, and 
$\alpha(q|z)= p_{Q|Z}(q|z)/\PP(Q\leq  q < T|Z=z)$ is identified  by the observed data distribution.  
\end{lemma}

Expression \eqref{eq:identify_psi} suggests an inverse probability of truncation weighted estimator for $\theta$ from an observed sample $(Q_i,T_i,Z_i)_{i=1}^n$: 
\begin{align}
\hat\theta_{\text{IPW.Q}} = \left.\left\{\frac{1}{n}\sum_{i=1}^n \frac{\nu(T_i)}{\hat G(T_i|Z_i)}\right\}\right/\left\{\frac{1}{n}\sum_{i=1}^n \frac{1}{\hat G(T_i|Z_i)}\right\}, \label{eq:est_IPW.Q}
\end{align}
where  $\hat G$ is an estimator of $G$ that will be discussed in more details later. 
When $\nu(t) = \mathbbm{1}(t>t_0)$ for a given $t_0$, this estimator coincides with the `CW' estimator proposed in \cite{vakulenko2022nonparametric}, when $\hat G$ is estimated from a Cox proportional hazards model.  

As mentioned in the Introduction, the IPW estimator \eqref{eq:est_IPW.Q}  relies on the correct specification of the truncation model and is known to be inefficient. Augmented inverse probability weighted (AIPW) estimators are often developed in the literature  to improve upon the IPW estimators; they have been shown to provide extra robustness and to be more efficient \citep{robins1997marginal, van2003unified, tsiatis2006semiparametric}. In order to develop the improved estimators, we leverage the semiparametric theory below.

\subsection{Efficient influence curve}

Semiparametric theory has been developed and reviewed in the literature \citep{begun1983information, bickel1993efficient,  tsiatis2006semiparametric, van2000asymptotic,  newey1994asymptotic, kosorok2008introduction}.
In the following, we apply the semiparametric theory to first compute the efficient influence curve  of $\theta$, which in turn suggests estimators that have favorable properties such as double robustness. 

Consider the Hilbert space $\HH$ of all one-dimensional mean-zero measurable functions of $O=(Q, T, Z)$ with finite second moments, equipped with the covariance inner product under the observed data distribution $P$. 
We proceed as follows: we will first characterize the tangent space in $\HH$; 
we then introduce and find an influence curve ; the projection of the influence curve onto the tangent space gives the efficient influence curve \citep{tsiatis2006semiparametric, bickel1993efficient}. 

Consider a regular \citep{bickel1993efficient} parametric submodel $P_\epsilon$ indexed by a real-valued parameter $\epsilon$ for the observed data distribution of $O$ that satisfies Assumptions \ref{ass:quasi-indpendent} and \ref{ass:beta*(z)>0}, and equals the true distribution of the observed data at $ \epsilon = 0$.  
Its {\it score function} is
\begin{align}
\Sc(O) = \left.\frac{\partial}{\partial \epsilon}\log p_\epsilon(O) \right|_{\epsilon = 0};
\end{align}
in the notation above as well as in the rest of the paper we suppress the dependence of $\Sc$ on the parametric submodel. 
The closure of the linear span of the score functions of all such regular parametric submodels is  the {\it tangent space}, denoted by $ \dot P $. 

Denote $L^2_0(P_{T,Z})$ the space of all measurable functions of $(T,Z)$ with mean zero  and finite second moment under the law $P_{T,Z}$, and $L^2_0(P_{Q,Z})$ the space of all measurable functions of $(Q,Z)$ with  mean zero and finite second moment  under the law $P_{Q,Z}$. 
The following Lemma is proved in the Supplementary Material. 

\begin{lemma}[Tangent space]\label{lem:tangent_space}
	The tangent space under the semiparametric model imposed by Assumptions \ref{ass:quasi-indpendent} and \ref{ass:beta*(z)>0} is
	\begin{align}
	\dot P = L^2_0(P_{T,Z}) + L^2_0(P_{Q,Z}).  \label{eq:tangent_space}
	\end{align}
	In addition, the orthogonal complement of the tangent space is
	\begin{align}
	(\dot P)^\perp  &= \left\{ \xi\in \HH:
	\begin{array}{l}
	\E\left\{\xi(Q,T,Z)a(Q,Z)\right\} = 0, \ \ \text{for all } a \in L^2_0(P_{Q,Z});\\
	\E\left\{\xi(Q,T,Z)b(T,Z)\right\} = 0, \ \ \text{for all } b \in L^2_0(P_{T,Z}).
	\end{array}
	\right\}.
	\end{align}
\end{lemma}

We now introduce the influence curve.
The parameter of interest, $\theta = \theta(P)$, is {\it pathwise differentiable} with {\it influence curve} $\varphi \in \HH$ if, for any regular parametric submodel $\{P_\epsilon: \epsilon\in \mathbb{R}\}$,
\begin{align}\label{identity}
\left.\frac{\partial}{\partial \epsilon} \theta(P_\epsilon) \right|_{\epsilon = 0} = \int \varphi(o) \Sc(o) dP(o) = \E\left\{\varphi(O)\Sc(O)\right\}.
\end{align}
\begin{remark}
In the literature, there are two practically exchangeable terms: influence function (IF) and influence curve (IC). 
Here we follow the convention as reviewed in \citet{kennedy2017semiparametric}, to use IF for estimators, and use IC for parameters. 
 Note that if $\varphi$ is the IF for an arbitrary regular asymptotically linear (RAL) estimator \citep{tsiatis2006semiparametric} of $\theta$, then the IF for any RAL estimator of $\theta$ must lie in the space $\{\varphi(O)+ (\dot P)^\perp\}$ \citep{tsiatis2006semiparametric, bickel1993efficient}. 
 RAL estimators and their IF's will be used later in the paper. The IF of a RAL estimator suggests the asymptotic variance of the estimator. 
\end{remark}

For the remainder of this paper, we consider the semiparametric model under Assumptions \ref{ass:quasi-indpendent} and \ref{ass:beta*(z)>0}, and assume that  the true distribution also satisfies the following overlap assumption, which is stronger than Assumption \ref{ass:beta*(z)>0}. We borrow the word `overlap' from the causal inference literature, which is also sometimes referred to as a strict positivity assumption.

\begin{assumption}[Overlap]\label{ass:overlap}
	There exist $0<\tau_1 <\tau_2 <\infty$ and constants $\delta_1, \delta_2>0$ such that $T\geq \tau_1$ a.s. and $Q\leq \tau_2$ a.s. in the full data;  
	$1-F(\tau_2|Z)\geq \delta_1$ a.s., and $G(\tau_1|Z)\geq \delta_2$ a.s.. 
\end{assumption}


Assumption \ref{ass:overlap} makes sure that the minimum support of $T$ is greater than the minimum support of $Q$, and that the maximum support of $T$ is greater than the maximum support of $Q$. This assumption is commonly made for left truncated data, see for example  \citet{cheng2012estimating}. 
The existance of $\tau_1 $ and $\tau_2 $ 
is in line with most survival analysis approaches where a maximum follow-up time is assumed, so that the parameter space for a cumulative baseline hazard, for example, is bounded. It is also common for left truncated data, that any reasonable estimand, is conditional upon survival until a minimum truncation time when the data become available \citep{tsai1987note}. 

Assumption \ref{ass:overlap} is simpler but stronger than Assumption \ref{ass:overlap2} below, which is in fact what is necessary. 
For the ease of exposition, we will continue with the Assumption \ref{ass:overlap}. 
For data examples, however, we will verify Assumption \ref{ass:overlap2} which is more likely to be satisfied in practice. 
Denote $\Z$  the support of $Z$ in the full data. 

\begin{assumptionp}{3'}[overlap]\label{ass:overlap2}
    There exist constants $\delta_1, \delta_2>0$ and $0<\tau_1(z) <\tau_2(z) <\infty$ for each $z\in\Z$ 
    such that conditional on $Z=z$, 
    $T\geq \tau_1(z)$ a.s. and $Q\leq \tau_2(z)$ a.s. in the full data; and with probability one, 
    $1-F\{\tau_2(Z)|Z\}\geq \delta_1$ or $1-F\{T|Z\}\geq \delta_1$, and 
	$G\{\tau_1(Z)|Z\}\geq \delta_2$ or $G\{Q|Z\}\geq \delta_2$.
\end{assumptionp}



\begin{lemma}\label{lem:IF}
	Under the semiparametric model imposed by Assumptions \ref{ass:quasi-indpendent} and \ref{ass:beta*(z)>0}, suppose that the true distributions $F$ and $G$ satisfy Assumption \ref{ass:overlap}, then
 \begin{align}
	\varphi(O;\theta, F,G,H) = \beta  \left[\frac{\nu(T)-\theta}{G(T|Z)} - \int_0^\infty \frac{\int_0^v \{\nu(t) -\theta\} dF(t|Z)}{ 1-F(v|Z) }  \frac{d\bar M_Q(v;G)}{G(v|Z)}  \right] \label{eq:EIC}
  \end{align}
	is an influence curve for $\theta$, where 
 $\beta=pr^*(Q<T)$ is defined in Assumption \ref{ass:quasi-indpendent}.
\end{lemma}

It turns out that $\varphi$ given in \eqref{eq:EIC} lies in $\dot P$, and is therefore the efficient influence curve (EIF). We summarize this result formally  below.

\begin{proposition}[Efficient influence curve]\label{thm:EIF}
	Under the semiparametric model imposed by Assumptions \ref{ass:quasi-indpendent} and \ref{ass:beta*(z)>0} and at the law where Assumption \ref{ass:overlap} holds, the influence curve $\varphi$ given in \eqref{eq:EIC} is the efficient influence curve for $\theta$, and hence the semiparametric efficiency bound for estimating $\theta$ is $\E(\varphi^2)$.
\end{proposition}

\begin{remark}
	\citet{chao1987influence} derived the influence function for the product-limit (PL) estimator 
    of the survival function of $T$ under random left truncation, in the absence of covariates. Since the PL estimator 
    is the nonparametric maximum likelihood estimator (NPMLE), which is known to be asymptotically efficient, the IF derived in \citet{chao1987influence} is the EIF. 
    In the Supplementary Material we showed that when there are no covariates and $\theta = \E^*\{\mathbbm{1}(T>t_0)\} = 1- F(t_0)$, the EIC derived in this paper matches the (E)IF derived in \citet{chao1987influence}. 
\end{remark}

\subsection{Doubly robust estimating function}

Leaving out the constant factor $\beta$ in \eqref{eq:EIC}, we have the following estimating function,  assuming that $F$ and $G$ are known for the moment:
\begin{align}
U(\theta; F,G) = \frac{\nu(T)-\theta}{G(T|Z)} - \int_0^\infty \frac{\int_0^v \{\nu(t) -\theta\} dF(t|Z) }{ 1-F(v|Z) } \cdot \frac{d\bar M_Q(v;G)}{G(v|Z)},  \label{eq:score}
\end{align}
where $\bar M_Q(\cdot;G)$ is the backwards martingale related to $Q$ that is defined in the preliminaries. 
The above can be viewed as the  inverse  probability of truncation weighted full data estimating function $\nu(T)- \theta$, plus an augmentation term. We note that the augmentation term  can also be written as 
\begin{align}
- \int_0^\infty \E^* \{ \nu(T)-\theta | T<v, Z \}
\frac{F(v|Z)}{1-F(v|Z)} \cdot
 \frac{d\bar M_Q(v;G)}{G(v|Z)}. \label{eq:augmentation}
\end{align}
The above  resembles other augmented IPW (AIPW) estimating functions. 
For example, for the augmented inverse probability censoring weighted (AIPCW) estimating functions \citep{tsiatis2006semiparametric}, the augmentation term involves  a weighted censoring time martingale integral, with the integrand of a full data conditional expectation. 
In comparison, \eqref{eq:augmentation} is a weighted truncation time (backwards)
martingale integral, with the integrand also a full data conditional expectation, but with the additional  multiplicative factor  of the odds $ \{F(v|Z)\}/\{1-F(v|Z)\}$.

Denote $\theta_0, F_0, G_0$ the true values of $\theta, F, G$, respectively. The following theorem shows that $U$ is a doubly robust estimating function.

\begin{theorem}[Population Double robustness]\label{thm:DR}
	Under Assumption \ref{ass:quasi-indpendent}, suppose that $F$ and $G$ satisfy Assumption \ref{ass:overlap}, then 
	$\E\{ U(\theta_0; F, G) \} = 0$
	if either $F=F_0$ or $G=G_0$. 
\end{theorem}

Theorem \ref{thm:DR} helps us to construct doubly robust estimators next.

\section{Estimation}\label{sec:estimation}

\subsection{Estimating equation} 

Let $\{O_i\}_{i=1}^n$ be a random sample of size $n$ from the observed data distribution, where $O_i = (Q_i,T_i,Z_i)$. 
From  \eqref{eq:score} a natural estimating equation for $\theta$ would be
$ \sum_{i = 1}^n U_i(\theta; F, G) = 0$,  
if $F$ and $G$ were known. 
In practice we would first estimate $F$ and $G$,  then solve for $\theta$ using
\begin{equation}
\sum_{i = 1}^n U_i(\theta; \hat F, \hat G) = 0 \label{eq:est_eq}.
\end{equation}
In particular, $F$ can be estimated from  $(Q_i,T_i,Z_i)_{i=1}^n$ using existing regression methods that handle  left truncation; and $G$ can be estimated similarly from  $(\tau-T_i, \tau-Q_i,Z_i)_{i=1}^n$ on the reversed time scale, where $\tau-Q_i$ is left truncated by $\tau-T_i$, and $\tau<\infty$ is an upper bound of time. 
Since $U(\theta;\hat F,\hat G)$ is linear in $\theta$, we have a closed-form solution for \eqref{eq:est_eq}: 
\begin{align}
\hat\theta &= \left(\sum_{i=1}^n \left[\frac{1}{\hat G(T_i|Z_i)} - \int_0^\infty \frac{\hat F(v|Z_i)}{\hat G(v|Z_i)\{1-\hat F(v|Z_i)\}} d\bar M_{Q,i}(v; \hat G) \right]\right)^{-1} \nonumber \\
&\quad\quad\quad\quad  \times \left(\sum_{i=1}^n \left[\frac{\nu(T_i)}{\hat G(T_i|Z_i)} - \int_0^\infty \frac{\int_0^v \{\nu(t) -\theta\} d\hat F(t|Z_i)}{\hat G(v|Z_i)\{1-\hat F(v|Z_i)\}} d\bar M_{Q,i}(v; \hat G) \right] \right). \label{eq:est_DR}
\end{align}

We have immediately the following special cases in which we use degenerative constant (and likely wrong) estimates for $\hat F$ or $\hat G$. The double robustness properties in the next subsection guarantees that they are CAN if conditions are met. 
More specifically, by setting $\hat F \equiv 0$, we obtain the IPW estimator $\hat\theta_{\text{IPW.Q}}$ in \eqref{eq:est_IPW.Q}.
On the other hand, by setting $\hat G \equiv 1$, we obtain a regression-based estimator:
\begin{align}
\hat\theta_{\text{Reg.T1}}
& = \left\{\sum_{i=1}^n \frac{1}{1-\hat F(Q_i|Z_i)}\right\}^{-1} 
\left[\sum_{i=1}^n \frac{\nu(T_i)\{1-\hat F(Q_i|Z_i)\} + \int_0^{Q_i} \nu(t)  d\hat F(t|Z_i) }{1-\hat F(Q_i|Z_i) }\right]. \label{eq:est_Reg.T1}
\end{align}
The reason we call it regression-based is shown 
in the Appendix: $\nu(T)\{1- \hat F(Q|Z)\} + \int_0^{Q} \nu(t) d\hat F(t|Z)$ estimates  
$\E^*\left\{\nu(T)|Q, Z\right\}$, while the inverse probability weight $ {1}/\{1-\hat F(Q|Z)\} $  corrects for the selection bias caused by left truncation \citep{cheng2012estimating, ertefaie2014propensity}. 
Motivated by $\hat\theta_{\text{Reg.T1}}$, we also consider another regression-based estimator:
\begin{align}
\hat\theta_{\text{Reg.T2}} = \left\{\frac{1}{n} \sum_{i=1}^n \frac{1}{1- \hat F(Q_i|Z_i)}\right\}^{-1} \left[ \frac{1}{n}\sum_{i=1}^n \frac{ \int_0^{\infty} \nu(t)  d\hat F(t|Z_i)}{1-\hat F(Q_i|Z_i)} \right], \label{eq:est_Reg.T2}
\end{align}
where $\int_0^{\infty} \nu(t)  d\hat F(t|Z)$ estimates $\E^*\{\nu(T)|Z\}$. 

For the rest of this section, we construct estimators with doubly robust properties from \eqref{eq:est_eq} under two scenarios.
In Section \ref{sec:est_mdr}, we consider the more classical case where both $\hat F$ and $\hat G$ are asymptotically linear. 
In Section \ref{sec:est_cf}, we consider nonparametric or machine learning methods and  estimate  $\theta$ using a cross-fitting procedure.

\subsection{Model double robustness under asymptotic linearity}\label{sec:est_mdr}

When parametric or semiparametric models like the Cox proportional hazards model are used to estimate $F$ and $G$, 
the estimators $\hat F$ and $\hat G$ are root-$n$ consistent and asymptotically linear. In this case, the observed data $\{O_i\}_{i=1}^n$ are used to obtain $\hat F$ and $\hat G$, and the same data are then used to solve equation \eqref{eq:est_eq}. We denote the resulting estimator $\hat \theta_{dr}$. 
We show that $\hat \theta_{dr}$ is model doubly robust, in  that it is consistent and asymptotically normal when one of the models for $F$ and $G$ is correctly specified, but not necessarily both.

Let us first introduce some norms. 
For a random function $X(t,z)$ with $t\in [\tau_1,\tau_2]$ and $z\in \Z$, define
$\|X(\cdot,Z)\|_{\text{sup},2}^2 =  \E\left\{\sup_{t\in[\tau_1,\tau_2]} \left|X(t,Z)\right|^2\right\}$ and $\|X(\cdot,Z)\|_{\TV,2}^2  =  \E\left\{\TV\{X(\cdot,Z)\}^2\right\}$, 
where 
$\TV\{X(\cdot,z)\} = \sup_{\mathcal{P}} \sum_{j=1}^J |X(x_{j+1},z) - X(x_j,z)|$ 
is the total variation of $X(\cdot,z)$ on the interval $[\tau_1,\tau_2]$, and $\mathcal{P}$ is the set of all possible partitions $\tau_1=x_0<x_1<...<x_J= \tau_2$ of $[\tau_1, \tau_2]$. 
We assume the following.

\begin{assumption}[Uniform Convergence]\label{assump:uniformcons1}
	There exist CDF's $F^\divideontimes$ and $G^\divideontimes$ such that
	\begin{align}
	\left\|\hat F(\cdot|Z) - F^\divideontimes(\cdot|Z)\right\|_{\sup, 2} = o(1), 
	\quad \ \left\|\hat G(\cdot|Z) - G^\divideontimes(\cdot|Z)\right\|_{\sup, 2} = o(1).
	\end{align}
\end{assumption}

\begin{assumption}[Asymptotic Linearity]\label{assump:if}
	For fixed $(t,z)\in[\tau_1,\tau_2]\times \Z$, $\hat F(t|z)$ and $\hat G(t|z)$ are regular and asymptotically linear estimators for $F(t|z)$ and $G(t|z)$ with influence functions  
 $\xi_1(t,z,O)$ and $\xi_2(t,z,O)$, respectively. In addition, denote
    \begin{align}
        R_1(t,z) & = \hat F(t|z) - F^\divideontimes(t|z) - \frac{1}{n}\sum_{i = 1}^n \xi_{1}(t, z, O_i) , \label{eq:F_hat_AL}\\
        R_2(t,z) & = \hat G(t|z) - G^\divideontimes(t|z) - \frac{1}{n}\sum_{i = 1}^n \xi_{2}(t, z, O_i). \label{eq:G_hat_AL}
    \end{align}
    Suppose $\left\|R_1(\cdot,Z)\right\|_{\sup, 2} = o(n^{-1/2})$, $\left\|R_2(\cdot,Z)\right\|_{\sup, 2} = o(n^{-1/2})$, and either $\left\|R_1(\cdot, Z)\right\|_{\TV, 2} = o(1)$ or $\left\|R_2(\cdot,Z)\right\|_{\TV, 2} = o(1)$. 
\end{assumption}

In general, the rate conditions in terms of $\|\cdot\|_{\TV,2}$ does not imply those under $\|\cdot\|_{\sup,2}$, and vice versa. 
As an example, if $\hat F$ is the nonparametric maximum likelihood estimator (NPMLE) under the  proportional hazards model, 
$\|\hat F - F^\divideontimes\|_{\sup,2}$ is usually of order $O(n^{-1/2})$, but even $o(1)$ rate  is not achievable for $\|\hat F - F^\divideontimes\|_{\TV,2}$. In fact, we can show that  for the simple case of estimating a smooth CDF 
by its empirical CDF, 
the $\|\cdot\|_{\TV,2}$ of its error  is $ O(1)$.
On the other hand, the rate conditions on the remainder terms  in Assumption \ref{assump:if} are reasonable; 
in particular, we verify in the Supplementary Material that Assumption  \ref{assump:if} is satisfied when the  proportional hazards model is used to estimate $F$ and $G$.
The rate condition 
in terms of $\|\cdot \|_{\TV,2}$ is needed to handle the involvement of time in both nuisance parameters $F$ and $G$. To the best of our knowledge, the technical complication involved in  the asymptotics of such doubly robust estimators have not been carefully studied in the previous literature.
Using integration by parts, we can show that only one of the two $\|\cdot \|_{\TV,2}$ conditions in Assumption \ref{assump:if} is needed.

In the following, we will use ``$\convp$" to indicate convergence in probability, and  ``$\convd$" to indicate convergence in distribution.
We have the following model double robustness property for $\hat\theta_{dr}$.

\begin{theorem}[Model double robustness]\label{thm:mdr}
	Under Assumptions \ref{ass:quasi-indpendent}, \ref{assump:uniformcons1}, and regularity Assumption \ref{ass:U_stat} in the Supplementary Material, assuming also that both ($F^\divideontimes$, $G^\divideontimes$) and  ($F_0$, $G_0$) satisfy Assumption \ref{ass:overlap},
 if either $F^\divideontimes = F_0$ or $G^\divideontimes = G_0$, we have: \\
	(i)  $\hat \theta_{dr} \convp \theta_0$;\\
	(ii) if in addition Assumptions \ref{assump:if} holds, then
	$\sqrt{n}(\hat \theta_{dr} - \theta_0) \convd N(0, \sigma^2)$.
	Furthermore, when both $F^\divideontimes = F_0$ and $G^\divideontimes = G_0$, $\hat \theta_{dr}$ achieves the semiparametric efficiency bound, and $\sigma^2$ can be consistently estimated by $\hat\sigma$, where 
 $$\hat\sigma^2 = \hat\beta^2 \ \frac{1}{n} \sum_{i=1}^n U_i^2( \hat\theta_{dr}, \hat F, \hat G),
 \quad \hat \beta = \left\{\frac{1}{n} \sum_{i=1}^n \frac{1}{\hat G(T_i|Z_i)} \right\}^{-1}.$$ 
\end{theorem}
 The proof of the theorem is given in the Supplementary Material, where the consistency proof utilizes concentration inequalities, and the asymptotic normality proof further utilizes asymptotic results of $U$-Statistics.

\begin{remark}
    As mentioned earlier, since $\hat\theta_{\text{IPW.Q}}$ and $\hat\theta_{\text{Reg.T1}}$ are special cases of $\hat\theta_{dr}$, they are CAN when $\hat G$ and $\hat F$ are uniformly consistent and asymptotically linear, respectively. 
    In addition, when $\hat F$ is uniformly consistent and asymptotically linear, $\hat\theta_{\text{Reg.T1}}$ and $\hat\theta_{\text{Reg.T2}}$ can be shown to be asymptotically equivalent under the slightly stronger assumption that $Q$ is independent of $T$ given $Z$ in the full data.
\end{remark}

\subsection{Rate double robustness with cross-fitting} \label{sec:est_cf}

In the contemporary data science era, analysts may wish to consider more flexible nonparametric or machine learning methods in order to estimate $F$ and $G$. These estimators are known to converge slower than the root-$n$ rate, and are not asymptotically linear. To incorporate such methods, we utilize the cross-fitting procedure that are commonly considered in the literature \citep{hasminskii1978nonparametric, bickel1982adaptive, robins2008higher, chernozhukov2018double}.  

The cross-fitting procedure in Algorithm \ref{alg:cf} introduces independence 
between the estimated nuisance parameters and the data used to estimate $\theta$, 
thereby avoiding the asymptotic linearity assumption required in Section \ref{sec:est_mdr}. 


\begin{algorithm}
\caption{Estimation of $\theta$ via $K$-fold cross-fitting.}\label{alg:cf}
\KwData{split the data into $K$ folds of roughly equal size with the index sets $\I_1,...,\I_K$.
}
\For{$k=1$ to $K$ }{
\medskip
Estimate $F$ and $G$ using the out-of-$k$-fold data indexed by $\I_{-k} = \{1,...,n\}\backslash \I_k$; \\
denote the estimates by $\hat F^{(-k)}$ and $\hat G^{(-k)}$.
}
\KwResult{Obtain $\hat \theta_{cf}$ by solving
\begin{equation}
    \sum_{k=1}^K \sum_{i\in \I_k} U\{O_i;\theta, \hat F^{(-k)}, \hat G^{(-k)}\} = 0.
\end{equation}
}
\end{algorithm}
 
We again introduce some norms here, and  show that $\hat\theta_{cf}$ enjoys a 
rate doubly robust property 
\citep{ rotnitzky2021characterization, hou2021treatment} as described below. 
Denote $\mathcal{O}$ the data used to obtain $\hat F$ and $\hat G$, 
Denote $O_\dagger = (Q_\dagger, T_\dagger,Z_\dagger)$ an copy of the data that is independent of, but from the same distribution, as $\mathcal{O}$. 
Let $\E_\dagger$ denote  expectation taken with respect to  $O_\dagger$ conditional on the original data $\mathcal{O}$, and $\E$  the expectation  with respect to $\mathcal{O}$. 
Define 
\begin{align}
\|\hat F - F_0\|_{\dagger, \sup,2}^2 & =  \E\left(\E_\dagger\left[\left\{\sup_{t\in[\tau_1,\tau_2]} \left|\hat F(t|Z_\dagger) - F_0(t|Z_\dagger)\right|\right\}^2\right]\right), \\
\|\hat G - G_0\|_{\dagger, \sup,2}^2 & =  \E\left(\E_\dagger\left[\left\{\sup_{t\in[\tau_1,\tau_2]} \left|\hat G(t|Z_\dagger) - G_0(t|Z_\dagger)\right|\right\}^2\right]\right), \\
\D_{\dagger}(\hat F,\hat G; F_0,G_0) 
& = \E\left(\E_\dagger\left[\left| \int_{\tau_1}^{\tau_2} \left\{a(v, Z_\dagger; \hat F) - a(v, Z_\dagger; F_0)\right\}\right.\right.\right.\\
&\quad\quad\quad\quad\quad\quad\quad \left.\left.\left.  \cdot\  Y_\dagger(v)\ d\left\{\frac{1}{\hat G(v|Z_\dagger)} - \frac{1}{G_0(v|Z_\dagger)} \right\}\right|\right]\right), 
\end{align}
where $Y_\dagger(t) = \mathbbm{1}(Q_\dagger\leq v<T_\dagger)$ and $a(v, Z; F) = \int_0^v\{\nu(t)-\theta\} dF(t|Z)/\{1-F(v|Z)\}$. 
We refer to $ \D_{\dagger}(\hat F,\hat G; F_0,G_0) $ as the {\it out-of-sample cross integral product}.

\begin{assumption}[Uniform Consistency]\label{assump:uniformcons2}
	$\|\hat F - F_0\|_{\dagger, \sup,2} = o(1)$ and 
	$\|\hat G - G_0\|_{\dagger, \sup,2} = o(1)$.
\end{assumption}

\begin{assumption}[Product rate condition] \label{ass:prodrate}
	$\D_{\dagger}(\hat F, \hat G; F_0, G_0) = o(n^{-1/2})$.
\end{assumption}

In the above we assume that the model classes for estimating $F$ and $G$ are large enough so that $\hat F$ and $\hat G$ are uniformly consistent for the truth. 
The rate conditions in Assumptions \ref{assump:uniformcons2} and \ref{ass:prodrate} are extensions of common rate conditions assumed for the rate doubly robust estimators in the literature \citep{rotnitzky2021characterization, hou2021treatment, rava2023doubly}, 
 again with particular attention to the handling of the fact that $F$ and $G$ are both functions of $t$. 
Previous literature \citep{chernozhukov2018double, rotnitzky2021characterization} typically imposes their product rate condition as a product of the error rates from  estimating the two nuisance parameters, instead of an integral form as in  Assumption \ref{ass:prodrate} here. This is because at most {one of} their nuisance parameters involves time $t$, which results in the `mixed bias property' studied in \citet{rotnitzky2021characterization}, and enables an easy application of the Cauchy-Schwartz inequality in the proofs. The time construct in both nuisance functions is inevitable in our study, and a more detailed discussion on this topic can be found in \citet{ying2023cautionary}.

A technical note is that we have assumed conditions on the expectations of the estimation errors and the out-of-sample cross integral product. These are slightly stronger than the conditions where we would have only the inner expectation $\E_\dagger$ and assume that the rate is $o_p(\cdot)$ \citep{rotnitzky2021characterization, hou2021treatment}. But they are equivalent if we have regularity condition that the quantities are bounded almost surely. Our assumption here simplifies the proofs.

We have the following rate double robustness property for $\hat\theta_{cf}$.

\begin{theorem}[Rate double robustness]\label{thm:rdr}
	Under Assumptions \ref{ass:quasi-indpendent} and \ref{assump:uniformcons2}, assuming also that $F_0$ and $G_0$ satisfy Assumption \ref{ass:overlap}, we have: \\
	(i) $\hat\theta_{cf} \convp \theta_0$;\\
	(ii) if in addition Assumptions \ref{ass:prodrate} holds, then 
	$n^{1/2}(\hat\theta_{cf} - \theta_0) \convd N(0, \sigma^2)$, where $\sigma^2 = \beta_0^2 \E\{U(\theta_0, F_0, G_0)^2\} = \E(\varphi^2)$. 
	Therefore, $\hat\theta_{cf}$ achieves the semiparametric efficiency bound derived in Proposition \ref{thm:EIF}.
 In addition, $\sigma^2$ can be consistently estimated by $\hat \sigma^2_{cf}$, where
 \begin{align}
     \hat\sigma_{cf}^2 = \hat\beta_{cf}^2 \ \frac{1}{n} \sum_{k=1}^K\sum_{i\in\I_k} U_i^2\{ \hat\theta_{cf}, \hat F^{(-k)}, \hat G^{(-k)}\} , \quad \hat \beta_{cf} = \left\{\frac{1}{n}\sum_{k=1}^K\sum_{i\in\I_k} \frac{1}{\hat G^{(-k)}(T_i|Z_i)} \right\}^{-1}.
 \end{align}
\end{theorem}
The proof of the Theorem is given in the Supplementary Material.




\begin{remark}
   Neyman orthogonal scores \citep{neyman1959optimal} have long been considered in the context of  slower than root-$n$ estimation for  nuisance parameters   \citep{newey1994asymptotic, bickel1993efficient}. 
    Under  regularity conditions, it can be shown that all influence curves are Neyman orthogonal scores \citep{rava2023doubly}, so is the efficient influence curve and thus the estimating function $U$ in \eqref{eq:score}. In addition, it can be shown that all doubly robust estimating functions are Neyman orthogonal scores.    With cross-fitting Neyman orthogonal scores can suggest root-$n$ consistent estimators when the nuisance parameters  are estimated at faster than $n^{-1/4}$ rate \citep{newey1994asymptotic, rotnitzky2021characterization}.
However, as pointed out by \citet{bilodeau2022blair}, \citet{ogburn2022elizabeth} and \citet{tang2022yanbo}, the $n^{-1/4}$ rate requirement can rule out a number of data adaptive machine learning methods. 
 With the help of rate double robustness, $\hat\theta_{cf}$ alleviates the $n^{-1/4}$ rate requirement; it allows  one of $\hat F$ and $\hat G$ to converge arbitrarily slow, as long as the the product error rate is faster than root-$n$.
 In practice few convergence rate results are known for nonparametric methods under left truncation, and in numerical studies later we will investigate empirically the performance of some of these methods.  
\end{remark}

\section{Estimation under right censoring}\label{sec:censoring}


\subsection{Censoring before truncation}\label{sec:c1}

Now we extend the above estimators to handle right censoring in addition to left truncation. 
Denote $C$ the censoring time,  
and let $X = \min(T,C)$ be the possibly censored event time, and $\Delta = \mathbbm{1}(T<C)$ the event indicator. 
As discussed in \citet{qian2014assumptions}, there are two scenarios, depending on whether censoring can occur before truncation. In this subsection, we consider the first scenario where censoring can happen before left truncation, that is, $\PP^*(C<Q)>0$. 
As an example, in the central nervous system (CNS) lymphoma data that  will be analyzed in the applications later, the time to death is left truncated by the time to relapse, and censoring is due to loss to follow-up which can happen before  a patient progresses to relapse.

In this first censoring scenario, subjects with $Q<X$ are included in the data, and 
$X$ plays the role of $T$ in the case without the censoring.
We observe $(Q,X,\Delta,Z)$ for each subjects. Denote $F_x$ the conditional  CDF of $X$ given $Z$ in the full data, and  $S_c(t) = \PP^*(C>t)$. 
We  assume that $C$ is independent of $(Q,T,Z)$ in the full data, and assume the positivity condition that $S_{c}(T)>\delta_c$ a.s. for some $\delta_c>0$.

Due to the above data generating mechanism, we construct an estimating function by first applying inverse probability of censoring weights to the full data estimating function $\nu(T) - \theta$, and then adjusting for left truncation. 
As such, an extended estimating function is 
\begin{align}
    U_{c1}(\theta; F_x,G, S_c) = \frac{\Delta\{\nu(X) - \theta\}}{S_c(X)G(X|Z)} - \int_0^\infty \frac{\int_0^v \Delta \{\nu(x)-\theta\}/S_c(x) dF_x(x|Z)}{ 1-F_x(v|Z) } \cdot \frac{d\tilde M_{Q}(v;G)}{G(v|Z)}, \label{eq:score_c}
\end{align}
where $\tilde M_{Q}(t;G) = \tilde N_Q(t) - \int_0^t \mathbbm{1}(Q\leq v < X) \alpha(v|Z)dv$, and $\tilde N_Q(t) = \mathbbm{1}(t\leq Q<X)$. 
As we mentioned above, $X$ plays the role of $T$ in the case without right censoring, so it is not surprising that the nuisance parameters involved  are $F_x$ and $G$. 
Again we use subscript 0 to denote the true values of the parameters. 
We show in the Supplementary Material that $U_{c1}$ is doubly robust in the sense that $\E\left\{U_{c1}(\theta_0; F_x,G, S_{0c})\right\} = 0$ if $F_x = F_{0x}$ or $G = G_0$.

For estimation of $\theta$ we need to first estimate $F_x$, $G$ and $S_c$. 
In particular, $F_x$ can be estimated using $(Q, X, Z)$, where $X$ is left truncated by $Q$. 
Since $C$ is left truncated by $Q$ and right censored by $T$,
$ S_c$ can be estimated from  $(Q, X, 1-\Delta)$ by the product-limit estimator \citep{wang1991nonparametric, qian2014assumptions, vakulenko2022nonparametric}. 
And finally, similar to before, $G$ can be estimated from data $(\tau-X, \tau-Q, Z)$ on the reverse time scale. 
This way we can obtain a model doubly robust estimator and a rate doubly robust estimator, as well as extensions of the IPW and regression-based estimators, the expressions of which are all provided in the Supplementary Material. 
For the rate doubly robust estimator, $S_c$ is estimated using the same cross-fitting data as for estimating $F_x$ and $G$.

\begin{remark}
    As discussed before the IPW estimator and the `CW' estimator are the same when there is no censoring. 
    In the Supplementary Material, the extended IPW estimator under censoring is 
    however  different from the extended `CW' estimator in \citet{vakulenko2022nonparametric} for this censoring scenario.
\end{remark}

\subsection{Censoring after truncation}\label{sec:c2}

The second censoring scenario is $\PP^*(Q<C) = 1$. 
This is the case for many studies where censoring is due to loss to follow-up after participants enter the study. 
The assumption $\PP^*(Q<C) = 1$ implies that the left truncation time and the censoring are dependent \citep{vakulenko2022nonparametric}, unless their supports do not overlap. 
Following the setup in \citet{qian2014assumptions}, we consider censoring on the residual time scale.
In particular, denote $D=C-Q$ the residual censoring time. Therefore, $C = Q+D$, $X = \min(T,Q+D)$, and $\Delta = \mathbbm{1}(T<Q+D)$.
Denote $S_D(t) = \PP(D>t)$.  
We assume that $D$ is independent of $(Q,T,Z)$  in the observed data, and  the positivity condition that $S_{D}(T-Q)>\delta_D$ a.s. for some $\delta_D>0$. 

Under this scenario, subjects are included in the data set only if $Q<T$, and subsequently the subjects are possibly right censored. Therefore, we extend the estimating function \eqref{eq:score} by first handling left truncation,  then applying inverse probability of censoring weighting:  
\begin{align}
    U_{c2}(\theta;F,G,S_D) = \frac{\Delta}{S_{D}(X-Q)} \left[\frac{\nu(X)-\theta}{G(X|Z)} - \int_0^\infty \frac{\int_0^v \{\nu(t) -\theta\} dF(t|Z) }{ 1-F(v|Z) } \cdot \frac{d\tilde M_Q(v;G)}{G(v|Z)} \right],
\end{align}
where $\tilde M_Q$ is defined in the previous subsection.
We show in the Supplementary Material that $U_{c2}$ is doubly robust in the sense that 
$\E\left\{U_{c2}(\theta_0;F,G,S_{0D})\right\} = 0$ if $F = F_0$ or $G = G_0$. 

As before, we estimate $\theta$ by first estimating $F$, $G$ and $S_D$. 
In particular, $F$ can be estimated using existing regression methods for left truncated and right censored data;
$S_D$ can be estimated by the Kaplan-Meier estimator with the right censored data $(X-Q, 1-\Delta)$ on the residual time scale, where $(C-Q)$ is right censored by $(T-Q)$.
For estimateion of $G$ 
after reversing the time scale, 
$(\tau-T)$ is only observed for uncensored subjects. 
We therefore apply inverse probability of censoring weighting following \citet{vakulenko2022nonparametric}, and  restrict  to the uncensored subjects and assign them case weights $1/\hat S_D(X-Q)$. 
This way we can obtain doubly robust estimators, as well as extensions of the IPW and regression-based estimators. The expressions of these estimators are provided in the Supplementary Material. 
For the rate doubly robust estimator, $S_D$ is estimated using the same cross-fitting data as for estimating $F$ and $G$.


\section{Simulation}\label{Sec:simulation}

In this section, we study the finite sample performance of the proposed estimators.
We consider sample size $n=1000$ for the observed data sets, and 500 data sets are simulated, which gives a margin of error of about $+/- 1.91\%$ for the coverage probability of nominal 95\% confidence intervals. 
We focus on the setting without censoring here. Simulation results for the two censoring scenarios considered in Section \ref{sec:censoring} are shown in the Supplementary Material.

\subsection{Under semiparametric models with asymptotic linearity}\label{sec:simu_mdr}

We study the finite sample performance of $\hat\theta_{dr}$ when semiparametric models are used to estimate $F$ and $G$. In particular, we consider using the \citet{cox1972regression} proportional hazards model, under which Assumption \ref{assump:if} is satisfied as we mentioned in Section \ref{sec:est_mdr}.
We compare $\hat\theta_{dr}$ with the IPW estimator and the regression-based estimators. We also consider the naive estimator that is the simple average of the $\nu(T_i)$'s from the observed data without accounting for left truncation, and the oracle full data estimator that is the average of the $\nu(T_i)$'s from the full data. 

We generate the full data $(Q,T,Z)$, where $Z = (Z_1,Z_2)$ with 
$Z_1\sim \text{Uniform}(-1,1)$ and  
$Z_2\sim \text{Bernoulli}(0.5)-0.5$.
Following that $Q$ and $T$ are generated from seven different scenarios that are summarized in Table \ref{tab:scenarios}. The full details of the seven scenarios are described in the Supplementary Material.
For the constants involved in the seven scenarios, we take $\tau = 20$, $\tau_1 = 5$, and $\tau_2 = 8$, resulting in about 55\% truncation for Scenarios 1, 2, 4 and 6, 66\% truncation for Scenario 3, 75\% truncation for Scenario 5, and 85\% truncation for Scenario 7, after excluding  subjects with $Q > T$.

Our estimand $\theta = \PP^*(T>7)$  is 0.2370 for Scenarios 1, 2 and 3, 0.2441 for Scenarios 4 and 6, and 0.0976 for Scenarios 5 and 7, each computed from a full data sample with sample size $10^7$.

The nuisance parameters $F$ and $G$ are estimated using approaches described in Section \ref{sec:estimation}.
Specifically, we fit 
$\lambda_1(t|Z_1,Z_2) = \lambda_{01}(t) e^{\beta_{11}Z_1+\beta_{21}Z_2}$
for $T$, and fit 
$\lambda_{2}(t|Z_1,Z_2) = \lambda_{02}(t) e^{\beta_{12}Z_1+\beta_{22}Z_2}$
for $(\tau-Q)$, both with risk set adjustment for left truncation.
Obviously the proportional hazards model for $F$ is correctly specified for Scenarios 1, 2 and 3, moderately misspecified for Scenarios 4 and 6, and severely misspecified for Scenarios 5 and 7; the proportional hazards model for $G$ is correctly specified for Scenarios 1, 4 and 5, moderately misspecified for Scenarios 2 and 6, and severely misspecified for Scenarios 3 and 7. 

We report the bias, \%bias  $ = \text{(bias/truth)}\times 100\%$, empirical standard deviation (SD), average of the model-based standard errors (SE), average of the bootstrapped standard errors (boot SE), and the coverage probabilities (CP) of the nominal 95\% confidence intervals constructed using the model-based SE and the boot SE, respectively. The model-based SE's for $\hat\theta_{dr}$ is from the variance estimator in Theorem \ref{thm:mdr}.
The model-based SE for $\hat\theta_{\text{IPW.Q}}$ is from the sandwich variance estimator assuming that the weights are known. Bootstraps are carried out using resampling with replacement, with 100 bootstrap replicates, which are typically adequate for boostrapped variance estimates \citep[Page 52]{efron1994introduction}.

The simulation results are are shown in Figure \ref{fig:simu_7scenarios} as well as Table \ref{tab:simu_cox_naive_full}, with full details in the table. 
As expected, the full data estimator has the smallest bias and close to perfect coverage in all scenarios. 
The estimator $\hat\theta_{dr}$ has in general small bias and attain close to nominal coverage rate with boot SE when at least one of the models for $F$ and $G$ is correctly specified. 
When both models are correctly specified, 
the average model-based SE is very close to the empirical SD and the coverage associated with the model-based SE is close to 95\%. 
When only one model is correctly specified,
the average model-based SE is smaller than SD and the corresponding coverage is lower. 
This is more apparently in Scenarios 3 and 5, where the model misspecifications are severe. 

For the other estimators, we see that $\hat \theta_{\text{IPW.Q}}$ has a large bias and poor coverage when the model for $G$ is misspecified, especially for Scenarios 3 and 7 where the model misspecification is severe; and similarly, $\hat \theta_{\text{Reg.T1}}$ and $\hat \theta_{\text{Reg.T2}}$ have large bias and poor coverage when the model for $F$ is misspecified, especially for Scenarios 5 and 7 where the model misspecification is severe. 
For all scenarios, the naive estimator substantially overestimates the parameter of interest, which coincides with the direction of left truncation bias.

\subsection{With machine learning methods}\label{sec:simu_rdr}

We study the finite sample performance of $\hat\theta_{cf}$ with machine learning methods to estimate $F$ and $G$.
In particular, we consider using the relative risk forest (`RF') for left truncated and right censored data \citep{yao2020ensemble}.
We compare it with $\hat\theta_{dr}$ that uses semiparametric models to estimate $F$ and $G$, as well as the IPW and regression-based estimators that use semiparametric models or machine learning methods to estimate $F$ or $G$, whichever is involved. 

We generate the full data $(Q,T,Z)$, where $Z = (Z_1,Z_2)$ is generated in the same way as in Section \ref{sec:simu_mdr}.
Following that $T$ and $(\tau_2-Q)$ are generated from the following  proportional hazards models, respectively:
$$\lambda_1(t|Z_1,Z_2) = \lambda_{01}(t) e^{0.3 Z_1+0.5 Z_2}, \quad \lambda_2(t|Z_1,Z_2) = \lambda_{02}(t) e^{0.3 Z_1+0.5 Z_2},$$
where the baseline hazards $\lambda_{01}(t) = 2e^{-1}(t-\tau_1)$ if $t\geq \tau_1$, and 0 otherwise; 
$\lambda_{02}(t)$ is the hazard function associated with the $\text{Uniform}(0,\tau_2)$ distribution. 
We take $\tau_1 = 1$ and $\tau_2 = 4.5$, resulting in about 29.5\% truncation after excluding  subjects with $Q > T$.

Our estimand $\theta = \PP^*(T>3) = 0.576$, computed from a simulated full data set with sample size $10^7$.
Again the nuisance parameters $F$ and $G$ are estimated using approaches described in Section \ref{sec:estimation}, with $\tau = \max_i T_i +1$ when computing the reversed-time variables $(\tau-T)$ and $(\tau-Q)$.
The models we consider are: 
\begin{itemize}
\item 			Cox1: $\lambda(t|Z_1,Z_2) = \lambda_{0}(t) \exp\{\beta_{1}Z_1+\beta_{2}Z_2\}$ ;
\item			Cox2: $\lambda(t|Z_1,Z_2) = \lambda_{0}(t) \exp\{\beta_{1}Z_1^2+\beta_{2}Z_1Z_2\}$ ;
\item 			RF: \texttt{LTRCforests::ltrcrrf} with \texttt{ntree = 100} and \texttt{mtry = 2}. 
\end{itemize}
In particular, the relative risk forest (`RF') for left truncated and right censored data is implemented in the R package `\texttt{LTRCforests}', 
and the estimates $\hat G$ and $(1-\hat F)$ from `RF' are further bounded below by 0.05 to ensure the empirical overlap condition.

We consider the doubly robust estimators $\hat\theta_{dr}$ and $\hat\theta_{cf}$ (with  $K = 10$), denoted using the following convention;  
for instance, `dr-Cox1-Cox2' corresponds to $\hat\theta_{dr}$ with $F$ estimated from model `Cox1' and $G$ estimated from model `Cox2'. 
Obviously `Cox1' is correctly specified for both $F$ and $G$, while the  `Cox2' is misspecified for both $F$ and $G$. 
We also consider the IPW estimator $\hat\theta_{\text{IPW.Q}}$, and the regression-based estimators  $\hat\theta_{\text{Reg.T1}}$ and $\hat\theta_{\text{Reg.T2}}$, denoted using similar convention as above. 
In addition, we  include  the product-limit (PL) estimator \citep{wang1991nonparametric} that assumes random left truncation, as well as the naive estimator and the oracle full data estimator described in Section \ref{sec:simu_mdr}.

As before, we report the bias, SD, average of the model-based SE, average of the boot SE, and the CP of the nominal 95\% confidence intervals constructed using the model-based SE and the boot SE, respectively. The model-based SE for $\hat\theta_{cf}$ is from the variance estimator in Theorem \ref{thm:rdr}.
The model-based SE's for $\hat\theta_{dr}$ and $\hat\theta_{\text{IPW.Q}}$ as well as the boot SE's are computed in the same way as in Section \ref{sec:simu_mdr}. 

 The simulation results are summarized in Table \ref{tab:simu}, with additional visualization in Figure \ref{fig:simu} of the Supplementary Material. 
As expected, the full data estimator has the smallest bias and close to perfect coverage of the CI's. 
The estimator $\hat\theta_{dr}$ also has small biases and attain close to nominal coverage rates with boot SE, when at least one of the two models is correctly specified. 
When both models are correctly specified, the average model-based SE is very close to the empirical SD and the coverage associated with model-based SE is close to 95\%.   
When only one model is correctly specified, the average model-based SE is smaller than SD and the corresponding coverage is lower. 
We also see that `cf-RF-RF' has a small bias and good coverage rate associated with the model-based SE.

For the other estimators, we see that  when the model for $G$ is misspecified, $\hat \theta_{\text{IPW.Q}}$ has a large bias and poor coverage; and similarly,  when the model for $F$ is misspecified, $\hat \theta_{\text{Reg.T1}}$ and $\hat \theta_{\text{Reg.T2}}$ have large biases and poor coverages. 
In addition, the above three estimators have large biases when the nuisance parameter involved is estimated by `RF', which is known to have slower than root-$n$ convergence.
The product-limit estimator  has a large bias, which is expected because it ignores the dependency between $Q$ and $T$.
The naive estimator substantially overestimates the parameter of interest, which coincides with the direction of left truncation bias.

\section{Applications}\label{sec:application} 

We apply the proposed estimators to analyze the data from a study on central nervous system (CNS) lymphoma and the data from Honolulu Asia Aging Study. The two data sets contain the two different censoring scenarios discussed earlier. 


\subsection{CNS Lymphoma Data}

We analyze the data from a study on central nervous system (CNS) lymphoma 
\citep{wang2015relapse}, which is publicly available in the Supplementary Material of \citet{vakulenko2022nonparametric}. 
The data set initially contained 172 immunocompetent patients with primary CNS diffuse large B-cell lymphoma that achieved radiographic complete response
(CR) to induction therapy. Within these patients, 98 progressed to relapse, 8 died without relapse, and 66 did not relapse by the end of the follow-up.
The event time of interest is the overall survival time measured from the initiation of treatment at the time of diagnosis. 

Following \citet{vakulenko2022nonparametric}, we restrict the data set to the 98 patients that progressed to relapse, where the overall survival time was left truncated by the time to relapse. Since patients could be lost to follow-up before progressing to relapse, the appropriate censoring model was the one described in Section \ref{sec:c1}.
Among the 98 patients that progressed to relapse, 52 died during the follow-up. The patients who did not die by the end of follow-up were also right censored.
Since it was possible that patients did not progress to relapse by the end of follow-up, the appropriate censoring model must allow censoring prior to left truncation, that is, the one described in Section \ref{sec:c1}.

As discussed in \citet{vakulenko2022nonparametric}, it is plausible that time to death and time to relapse are dependent, and treatment is strongly associated with both. We therefore follow \citet{vakulenko2022nonparametric} and include the two binary treatment variables as covariates for the analysis: only MTX-based chemotherapy was used initially (yes/no), and initial radiation therapy was used (yes/no). 
We examine the overlap assumption for this data set in the Supplementary Material 
and show that it is satisfied empirically. 

We consider estimating the overall survival curve for 
$t \geq$ 8.8 months, the minimum event time in the data set. 
We consider the estimators $\hat\theta_{dr}$, $\hat\theta_{\text{IPW.Q}}$, $\hat\theta_{\text{Reg.T1}}$ and $\hat\theta_{\text{Reg.T2}}$, with the nuisance parameters $F_x$ and $G$  estimated from the  proportional hazards model. We also consider $\hat\theta_{cf}$ with 10-fold cross-fitting, and  $F_x$ and $G$ estimated from the relative risk forest with \texttt{mtry = 2} and \texttt{ntree = 2000}. The hyperparameter \texttt{ntree} is taken to be much larger than that in the simulation in order to obtain more stable estimates with the small sample size. 
As  comparisons, we consider the PL estimator 
that assumes random left truncation, as well as the naive Kaplan-Meier estimator that ignores left truncation. 
In addition, we also compute the `CW' estimator from \citet{vakulenko2022nonparametric}. 

Figure \ref{fig:CNS} shows the estimated survival curves from different estimators.  The corresponding 95\% bootstrap confidence intervals are shown in the Supplementary Material. 
Due to the small sample size, in certain bootstrap runs $\hat\theta_{dr}$ and $\hat\theta_{cf}$ can be greater than one  at early 
times up to about 50 months and we force them to be one in those cases. 
In addition, the estimates and their 95\% confidence intervals at  36, 60, 120, and 180 months are also tabulated in the Supplementary Material.

From Figure \ref{fig:CNS}, we see that the regression-based estimates and $\hat\theta_{dr}$ are very close to each other, suggesting that the Cox model might be a reasonable fit to $F_x$.
Test of proportionality in the Supplementary Material using cumulative martingale residuals \citep{lin1993checking}  with risk set adjustment for left truncation detected no violations, although admittedly the test is probably under powered with such a small sample size. 
The PL estimate is visibly lower than the other estimates, consistent with the violation of random truncation assumption. The KM estimate is substantially larger than these other estimates, reflecting the known left truncation bias.
The two IPW estimators, $\hat\theta_{\text{IPW.Q}}$ and the `CW' estimator, are seen to be quite different for this data set.


\subsection{HAAS data}

We analyze the data collected between 1965 and 2012 from the Honolulu Heart Program (HHP, 1965-1990) and the subsequent Honolulu Asia Aging Study (HAAS, 1991-2012) \citep{p2012honolulu, zhang2022marginal}.
Age is the time scale of interest, and we are interested in the cognitive impairment-free survival, that is, the age to moderate cognitive impairment or death, whichever happened first.
Only subjects that were alive and did not have cognitive impairment at the start of HAAS were included in the data set, therefore their ages for impairment-free survival were left truncated by their ages at the start of HAAS. Among the 2560 subjects in the data set, 463 (18.1\%) were right censored due to loss to follow-up. Since loss to follow-up may only happen after subjects entered HAAS, this is censoring scenario 2 described in Section \ref{sec:c2}. 

We detect dependency between age for impairment-free survival and age when entering HAAS using the conditional Kendall's tau test \citep{tsai1990testing, martin2005testing}, with tau = 0.0426 and $p$-value of 0.0014. Therefore, we consider the risk factors collected during HHP that were thought to affect the outcomes including years of education, APOE (a `gene'  related to Alzheimer’s disease) positive (yes/no), mid-life alcohol consumption (heavy/not heavy), 
and mid-life cigarette consumption (packs per year). 
We examine the overlap assumption for this data set in the Supplementary Material and show that it is satisfied empirically. 

We focus on estimating the cognitive impairment-free survival 
for $t\geq$ 72.6 years old, the minimum event time in the data set. 
We consider the estimators $\hat\theta_{dr}$, $\hat\theta_{\text{IPW.Q}}$, $\hat\theta_{\text{Reg.T1}}$ and $\hat\theta_{\text{Reg.T2}}$, with the nuisance parameters $F$ and $G$  estimated from the  proportional hazards model with the four covariates above. We do not implement $\hat\theta_{cf}$ because to the best of our knowledge, there is no existing R package of random forests for left truncated data that can incorporates case weights, which is needed in order to estimate $G$  as discussed in Section \ref{sec:c2}. As  comparisons, we consider the PL estimator 
that assumes random left truncation, as well as the naive Kaplan-Meier estimator that ignores left truncation.

Figure \ref{fig:HAAS} shows the estimated survival curves on the age interval 84 to 96 from different estimators and the corresponding 95\% bootstrap confidence intervals. The entire survival curves are shown in the Supplementary Material. In addition, the estimates and their 95\% confidence intervals at ages 80, 85, 90, and 95 years old are also tabulated in the Supplementary Material.

From Figure \ref{fig:HAAS}, we see that $\hat\theta_{dr}$, $\hat\theta_{\text{IPW.Q}}$ and $\hat\theta_{\text{Reg.T1}}$ are close to each other. 
We also see that the $\hat\theta_{\text{Reg.T2}}$ is visibly larger than $\hat\theta_{\text{Reg.T1}}$, especially for the later ages. 
Test using cumulative martingale residuals (see Supplementary Material)  detected violation of proportional hazards assumption for education and cigarette consumption in the model for $F$  at significance level 0.05, and education in the model for $G$. 
Model misspecification in this case appears to affect $\hat\theta_{\text{Reg.T2}}$ more than $\hat\theta_{\text{Reg.T1}}$; the difference between the two lies in the estimation of $\E^*\{\nu(T)|Z\}$ in the numerators, 
by $\int_0^{Q_i} \nu(t)  d\hat F(t|Z_i) + \nu(T_i)\{1-\hat F(Q_i|Z_i)\} $ for $\hat\theta_{\text{Reg.T1}}$,  
and $\int_0^\infty\nu(t) d\hat F(t|Z)$ for $\hat\theta_{\text{Reg.T2}}$. 
The latter seems to rely more on the correct $F$ model. 

Finally the PL estimate is substantialy larger, most likely due to violation of the random truncation assumption. The KM estimate is obviously larger than all the other estimates, reflecting the known left truncation bias.

\section{Discussion}\label{sec:discussion} 

It has been observed in the literature \citep{van2003unified, tsiatis2006semiparametric, dukes2019doubly}  that influence curves can yield doubly robust estimating functions. However, the derivation of influence curves and  construction of the corresponding estimators can be non-trivial.
In this paper, we have derived the EIC for the mean of an arbitrarily transformed survival time in the presence of covariate-induced dependent left truncation, and constructed model doubly robust and rate doubly robust estimators from an EIC-based doubly robust estimating function.
The rate doubly robust estimator is constructed via cross-fitting, while the model doubly robust estimator does not require cross-fitting and is therefore more computationally efficient. 

\citet{rava2023doubly} 
and \citet{luo2022doubly} have also considered constructing doubly robust estimators  with one of the nuisance parameters estimated from a parametric or semiparametric model and the other from a nonparametric model. This situation is not covered in this paper. However, estimators constructed this way are not guaranteed to be root-$n$ consistent unless the the parametric or semiparametric model is correctly specified. As such, these estimators appear to have few advantages compared with the IPW or regression-based estimators. 

The extension to right censoring assumes that the (residual) censoring variable is independent of the other variables, which may not be true in practice. A weaker assumption is that the censoring time and the left truncation and event times are independent conditional on measured baseline covariates. In the latter case inverse probability of censoring weighting may still be leveraged; the resulting estimator, however, relies on the correct specification of a censoring model given the covariates. It would be of interest to explore robust estimators that also have protection against censoring model misspecification, although this is beyond the scope of the current work.

Finally, the quasi-independence assumption requires the dependency between the left truncation time and the event time to be completely captured by measured baseline covariates, which may be unrealistic 
in practice. Alternative approaches include the ``proximal identification framework'' \citep{tchetgen2020introduction}, which leverages domain knowledge to avoid the ``conditional independence'' structure imposed by Assumption \ref{ass:quasi-indpendent}. This can be a promising approach and has been considered in a survival setting with right-censoring \citep{ying2022proximal2}.

The code for the simulation and data analysis involved in this paper is available at \\
\href{https://github.com/wangyuyao98/left_trunc_DR}{https://github.com/wangyuyao98/left\_trunc\_DR}.

\section*{Acknowledgement}

This research was partially supported by  NIH/NIA grant R03 AG062432.

\bibliographystyle{chicago}
\bibliography{left_trunc}


\newpage

\begin{table}[H]
    \centering
	\renewcommand{\arraystretch}{0.6}
	\caption{Summary of the data generating mechanisms for Scenarios 1-7.}\label{tab:scenarios}
 \resizebox{\textwidth}{!}{%
\begin{tabular}{ll}
\toprule
Scenarios & Data-generating mechanism\\ 
\midrule
1. Both models correct &  Cox: $T \sim 0.3 Z_1+0.5 Z_2$. \\
 & Cox: $(\tau-Q)\sim 0.3 Z_1+0.5 Z_2$. \\
 \midrule
2. $Q$ model moderately wrong & Cox: $T\sim 0.3 Z_1+0.5 Z_2$. \\
  & Cox: $(\tau-Q) \sim 0.3Z_1+0.5Z_2 +0.6(Z_1^2-1/3)+0.5Z_1Z_2$. \\
\midrule
3. $Q$ model severely wrong  & Cox: $T \sim 0.3 Z_1+0.5 Z_2$. \\
 & mixture: Cox: $(\tau-Q)\sim 0.3Z_1+0.5Z_2 +0.6(Z_1^2-1/3)+0.5Z_1Z_2$,\\
& \quad\quad\quad\quad\quad\quad\quad\quad\quad\quad  if $0.3Z_1+0.5Z_2 <0$.\\
& $\quad\quad\quad \ \  $ AFT: $(\tau_2-Q) \sim 0.3Z_1+0.5Z_2 +0.6(Z_1^2-1/3)+0.5Z_1Z_2$,\\
& \quad\quad\quad\quad\quad\quad\quad\quad\quad\quad  if $0.3Z_1+0.5Z_2 \geq 0$.\\
\midrule
4. $T$ model moderately wrong & Cox: $T\sim 0.3Z_1+0.5Z_2 +0.6(Z_1^2-1/3)+0.5Z_1Z_2$. \\
 & Cox: $(\tau-Q) \sim 0.3 Z_1+0.5 Z_2$. \\
\midrule
5. $T$ model severely wrong & mixture: AFT: $(T-\tau_1)\sim 0.3Z_1+0.5Z_2 +0.6(Z_1^2-1/3)+0.5Z_1Z_2$,\\
 & \quad\quad\quad\quad\quad\quad\quad\quad\quad\quad  if $0.3Z_1+0.5Z_2 <0$.\\
& $\quad\quad\quad \ \ \  $ Cox: $T \sim 0.3Z_1+0.5Z_2 +0.6(Z_1^2-1/3)+0.5Z_1Z_2$,\\
& \quad\quad\quad\quad\quad\quad\quad\quad\quad\quad  if $0.3Z_1+0.5Z_2 \geq 0$.\\
& Cox: $(\tau-Q) \sim 0.3 Z_1+0.5 Z_2$. \\
\midrule
6. both models moderately wrong & Cox: $T\sim 0.3Z_1+0.5Z_2 +0.6(Z_1^2-1/3)+0.5Z_1Z_2$ \\
 & Cox: $Q\sim 0.3Z_1+0.5Z_2 +0.6(Z_1^2-1/3)+0.5Z_1Z_2$ \\
\midrule
7. both models severely wrong & mixture: AFT: $(T-\tau_1)\sim 0.3Z_1+0.5Z_2 +0.6(Z_1^2-1/3)+0.5Z_1Z_2$,\\
 & \quad\quad\quad\quad\quad\quad\quad\quad\quad\quad  if $0.3Z_1+0.5Z_2 <0$.\\
 & $\quad\quad\quad \ \ \  $ Cox: $T \sim 0.3Z_1+0.5Z_2 +0.6(Z_1^2-1/3)+0.5Z_1Z_2$,\\
& \quad\quad\quad\quad\quad\quad\quad\quad\quad\quad  if $0.3Z_1+0.5Z_2 \geq 0$.\\
& mixture: Cox: $(\tau-Q)\sim 0.3Z_1+0.5Z_2 +0.6(Z_1^2-1/3)+0.5Z_1Z_2$,\\
& \quad\quad\quad\quad\quad\quad\quad\quad\quad\quad  if $0.3Z_1+0.5Z_2 <0$.\\
& $\quad\quad\quad \ \  $ AFT: $(\tau_2-Q) \sim 0.3Z_1+0.5Z_2 +0.6(Z_1^2-1/3)+0.5Z_1Z_2$,\\
& \quad\quad\quad\quad\quad\quad\quad\quad\quad\quad  if $0.3Z_1+0.5Z_2 \geq 0$.\\
\bottomrule
\end{tabular}}
\end{table}

\newpage

\begin{table}[H]
	\centering
	\renewcommand{\arraystretch}{0.6}
	\caption{Simulation results for Scenarios 1-7; SD: standard deviation, SE: standard error, CP: coverage probability. }
 \label{tab:simu_cox_naive_full}
 \resizebox{0.8\textwidth}{!}{%
		\begin{tabular}{llrrcrr}
  \toprule
    Scenarios & Estimator & bias  & \%bias & SD & SE/boot SE & CP/boot CP \\ 
        \midrule
    1. & dr & -0.0015 & -0.6 & 0.021 & 0.020/0.020 & 0.946/0.938 \\ 
   Both models & IPW.Q & -0.0004 & -0.2 & 0.020 & 0.015/0.019 & 0.864/0.946 \\ 
   correct & Reg.T1 & -0.0004 & -0.2 & 0.019 & -  /0.019 & -  /0.942 \\ 
   & Reg.T2 & -0.0005 & -0.2 & 0.019 & -  /0.019 & -  /0.942 \\ 
   & naive & 0.2414 & 101.9 & 0.016 & 0.016/0.016 & 0.000/0.000 \\ 
   & full & -0.0005 & -0.2 & 0.009 & 0.009/0.009 & 0.954/0.946 \\ 
   \midrule
   2. & dr & -0.0017 & -0.7 & 0.020 & 0.018/0.019 & 0.930/0.948 \\ 
   $G$ model & IPW.Q & 0.0106 & 4.5 & 0.017 & 0.014/0.017 & 0.802/0.894 \\ 
   moderately & Reg.T1 & -0.0007 & -0.3 & 0.019 & -  /0.019 & -  /0.944 \\ 
   wrong & Reg.T2 & -0.0008 & -0.3 & 0.019 & -  /0.019 & -  /0.946 \\ 
   & naive & 0.2355 & 99.4 & 0.016 & 0.016/0.016 & 0.000/0.000 \\ 
   & full & -0.0006 & -0.2 & 0.009 & 0.009/0.009 & 0.950/0.954 \\ 
   \midrule
   3. & dr & -0.0027 & -1.1 & 0.023 & 0.018/0.024 & 0.896/0.950 \\ 
   $G$ model & IPW.Q & 0.0528 & 22.3 & 0.018 & 0.015/0.018 & 0.064/0.176 \\ 
   severely & Reg.T1 & -0.0019 & -0.8 & 0.023 & -  /0.024 & -  /0.948 \\ 
   wrong & Reg.T2 & -0.0020 & -0.8 & 0.023 & -  /0.024 & -  /0.950 \\ 
   & naive & 0.2877 & 121.4 & 0.016 & 0.016/0.016 & 0.000/0.000 \\ 
   & full & -0.0001 & -0.1 & 0.008 & 0.008/0.008 & 0.956/0.946 \\ 
   \midrule
   4. & dr & -0.0015 & -0.6 & 0.021 & 0.019/0.020 & 0.930/0.948 \\ 
   $F$ model & IPW.Q & -0.0007 & -0.3 & 0.021 & 0.016/0.020 & 0.872/0.944 \\ 
   moderately & Reg.T1 & 0.0113 & 4.6 & 0.018 & -  /0.017 & -  /0.894 \\ 
   wrong & Reg.T2 & 0.0110 & 4.5 & 0.018 & -  /0.017 & -  /0.894 \\ 
   & naive & 0.2446 & 100.2 & 0.016 & 0.016/0.016 & 0.000/0.000 \\ 
   & full & -0.0006 & -0.2 & 0.009 & 0.009/0.009 & 0.948/0.944 \\ 
   \midrule
   5. & dr & -0.0020 & -2.0 & 0.011 & 0.015/0.012 & 0.988/0.964 \\ 
   $F$ model & IPW.Q & -0.0009 & -0.9 & 0.009 & 0.007/0.010 & 0.840/0.956 \\ 
   severely & Reg.T1 & -0.0288 & -29.5 & 0.008 & -  /0.008 & -  /0.070 \\ 
   wrong & Reg.T2 & -0.0293 & -30.0 & 0.007 & -  /0.008 & -  /0.058 \\ 
   & naive & 0.2329 & 238.7 & 0.014 & 0.015/0.015 & 0.000/0.000 \\ 
   & full & -0.0001 & -0.1 & 0.005 & 0.005/0.005 & 0.956/0.956 \\ 
   \midrule
   6. & dr & 0.0230 & 9.4 & 0.019 & 0.017/0.018 & 0.720/0.734 \\ 
    Both models & IPW.Q & 0.0263 & 10.8 & 0.018 & 0.014/0.017 & 0.546/0.662 \\ 
   moderately & Reg.T1 & 0.0262 & 10.7 & 0.018 & -  /0.017 & -  /0.672 \\ 
   wrong & Reg.T2 & 0.0260 & 10.6 & 0.018 & -  /0.017 & -  /0.678 \\ 
   & naive & 0.2433 & 99.7 & 0.016 & 0.016/0.016 & 0.000/0.000 \\ 
   & full & -0.0005 & -0.2 & 0.009 & 0.009/0.009 & 0.950/0.952 \\ 
   \midrule
   7. & dr & -0.0570 & -58.4 & 0.006 & 0.006/0.007 & 0.000/0.004 \\ 
   Both models & IPW.Q & -0.0635 & -65.0 & 0.006 & 0.003/0.006 & 0.000/0.000 \\ 
   severely & Reg.T1 & -0.0589 & -60.4 & 0.006 & -  /0.006 & -  /0.000 \\ 
   wrong & Reg.T2 & -0.0590 & -60.5 & 0.006 & -  /0.006 & -  /0.000 \\ 
   & naive & 0.2217 & 227.2 & 0.015 & 0.015/0.015 & 0.000/0.000 \\ 
   & full & 0.0000 & 0.0 & 0.004 & 0.004/0.004 & 0.944/0.946 \\
   \bottomrule
    \end{tabular}}
\end{table}

\newpage

\begin{table}[H]
\centering
	\renewcommand{\arraystretch}{0.6}
	\caption{Simulation results for  different estimators; the models  in red are misspecified. True $\theta  =  P^*(T>3) =0.576$. SD: standard deviation, SE: standard error, CP: coverage probability.} \label{tab:simu}
		\begin{tabular}{lrcrr}
        \toprule
			Estimator & bias  & SD & SE/boot SE & CP/boot CP \\ 
        \midrule
			dr-Cox1-Cox1 & -0.0016 & 0.021 & 0.020/0.020 & 0.948/0.946 \\ 
			dr-Cox1-{\color{red}Cox2} & -0.0014 & 0.020 & 0.019/0.020 & 0.930/0.944 \\ 
			dr-{\color{red}Cox2}-Cox1 & -0.0010 & 0.020 & 0.019/0.020 & 0.938/0.946 \\
			dr-{\color{red}Cox2}-{\color{red}Cox2} & 0.0184 & 0.019 & 0.018/0.019 & 0.838/0.836 \\ 
			cf-RF-RF & 0.0032 & 0.021 & 0.023/0.025 & 0.966/0.976 \\
			\midrule
			IPW.Q-Cox1 & -0.0004 & 0.020 & 0.018/0.020 & 0.924/0.944 \\ 
			IPW.Q-{\color{red}Cox2} & 0.0184 & 0.018 & 0.017/0.019 & 0.814/0.832 \\
			IPW.Q-RF & -0.0064 & 0.022 & 0.019/0.022 & 0.886/0.956 \\
			\midrule
			Reg.T1-Cox1 & -0.0008 & 0.020 & - /0.020 & - /0.944 \\ 
			Reg.T1-{\color{red}Cox2} & 0.0183 & 0.018 & - /0.019 & - /0.842 \\ 
			Reg.T1-RF & -0.0073 & 0.022 & - /0.022 & - /0.934 \\ 
			\midrule
			Reg.T2-Cox1 & -0.0010 & 0.020 & - /0.020 & - /0.942 \\ 
			Reg.T2-{\color{red}Cox2} & 0.0181 & 0.018 & - /0.019 & - /0.844 \\ 
			Reg.T2-RF & -0.0070 & 0.022 & - /0.022 & - /0.940 \\
			\midrule
			PL & 0.0193 & 0.018 & - /0.018 & - /0.824 \\ 
			naive & 0.1389 & 0.014 & 0.014/0.014 & 0.000/0.000 \\ 
			full data & -0.0007 & 0.013 & 0.013/0.013 & 0.956/0.944 \\ 
   \bottomrule
	\end{tabular}
\end{table}


\newpage

\begin{sidewaysfigure}[h]
	\centering
	\includegraphics[width=1\textwidth]{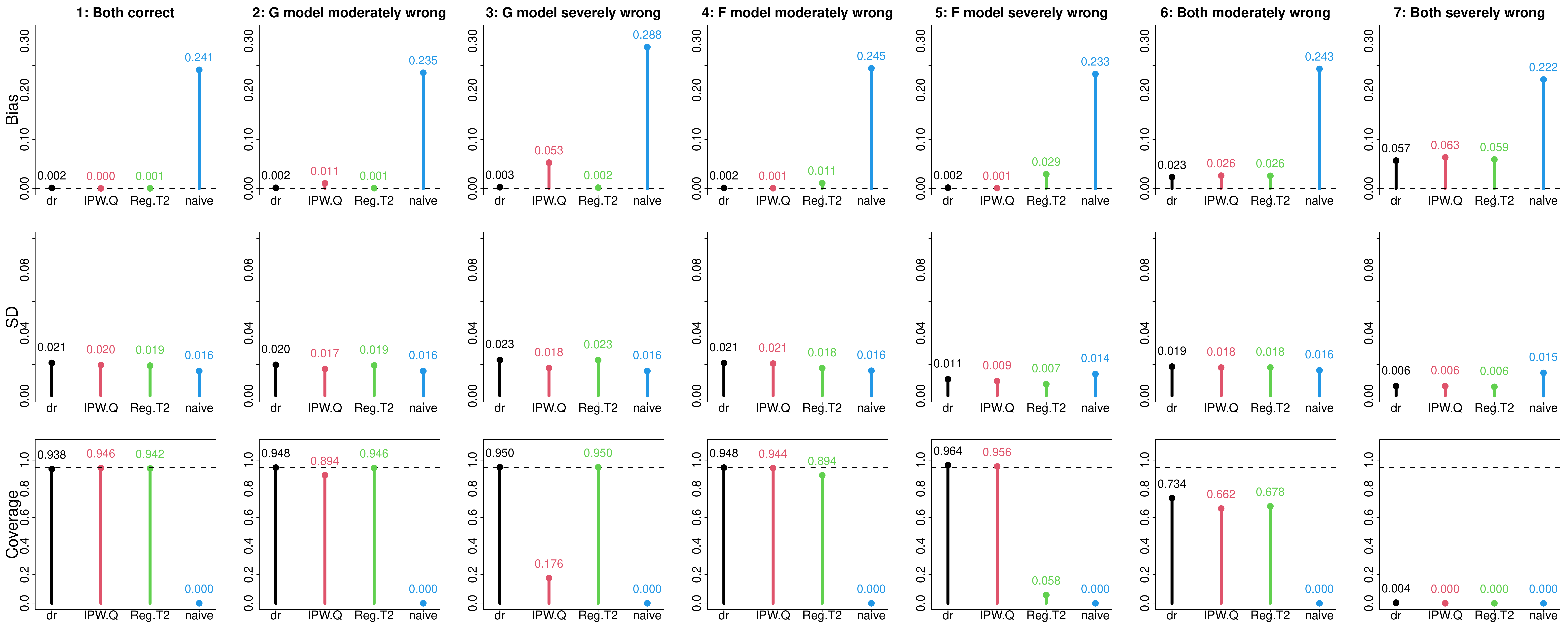}
	\caption{Absolute bias, empirical standard deviation, and coverage probability of the 95\% bootstrap CIs for different estimators, corresponding to the simulation for Scenarios 1 - 7 in Section \ref{sec:simu_mdr}.}
	\label{fig:simu_7scenarios}
\end{sidewaysfigure}

\clearpage

\begin{figure}[H]
	\centering
	\includegraphics[width=1\textwidth]{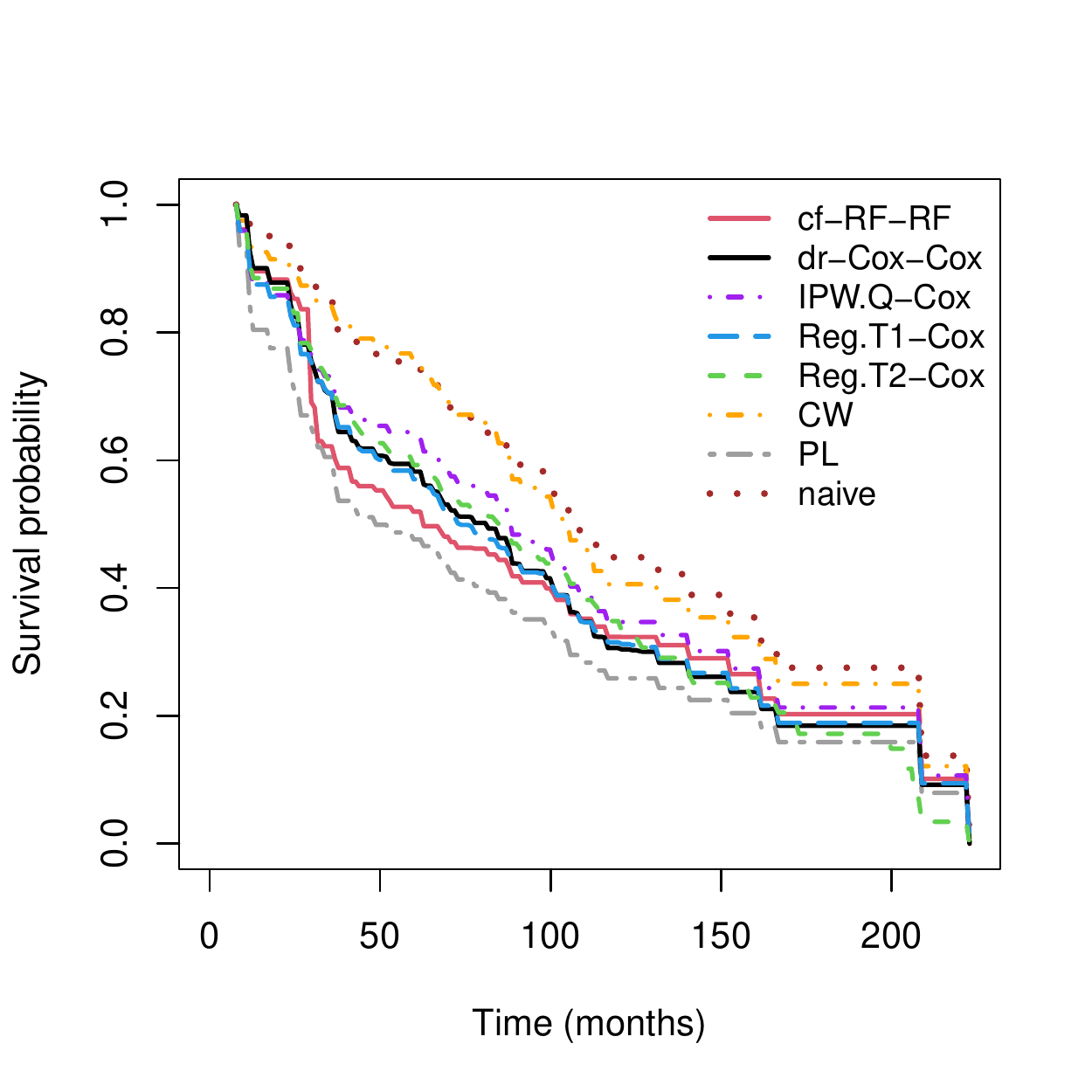}
	\caption{Estimates of the overall survival curve for the CNS lymphoma data.}
	\label{fig:CNS}
\end{figure}

\newpage

\begin{figure}[H]
	\centering
	\includegraphics[width=01\textwidth]{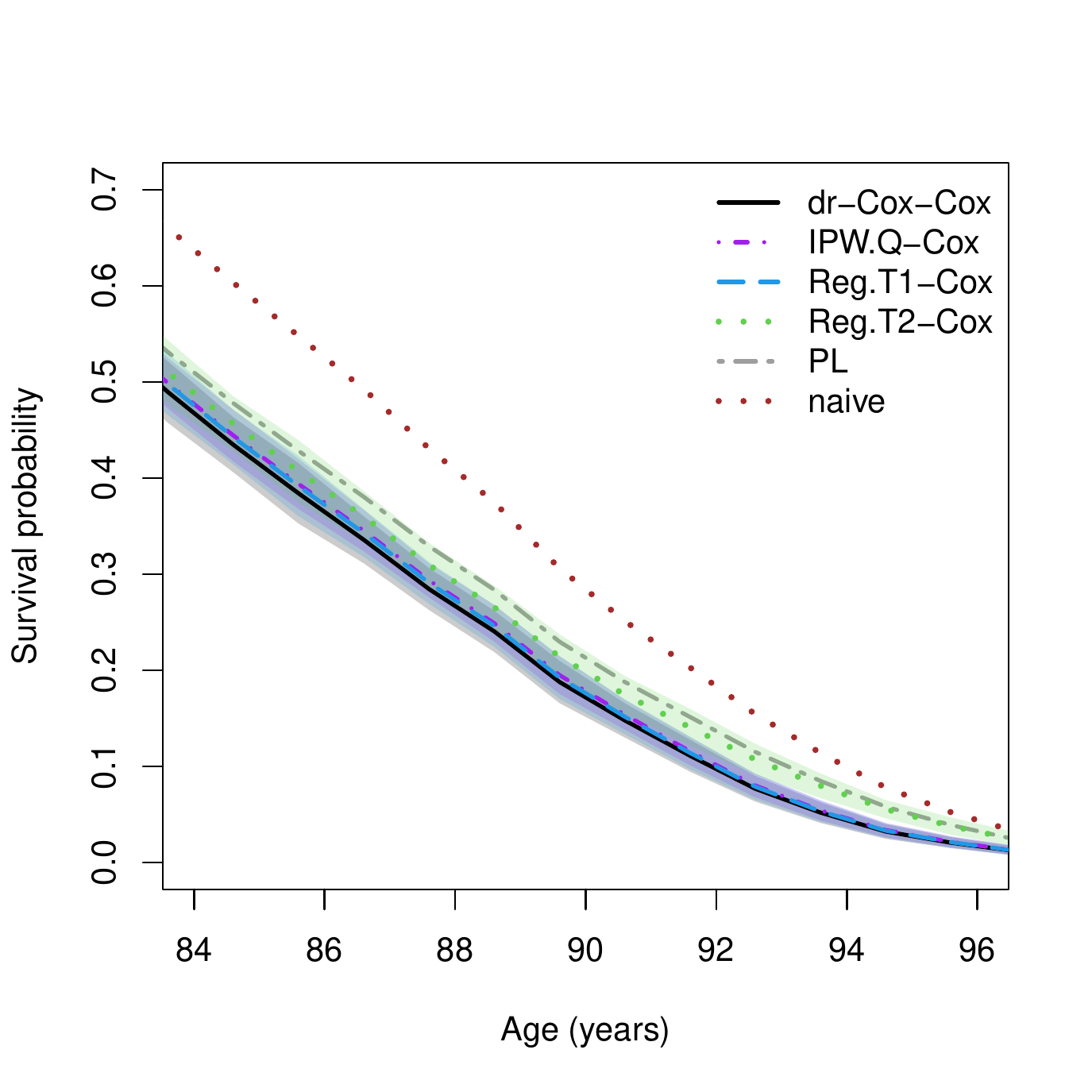}
	\caption{
 Estimated cognitive impairment-free survival between 84 and 96 years of age and their 95\% bootstrap confidence intervals (shaded) for the HAAS data.}
	\label{fig:HAAS}
\end{figure}

\newpage
\appendix
\section*{Appendix. Regression-based estimator}

Estimator \eqref{eq:est_Reg.T1} can be seen as  regression-based  because 
$\E^*\left\{\nu(T)|Q,Z \right\} = \E\left\{\nu(T)|Q,Z\right\} \{1-F(Q|Z)\} + \int_0^{Q} \nu(t) dF(t|Z)$.
This way $\nu(T)\{1- \hat F(Q|Z)\} + \int_0^{Q} \nu(t) d\hat F(t|Z)$ estimates the full data conditional expectation $\E^*\left\{\nu(T)|Q, Z\right\}$.

\newpage

{\centering\section*{Supplementary Material}}\label{SM}
The Supplementary Material here contains the proofs, details of the  estimators under right censoring, additional simulation results, and additional results for the analysis of the two  data sets. 
In particular, this supplement contains the following materials:
{\renewcommand{\arraystretch}{0.6}
\begin{enumerate}
    \item Preliminaries
    \item Martingale theory
    \item Identification of $\theta$
    \item Derivation of the EIC
    \item Double robustness of the estimating function
    \item Asymptotics
    \item Estimation under right censoring
    \item Additional materials for simulations
    \item Applications
    \item Equivalence with the IF in Chao (1987)
\end{enumerate}
}

\newpage
\section{Preliminaries}\label{app:quantities}

We will use $a \wedge b$ to denote the minimum of $a$ and $b$,  $a \vee b$ to denote the maximum of $a$ and $b$, and  `$\lesssim$' to denote less than or equal to up to a constant factor. 

If we do not specify the limit for integrals, the integrals are taken on the support of the corresponding variables.

Here we define the quantities used in the Supplementary Material. 

Let 
\begin{align}
    \beta(z) = P^*(Q<T|Z=z) =  \int_0^\infty \int_0^\infty \mathbbm{1}(q<t) f(t|z) g(q|z)\ dt\ dq. \label{eq:beta(z)}
\end{align}  
Based on the joint density for $(Q,T,Z)$ in the observed data given in Assumption \ref{ass:quasi-indpendent}, we can derive the following marginal  and conditional densities for observed data distribution in terms of $f,g$ and $h$:
\begin{align}
p_{Z}(z) &=  \frac{\beta(z) }{\beta}h(z), \\
p_{T,Z}(t,z) &=  \frac{G(t|z)}{\beta} f(t|z)h(z),\\
p_{Q,Z}(q,z) &= \frac{1-F(q|z)}{\beta} g(q|z)h(z), \label{eq:p(q,z)} \\
p_{T|Z}(t|z) & = f(t|z)  \frac{G(t|z)}{\beta(z)}, \label{eq:p(t|z)}\\
p_{Q|Z}(q|z) & = \frac{1-F(q|z)}{\beta(z)} g(q|z) , \label{eq:p(q|z)}\\
p_{Q,T|Z}(q,t|z) & = \frac{\mathbbm{1}(q<t)}{\beta(z)} f(t|z)g(q|z) , \label{eq:p(q,t|z)}\\
p_{Q|T,Z}(q|t,z) & =  \frac{\mathbbm{1}(q<t)}{G(t|z)} g(q|z),\label{eq:p(q|t,z)}\\
p_{T|Q,Z}(t|q,z) & = \frac{\mathbbm{1}(q<t)}{1-F(q|z)} f(t|z).  \label{eq:p(t|q,z)}
\end{align}

In addition, denote
\begin{align}
m(v,Z;F) &= \int_0^v \nu(t) dF(t|Z), \label{eq:m}\\
\mu(Z;F) &= \int_0^\infty \nu(t) dF(t|Z) =  m(\infty,Z;F), \label{eq:mu}\\
\tilde m(v,Z;F) &= \int_t^\infty \nu(t) dF(t|Z). \label{eq:m_tilde}
\end{align}

\newpage
\section{Martingale theory} \label{sec:martingale}
Recall that
\begin{align}
&\mathcal{\bar F}_t 
= \sigma \left\{Z, \mathbbm{1}(Q<T), \mathbbm{1}(s\leq T),\mathbbm{1}(s\leq Q<T): s\geq t  \right\},\\
&\bar M_Q(t; G) = \bar N_Q(t) - \bar A_Q(t;G),
\end{align}
where
\begin{align}
\bar N_Q(t) &= \mathbbm{1}(t\leq Q<T)\\
\bar A_Q(t;G) 
&= \int_t^\infty \mathbbm{1}(Q\leq s <T) \frac{dG(s|Z)}{G(s|Z)}.
\end{align}

\begin{lemma}\label{lem:M_Q}
	Under Assumption \ref{ass:quasi-indpendent} and \ref{ass:beta*(z)>0}, 
	$\{\bar M_Q(t; G)\}_{t\geq 0}$ is a backwards martingale with respect to $\{\mathcal{\bar F}_t\}_{t\geq 0}$ if $G$ is the true full data CDF of $Q$ given $Z$.
\end{lemma}

\begin{proof}
	In the following proof, we will first show that $\bar M_Q(t; G)$ is $\mathcal{\bar F}_t$-measurable and $\ee{|\bar M_Q(t; G)|}<\infty$ for all $t\geq 0$. Then we will show that $\ee{\bar M_Q(s; G)|\mathcal{\bar F}_t} = \bar M_Q(t;G)$ for all $0<s\leq t$ if $G$ is the true CDF of $Q$ given $Z$ in the full data.
	
	1) We first show that the at risk indicator $\mathbbm{1}(Q\leq t < T)$ is $\mathcal{\bar F}_{t+}$-measurable, where $\mathcal{\bar F}_{t+} = \bigcup_{s>t} \mathcal{\bar F}_s$. 
	For any $t\geq 0$,
	since
	\begin{align}
	\{T>t\} &= \cup_{n=1}^\infty \left\{T\geq t+\frac{1}{n}\right\}\in \mathcal{\bar F}_{t+}, \\
	\{Q<T, Q\leq t\} &= \cup_{n=1}^\infty \left\{Q<T, \ Q < t+\frac{1}{n}\right\}  \\
	&= \{Q<T\} \cap \left[\cup_{n=1}^\infty \left\{t+\frac{1}{n} \leq Q<T\right\}^c\right]
	\in \mathcal{\bar F}_{t+},
	\end{align}
	we have 
	\begin{align}
	\{Q \leq t< T\} = \{T> t\} \cap \{Q<T, Q\leq t\} \in \mathcal{\bar F}_{t+} \subseteq \mathcal{\bar F}_{t},
	\end{align}
	so $\bar A_Q(t;G)$ is $\mathcal{\bar F}_{t+}$-measurable, and $\bar M_Q(t;G)$ is $\mathcal{\bar F}_{t}$-measurable. 
	
	\bigskip
	2) We now show that for any $t\geq 0$, $\E\left\{|\bar M_Q(t; G)|\right\}<\infty$. It suffices to show that $\E\left\{\left.|\bar M_Q(t; G)|\right.|Z\right\}$ is bounded. By \eqref{eq:p(q,t|z)},
	\begin{align}
	\E\left\{\left.|\bar M_Q(t; G)|\right|Z\right\}
	& \leq \E\left\{\left.\bar N_Q(t)\right|Z\right\} + \E\left\{\left.\bar A_Q(t; G_0)\right|Z\right\} \\
	& \leq 1 + \E\left\{\left. \int_0^\tau \mathbbm{1}(Q\leq s <T) \alpha(s|Z) ds \right|Z\right\}\\
	& = 1+ \int_0^\infty \int_0^\infty \int_0^\infty \mathbbm{1}(q\leq s <t) \frac{dG(s|Z)}{G(s|Z)}  \frac{1}{\beta(Z)} dF(t|Z) dG(q|Z) \\
	& = 1+\frac{1}{\beta(Z)} \int_0^\infty \mathbbm{1}(s<t) dG(s|Z) dF(t|Z) \\
	& = 1+\frac{1}{\beta(Z)} \beta(Z) \\
	&= 2.
	\end{align}
	Therefore, 
	\begin{align}
	\E\left\{|\bar M_Q(t; G)|\right\} = \ee{\E\left\{\left.|\bar M_Q(t; G)|\right|Z\right\}} \leq 2 <\infty.
	\end{align}
	
	\bigskip
	3) Suppose that $G$ is the true CDF of $Q$ given $Z$ in the full data.
	We will show that for any $0<s\leq t$, $\E\left\{\bar N_Q(s) - \bar N_Q(t) |\mathcal{\bar F}_t\right\} = \E\left\{\bar A_Q(s;G) - \bar A_Q(t;G) | \mathcal{\bar F}_t\right\}$, which implies that $\E\left\{\bar M_Q(s;G)|\mathcal{\bar F}_t \right\} = \bar M_Q(t;G)$.
	Since the reversed time filtration 
	\begin{align}
	\mathcal{\bar F}_t 
	& = \sigma \left\{\mathbbm{1}(Q<T), \mathbbm{1}(s\leq Q<T), \mathbbm{1}(s\leq T),Z: s\geq t\right\}\\
	&= \sigma\left\{\mathbbm{1}(Q<T), \mathbbm{1}(Q<T)(Q\vee t),  T\vee t, Z\right\}, \label{eq:m_1}
	\end{align}
 by the definition of backwards martingale and properties of conditional expectations \citep{durrett2019probability},
	it suffices 
	to show that for any non-negative 
 measurable function $\rho$
 on $\{0,1\}\times\mathbb{R}_+\times \mathbb{R}_+ \times\mathbb{R}$ and any $0<s\leq t$,
	\begin{align}
	&\quad \ee{\{\bar N_Q(s) - \bar N_Q(t)\} \cdot \rho(\mathbbm{1}(Q<T), \mathbbm{1}(Q<T)(Q\vee t),  T\vee t, Z)} \\
	&= \ee{\{\bar A_Q(s;G) - \bar A_Q(t;G)\} \cdot \rho(\mathbbm{1}(Q<T), \mathbbm{1}(Q<T)(Q\vee t),  T\vee t, Z)}.
	\end{align}
	For any such function $\rho$, 
	\begin{align}
	&\ee{\{\bar N_Q(s) - \bar N_Q(t)\} \cdot 
	\rho(\mathbbm{1}(Q<T), \mathbbm{1}(Q<T)(Q\vee t),  T\vee t, Z)|Z} \\
	=& \E\left\{\mathbbm{1}(s\leq Q< t)\mathbbm{1}(Q<T) \cdot
	\rho(\mathbbm{1}(Q<T), Q\vee t,  T\vee t, Z)|Z\right\} \\
	=& \E\left\{\mathbbm{1}(s\leq Q< t)\mathbbm{1}(Q<T) \cdot
	\rho(1, 
	t,  T\vee t, Z)|Z\right\} \\
	=& \E\left[ \mathbbm{1}(s\leq Q\leq t)\mathbbm{1}(Q<T) \cdot \{\mathbbm{1}(T\geq t)\cdot \rho(1, t, T,Z)+\mathbbm{1}(T< t)\cdot \rho(1, t, t, Z)\}|Z\right ] \\
	=& \E\left[ \mathbbm{1}(s\leq Q\leq t) 
	\cdot \{ \mathbbm{1}(T\geq t)\cdot\rho(1, t, T,Z) 
	+ \mathbbm{1}(Q<T<t) 
	\cdot \rho(1, t, t, Z) \}|Z\right] \\
	=& \frac{1}{\beta(Z)}  \int_s^t \left\{ \int_t^\infty 
	\rho(1,t,u, Z) 
	+ \int_v^t 
	\rho(1,t,t, Z)\right\} f(u|Z) g(v|Z) 
	\ du\ dv \\
	=& \frac{1}{\beta(Z)}\{G({t-}|Z) - G({s-}|Z)\} \int_t^\infty 
	\rho(1,t,u, Z) dF(u|Z)\\
	&\quad\quad\quad + \frac{\rho(1,t,t, Z)}{\beta(Z)} \int_0^\infty \int_0^\infty \mathbbm{1}(s\leq v< u < t)  dF(u|Z) dG(v|Z).
	\end{align}
	On the other hand, 
	\begin{align}
	&\ee{\{\bar A_Q(s;G) - \bar A_Q(t;G)\} \cdot \rho(\mathbbm{1}(Q<T), \mathbbm{1}(Q<T)(Q\vee t),  T\vee t, Z)|Z} \\
	=& \E\left\{\left.\int_s^t \mathbbm{1}(Q\leq x < T) \frac{dG(x|Z)}{G(x|Z)}   \cdot \rho(\mathbbm{1}(Q<T), \mathbbm{1}(Q<T)(Q\vee t),  T\vee t, Z) \right|Z\right\} \\
	=& \E\left\{\left.\int_{Q\vee s}^{T\wedge t} \frac{dG(x|Z)}{G(x|Z)} \cdot \rho(1, t,  T\vee t, Z)\right|Z\right\} \\
	=& \E\left\{\left.\int_{Q\vee s}^{T\wedge t} \frac{dG(x|Z)}{G(x|Z)}  \mathbbm{1}(T\geq t ) \cdot \rho(1, t, T, Z)\right|Z\right\} \\
    &+ \E\left\{\left.\int_{Q\vee s}^{T\wedge t} \frac{dG(x|Z)}{G(x|Z)}\mathbbm{1}(T< t) \cdot \rho(1, t, t,Z)\right|Z\right\} \\
	=&\frac{1}{\beta(Z)} \int_0^\infty \int_0^\infty \int_{v\vee s}^{u\wedge t}  \frac{dG(x|Z)}{G(x|Z)}\mathbbm{1}(u\geq t) \cdot \rho(1,t,u,Z) f(u|Z) g(v|Z)\mathbbm{1}(v<u) \ du\ dv\\
	&+\frac{1}{\beta(Z)} \int_0^\infty \int_0^\infty \int_{v\vee s}^{u\wedge t}   \frac{dG(x|Z)}{G(x|Z)}\mathbbm{1}(u<t) \cdot \rho(1,t,t,Z) f(u|Z) g(v|Z)\mathbbm{1}(v<u) \ du\ dv\\
	=&\frac{1}{\beta(Z)} \int_t^\infty \int_s^t \rho(1,t,u,Z) \left\{\int_0^x g(v|Z)\ dv\right\}\frac{dG(x|Z)}{G(x|Z)} f(u|Z) \ du\\
	&+\frac{\rho(1,s,s,Z)}{\beta(Z)} \int_0^\infty \int_0^\infty  \mathbbm{1}(s\leq x < u < t) \left\{\int_0^xg(v|Z)\ dv\right\}\frac{dG(x|Z)}{G(x|Z)} f(u|Z)\ du\\
	=& \frac{1}{\beta(Z)}\{G({t-}|Z) - G({s-}|Z)\} \int_t^\infty \rho(1,t,u,Z) dF(u|Z)\\
	&
	+\frac{\rho(1,t,t,Z)}{\beta(Z)} \int_0^\infty \int_0^\infty  \mathbbm{1}(s\leq x< u< t) dF(x|Z) dG(v|Z).
	\end{align}
	Therefore,
	\begin{align}
	& \ee{\left.\{\bar N_Q(s) - \bar N_Q(t)\} \rho(\mathbbm{1}(Q<T), \mathbbm{1}(Q<T)(Q\vee t),  T\vee t, Z) \right|Z} \\
	=& \ee{\left. \{\bar A_Q(s) - \bar A_Q(t)\} \rho(\mathbbm{1}(Q<T), \mathbbm{1}(Q<T)(Q\vee t),  T\vee t, Z) \right|Z},
	\end{align}
	so
	\begin{align}
	& \ee{\{\bar N_Q(s) - \bar N_Q(t)\} \rho(\mathbbm{1}(Q<T), \mathbbm{1}(Q<T)(Q\vee t),  T\vee t, Z)} \\
	=& \ee{\{\bar A_Q(s) - \bar A_Q(t)\} \rho(\mathbbm{1}(Q<T), \mathbbm{1}(Q<T)(Q\vee t),  T\vee t, Z)}.
	\end{align}
  
\end{proof}

\bigskip
The following proposition states the property of $\bar M_Q$ that is useful in the later proofs.
\begin{proposition}\label{prop:int_dM}
Under Assumptions \ref{ass:quasi-indpendent} and \ref{ass:overlap},
	suppose that $\{\bar H(t)\}_{t\geq 0}$ is a process that is bounded on $[0,\tau_2]$, and $\bar H(t)$ is $\mathcal{\bar F}_{t+}$-measurable for all $t\geq 0$. Then if $G$ is the true CDF of $Q$ given $Z$ in the full data, we have 
	\begin{align}
	\E\left\{\left.\int_0^\infty \bar H(t) d\bar M_Q(t;G) \right|Z\right\} = 0, \quad 
    \E\left\{\int_0^\infty \bar H(t) d\bar M_Q(t;G)\right\} = 0. \label{eq:m1}
	\end{align}
\end{proposition}

\begin{proof}
	Let $H(t) = \bar H(\tau_2-t)$ for $t\in [0,\tau_2]$. By Assumption \ref{ass:overlap}, $Q\leq \tau_2$ a.s., so
	\begin{align}
	\int_0^\infty \bar H(t) d \bar M_Q(t;G) 
	& = \int_0^{\tau_2} \bar H(t) d \bar M_Q(t;G)\\
	&= -\int_0^{\tau_2} \bar H(\tau_2-u) d \bar M_Q(\tau_2-u;G)\\
	&= -\int_0^{\tau_2} H(u) d M_Q^{\tau_2}(u;G), 
	\end{align}
	where $M_Q^{\tau_2}$ is defined in Section \ref{sec:notation_preliminary} of the main paper. 
	Since $\bar H$ is bounded on $[0,\tau_2]$, $H$ is bounded on $[0,\tau_2]$. Besides, as shown in Lemma \ref{lem:M_Q},
	$\{\bar M_Q(t; G)\}_{t\geq 0}$ is a backwards martingale with respect to $\{\mathcal{\bar F}_t\}_{t\geq 0}$, so 
	$\{M_Q^{\tau_2}(t;G) \}_{0\leq t\leq \tau_2}$ is a martingale with respect to the filtration $\{\mathcal{G}_t^{\tau_2}\}_{0\leq t\leq\tau_2}$ (defined in Section \ref{sec:notation_preliminary} of the main paper). Therefore, $\{\int_0^t H(u) dM_Q^{\tau_2}(u;G)\}_{0\leq t\leq\tau_2}$ is a martingale with respect to the filtration $\{\mathcal{G}_t^{\tau_2}\}_{0\leq t\leq\tau_2}$ \citep{fleming2011counting}. 
 Thus,
 \begin{align}
	\E\left\{\int_0^\infty \bar H(t) d\bar M_Q(t;G)\right\} = -\E\left\{\int_0^{\tau_2} H(u) d M_Q^{\tau_2}(u;G)\right\} = 0.
\end{align}
 In addition, since $\sigma(Z) \subseteq \mathcal{G}_t^{\tau_2}$ for all $t\geq 0$, we have 
 \begin{align}
	\E\left\{\left.\int_0^\infty \bar H(t) d\bar M_Q(t;G)\right|Z\right\} = -\E\left\{\left.\int_0^{\tau_2} H(u) d M_Q^{\tau_2}(u;G)\right|Z\right\} = 0.
\end{align}
  
\end{proof}

\newpage
\section{Identification of $\theta$}\label{app:identification}

\begin{proof}[Proof of Lemma \ref{lem:psi_identification}]
	Under Assumptions \ref{ass:quasi-indpendent} and \ref{ass:beta*(z)>0},
	\begin{align}
	\E\left\{\frac{\nu(T)}{G(T|Z)}\right\} &= \int\int\int \frac{\nu(t)}{G(t|z)} \frac{ f(t|z)g(q|z)h(z) \mathbbm{1}(q<t)}{\beta}\ dq\ dt\ dz \\
	& = \frac{1}{\beta}\int\int \frac{\nu(t)}{G(t|z)} \left\{\int g(q|z) \mathbbm{1}(q<t) dq \right\} f(t|z)h(z)\ dt\ dz\\
	& =  \frac{1}{\beta}\int\int \nu(t)  f(t|z)h(z)\ dt\ dz \\
	& = \frac{\theta}{\beta}.
	\end{align}
	As a special case with $\nu\equiv 1$, we have $
	\E\left\{1/G(T|Z)\right\} = 1/\beta$. 
	Therefore,
	\begin{align}
	\theta = \left.\E\left\{\frac{\nu(T)}{G(T|Z)}\right\}\right/\E\left\{\frac{1}{G(T|Z)}\right\}.
	\end{align}

	We now show that $G$ can be characterized by the reverse time hazard function $\alpha$. 
	By definition, 
	\begin{align}
	\alpha(q|z) = \left.\left\{ \frac{\partial}{\partial q} G(q|z)\right\}\right/G(q|z).
	\end{align}
	Besides, $\lim_{q\to\infty} G(q|z)=1$. By solving the above differential equation, we get \begin{align}
	G(q|z) &= \exp\left\{-\int_q^\infty\alpha(v|z) dv\right\}. \label{eq:G_alpha}
	\end{align}
	
 Next, we will show that $\alpha$ is identifiable from observed data distribution. Specifically, we will show that under Assumptions \ref{ass:quasi-indpendent} and \ref{ass:beta*(z)>0}, 
	\begin{align}
	\alpha(q|z) &  = \frac{p_{Q|Z}(q|z)}{\pp{Q\leq  q < T|Z=z}}. \label{eq:identify_alpha}
	\end{align}
	Recall from \eqref{eq:p(q|z)} and \eqref{eq:p(q,t|z)},
	\begin{align}
	p_{Q|Z}(q|z) & = \frac{1-F(q|z)}{\beta(z)} g(q|z), \quad p_{Q,T|Z}(q,t|z)  = \frac{\mathbbm{1}(q<t)}{\beta(z)} f(t|z)g(q|z),
	\end{align}
	so 
	\begin{align}
	\pp{Q\leq q <T|Z=z} 
	& = \int_0^\infty \int_0^\infty \mathbbm{1}(v\leq q<u) p_{T,Q|Z}(u,v|z)\ du\ dv \\
	& = \int_0^\infty \int_0^\infty \mathbbm{1}(v\leq q<u) \frac{\mathbbm{1}(v<u)}{\beta(z)} f(u|z) g(v|z)\ du\ dv \\
	& = \frac{\{1-F(q|z)\}G(q|z)}{\beta(z)}.
	\end{align}
	Therefore, 
	\begin{align}
	\frac{ p_{Q|Z}(q|z)}{\pp{Q\leq q <T|Z=z} }  = \frac{g(q|z)}{G(q|z)} = \alpha(q|z).
	\end{align} 
  
\end{proof}

\newpage
\section{Derivation of the EIC}

As previously defined, $\{P_\epsilon\}$ denotes a regular parametric submodel for observed data and $p_\epsilon(\cdot)$ the corresponding density. We use for instance, $p_{Q,Z}(\cdot,\cdot;\epsilon)$ and $p_{Q|Z}(\cdot|\cdot;\epsilon)$ to denote the density of $(Q,Z)$ and conditional density of $Q$ given $Z$ under $P_\epsilon$.
We have for example, 
\begin{align}
    \Sc_{Q,Z}(q,z) = \left.\frac{\partial}{\partial \epsilon} \log p_{Q,Z}(q,z;\epsilon)\right|_{\epsilon = 0},
    \quad \Sc_{Q|Z}(q|z) = \left. \frac{\partial}{\partial \epsilon} \log p_{Q|Z}(q|z;\epsilon) \right|_{\epsilon = 0}.
\end{align}
The following properties of the scores  will be useful in the proofs. 
From the definition, 
\begin{align}
\Sc(\cdot) 
= \left.\frac{\partial}{\partial \epsilon} \log p_\epsilon(\cdot)\right|_{\epsilon = 0} = \left.\left.\left\{\frac{\partial}{\partial \epsilon} p_\epsilon(\cdot)\right|_{\epsilon = 0}\right\} \right/p(\cdot), \label{eq:app1}
\end{align}
so we have
\begin{align}
\left.\frac{\partial}{\partial \epsilon} p_\epsilon(\cdot)\right|_{\epsilon = 0}
= \Sc(\cdot) p(\cdot). \label{eq:dp=Sp}
\end{align}
In addition, from the definition, we can show that
$\ee{\Sc_{Q,T,Z}(Q,T,Z)} = 0$, and $ \ee{\Sc_{Q,T|Z}(Q,T,Z)|Z} = 0$, etc. %

\bigskip
The quantities $K$, $L$ and $R$ below appear during the computation to derive the influence curve. Let
\begin{align}
K(v|z) &= \pp{Q\leq v<T |Z=z}
=  \frac{G(v|z)\{1-F(v|z)\}}{\beta(z)} ,\label{eq:eq2} \\
L(v|z) &= \ee{\left.\frac{\{\nu(T)-\theta\}\mathbbm{1}(T\leq v)}{G(T|z)}\right|Z=z} \nonumber \\
&= \frac{1}{\beta(z)} \int_0^v \{\nu(t)-\theta\}  f(t|z) dt = \frac{m(v,z;F) - \theta F(v|z)}{\beta(z)}, \label{eq:eq3}\\
R(y,u|z) &= \int_0^\infty \frac{ L(v|z) }{K(v|z)^2}  \mathbbm{1}( y\leq v<u)   p_{Q|Z}(v|z) dv \nonumber \\
& = \int_0^\infty \frac{m(v,z;F) - \theta F(v|z)}{G(v|z)^2\left\{1-F(v|z)\right\}} \mathbbm{1}( y\leq v<u) dG(v|z). \label{eq:eq4}
\end{align} 

We have immediately 
\begin{align}
&\quad \frac{L(Q|Z)}{K(Q|Z)} - R(Q,T|Z) \nonumber \\
&=  \frac{m(Q,Z;F) - \theta F(Q|Z)}{G(Q|Z)\{1-F(Q|Z)\}} - \int_0^\infty \frac{m(v,Z;F) - \theta F(v|Z)}{G(v|Z)\{1-F(v|Z)\}}  \mathbbm{1}(Q\leq v<T)  \frac{dG(v|Z)}{G(v|Z)} \nonumber \\
&= - \int_0^\infty \frac{m(v,Z;F) - \theta F(v|Z)}{G(v|Z)\{1-F(v|Z)\}} d\bar M_Q(v;G). \label{eq:LKR}
\end{align}

\newpage
\begin{proof}[Proof of Lemma \ref{lem:tangent_space}]
	
	(i) To prove that $\dot P = L_2^0(P_{T,Z}) + L_2^0(P_{Q,Z})$, we will first that 
 $\dot P \subseteq L_2^0(P_{T,Z}) + L_2^0(P_{Q,Z})$ and then show that $\dot P \supseteq L_2^0(P_{T,Z}) + L_2^0(P_{Q,Z})$. 
 The proof follows the techniques used in \citet{bickel1993efficient} and \citet{tsiatis2006semiparametric}.

	Recall that under Assumption \ref{ass:quasi-indpendent}, the joint density for $(Q,T,Z)$ in observed data is
	\begin{align}
	p(q,t,z) 
	& = \frac{f(t|z)g(q|z)h(z) \mathbbm{1}(q<t)}{\beta}\\
	& = \frac{f(t|z)g(q|z)h(z) \mathbbm{1}(q<t)}{\int \mathbbm{1}(q<t) f(t|z)g(q|z)h(z)\ dt\ dq\ dz}.
	\end{align}
	For any regular parametric submodel $P_\epsilon$ satisfing Assumption \ref{ass:quasi-indpendent}, the corresponding observed data density is
	\begin{align}
	p_\epsilon(q,t,z) 
	&= \frac{f_\epsilon(t|z)g_\epsilon(q|z)h_\epsilon(z) \mathbbm{1}(q<t)}{\int \mathbbm{1}(q<t) f_\epsilon(t|z)g_\epsilon(q|z)h_\epsilon(z)\ dt\ dq\ dz}. 
	\end{align}
	Denote 
	\begin{align}
	a(t,z) & = \left.\frac{\partial}{\partial \epsilon} \log f_\epsilon(t|z)\right|_{\epsilon = 0}, \\
	b(q,z) & = \left.\frac{\partial}{\partial \epsilon} \log g_\epsilon(q|z)\right|_{\epsilon = 0}, \\
	c(z) & = \left.\frac{\partial}{\partial \epsilon} \log h_\epsilon(z)\right|_{\epsilon = 0}.
	\end{align}
	and denote 
	\begin{align}
	L_2^0(P_Z) = \left\{\tilde c\in L_2(P_{Z}):\ \int_\Z \tilde c(z)\ d P_{Z}(z) = 0 \right\}.
	\end{align}
	Then the score for observed data is
	\begin{align}
	\left.\frac{\partial}{\partial \epsilon} \log p_\epsilon(q,t,z)\right|_{\epsilon = 0}
	& = \left.\frac{\partial}{\partial \epsilon} \log f_\epsilon(t|z)\right|_{\epsilon = 0} + \left.\frac{\partial}{\partial \epsilon} \log g_\epsilon(q|z)\right|_{\epsilon = 0} + \left.\frac{\partial}{\partial \epsilon} \log h_\epsilon(z)\right|_{\epsilon = 0} \\
	& \quad\quad - \left.\frac{\partial}{\partial \epsilon} \log \int \mathbbm{1}(q<t) f_\epsilon(t|z)g_\epsilon(q|z)h_\epsilon(z)\ dq\ dt\ dz \right|_{\epsilon = 0} \\ 
	&= a(t,z) + b(q,z) + c(z)\\
	& \quad\quad -\frac{\int \mathbbm{1}(q<t) a(t,z)f(t|z)g(q|z)h(z)\ dt\ dq\ dz}{\int \mathbbm{1}(q<t) f(t|z)g(q|z)h(z)\ dq\ dt\ dz}\\
	& \quad\quad -\frac{\int \mathbbm{1}(q<t) b(q,z)f(t|z)g(q|z)h(z)\ dq\ dt\ dz}{\int \mathbbm{1}(q<t) f(t|z)g(q|z)h(z)\ dq\ dt\ dz}\\
	& \quad\quad -\frac{\int \mathbbm{1}(q<t) c(z)f(q|z)g(t|z)h(z)\ dq\ dt\ dz}{\int \mathbbm{1}(q<t) f(q|z)g(t|z)h(z)\ dq\ dt\ dz}\\
	& = a(t,z) + b(q,z)  + c(z) - \ee{a(T,Z)} - \ee{b(Q,Z)} - \ee{c(Z)}\\
	& = \left\{a(t,z)- \ee{a(T,Z)}\right\} + \left\{b(q,z)- \ee{b(Q,Z)}\right\}  + \left\{c(z) - \ee{c(Z)}\right\}\\
	& \in L_2^0(P_{T,Z}) + L_2^0(P_{Q,Z})+ L_2^0(P_Z).
	\end{align}
	Since $L_2^0(P_Z) \subseteq L_2^0(P_{T,Z})$, we have
	\begin{align}
	\dot P \subseteq L_2^0(P_{T,Z}) + L_2^0(P_{Q,Z})+ L_2^0(P_Z) \subseteq L_2^0(P_{T,Z}) + L_2^0(P_{Q,Z}).
	\end{align}
	
	On the other hand, we will show that $\dot P \subseteq L_2^0(P_{T,Z}) + L_2^0(P_{Q,Z})$. Since bounded functions of $(T,Z)$ and $(Q,Z)$ are dense in $L_2^0(P_{T,Z})$ and $L_2^0(P_{Q,Z})$, respectively, it suffices to show that for any bounded functions $a\in L_2^0(P_{T,Z})$, $b\in L_2^0(P_{Q,Z})$, there exists a parametric submodel whose score is $a(t,z) + b(q,z)$. For any bounded functions $a\in L_2^0(P_{T,Z})$, $b\in L_2^0(P_{Q,Z})$, let
	\begin{align}
	\tilde a(t,z) &= a(t,z) - \int a(t,z) f(t|z) dt, \\
	\tilde b(q,z) &= b(q,z) - \int b(q,z) g(q|z) dq, \\
	\tilde c(z) & =  \int a(t,z) f(t|z) dt +\int b(q,z) g(q|z) dq\\
	&\quad\quad\quad  - \int \left\{\int a(t,z) f(t|z) dt +\int b(q,z) g(q|z) dq \right\} h(z) \ dz.
	\end{align}
	Consider the parametric submodel 
	\begin{align}
	p_\epsilon(q,t,z) = \frac{f_\epsilon(t|z)g_\epsilon(q|z)h_\epsilon(z) \mathbbm{1}(q<t)}{\int \mathbbm{1}(q<t) f_\epsilon(t|z)g_\epsilon(q|z)h_\epsilon(z)\ dt\ dq\ dz} \label{eq:submodel}
	\end{align}
	with
	\begin{align}
	f_\epsilon(t|z) &= f(t|z)\{1+\epsilon \tilde a(t,z)\},\\
	g_\epsilon(q|z) &= g(t|z)\{1+\epsilon \tilde b(q,z)\},\\
	h_\epsilon(z) &= h(z)\{1+ \epsilon\tilde c(z)\},
	\end{align}
	and $\epsilon$ is the parameter chosen sufficiently small to guarantee that $f_\epsilon$, $g_\epsilon$ and $h_\epsilon$ are positive. Since 
	\begin{align}
	\int f_\epsilon(t|z) dt &= 1 + \epsilon \int \tilde a(t,z) f(t|z) dt = 1, \\
	\int g_\epsilon(q|z) dq &= 1 + \epsilon \int \tilde b(q,z) g(q|z) dq = 1,\\
	\int h_\epsilon(z) dz & = 1 + \epsilon \int \tilde c(z) h(z) dz = 1,
	\end{align}
	$f_\epsilon$, $g_\epsilon$, $h_\epsilon$ are density functions and, at $\epsilon=0$ they are (conditional) densities of
	of $T$ given $Z$, $Q$ given $Z$, and $Z$, respectively. Therefore the class of functions given in \eqref{eq:submodel} is a parametric submodel. Moreover, it can be verified that
	\begin{align}
	\frac{\partial}{\partial \epsilon} \log p_\epsilon(q,t,z) = a(t,z) + b(q,z).
	\end{align}
	Therefore, 
	\begin{align}
	L_2^0(P_{T,Z}) + L_2^0(P_{Q,Z}) \subseteq \dot P.
	\end{align}
	Thus,
	\begin{align}
	\dot P = L_2^0(P_{T,Z}) + L_2^0(P_{Q,Z}).
	\end{align}
 
	\medskip
	(ii) We now characterize $\dot P^{\perp}$. 
	Since the Hilbert space is $L_2^0(P_{Q,T,Z})$, 
	by the definition of orthogonal complement,
	\begin{equation}
	\dot P^{\perp} = \left\{ 
	\xi(q,t,z)\in L_2^0(P_{Q,T,Z}): \begin{array}{l}
	\ee{\xi(Q,T,Z)\{a(Q,Z)+b(T,Z)\}} = 0,\\
	\ \text{for all } a\in L_2^0(P_{Q,Z}),\  b\in L_2^0(P_{T,Z}).
	\end{array}
	\right\}. 
	\end{equation}
	We would like to find all functions $\xi(Q,T,Z)$ such that for any $a\in L_2^0(P_{T,Z})$ and $b\in L_2^0(P_{Q,Z})$, 
	\begin{align}
	\ee{\xi(Q,T,Z)\{a(Q,Z)+b(T,Z)\}} = 0. \label{eq:eq16}
	\end{align}
	Since $0\in L_2^0(P_{Q,Z})$ and $0\in L_2^0(P_{T,Z})$, it can be verified that  \eqref{eq:eq16} holds if and only if the following two conditions hold:
	\begin{align}
	&\E\left\{\xi(Q,T,Z)a(Q,Z)\right\} = 0, \quad \text{for all } a \in L_2^0(P_{Q,Z}), \\
	&\E\left\{\xi(Q,T,Z)b(T,Z)\right\} = 0, \quad  \text{for all }  b \in L_2^0(P_{T,Z}).
	\end{align}
	Therefore,
	\begin{align}
	(\dot P)^\perp  &= \left\{ \xi\in L_2^0(P_{Q,T,Z}):
	\begin{array}{l}
	\E\left\{\xi(Q,T,Z)a(Q,Z)\right\} = 0, \quad \text{for all } a \in L_2^0(P_{Q,Z}),\\
	\E\left\{\xi(Q,T,Z)b(T,Z)\right\} = 0, \quad  \text{for all } b \in L_2^0(P_{T,Z}).
	\end{array}
	\right\}.
	\end{align}
  
\end{proof}

\bigskip
\begin{proof}[Proof of Lemma \ref{lem:IF}]
	To derive an influence curve,  we consider a regular parametric submodel $P_\epsilon$ for the observed data $(Q, T, Z)$. 
	The proof contains much algebra related to the computation under $P_\epsilon$. 
	
	Let $\E_\epsilon$ denote the expectation taken with respect to the distribution $P_\epsilon$, and $G(\cdot|\cdot;\epsilon)$ the full data CDF of $Q$ given $Z$ under  $P_\epsilon$.
 Then
	\begin{align}
	\theta(P_\epsilon) = \frac{\E_\epsilon\left\{ {\nu(T)}/{G(T|Z;\epsilon)}\right\}}{\E_\epsilon\left\{{1}/{G(T|Z;\epsilon)}\right\}}.
	\end{align}
	We want to find a function $\varphi(O) = \varphi(Q,T,Z)$ such that 
	\begin{align}
	\left.\frac{\partial}{\partial \epsilon}\theta(P_\epsilon)\right|_{\epsilon=0} = \E\left\{\varphi(Q,T,Z)\Sc_{Q,T,Z}(Q,T,Z)\right\}, \label{eq:goal1}
	\end{align}
	and 
	$\E\left\{\varphi(O)\right\} = 0$.
	Then $\varphi$ will be an influence curve.
	We will first compute $\left.\partial\theta(P_\epsilon)/\partial \epsilon\right|_{\epsilon=0}$, and then try to write it in the form of the right hand side (RHS) of \eqref{eq:goal1}. 
 
	The derivative
	\begin{align}
	\frac{\partial }{\partial \epsilon}\theta(P_\epsilon)
	&= \left.\left[\frac{\partial }{\partial \epsilon}\E_\epsilon\left\{\frac{\nu(T)}{G(T|Z;\epsilon)}\right\}\right]\right/\E_\epsilon\left\{\frac{1}{G(T|Z;\epsilon)}\right\} \\
	&\quad\quad\quad  - \left.\E_\epsilon\left\{\frac{\nu(T)}{G(T|Z;\epsilon)}\right\} \left[\frac{\partial }{\partial \epsilon}\E_\epsilon\left\{\frac{1}{G(T|Z;\epsilon)}\right\}\right]\right/\left[\E_\epsilon\left\{\frac{1}{G(T|Z;\epsilon)}\right\}\right]^2.
	\end{align}
	Because in the proof of Lemma \ref{lem:psi_identification} we have shown that  $\E\left\{1/G(T|Z)\right\} = 1/\beta$ and $\E\left\{\nu(T)/G(T|Z)\right\} = \theta/\beta$, therefore
	\begin{align}
	\left.\frac{\partial }{\partial \epsilon}\theta(P_\epsilon)\right|_{\epsilon =0} 
	&= \beta\cdot \left.\frac{\partial }{\partial \epsilon}
	\E_\epsilon \left\{\frac{\nu(T)}{G(T|Z;\epsilon)}\right\}\right|_{\epsilon = 0} 
	-\beta\theta \cdot 
	\frac{\partial }{\partial \epsilon}
	\E_\epsilon \left.\left\{\frac{1}{G(T|Z;\epsilon)}\right\}\right|_{\epsilon = 0}
	\label{eq:lem2_1}
	\end{align}
	We will compute the two derivatives in \eqref{eq:lem2_1} next. 
	\begin{align}
	\frac{\partial }{\partial \epsilon}\E_\epsilon\left\{\frac{\nu(T)}{G(T|Z;\epsilon)}\right\}
	& = \frac{\partial }{\partial \epsilon} \int \frac{\nu(t)}{G(t|z;\epsilon)} p_{Q,T,Z}(q,t,z;\epsilon) \ dq\ dt\ dz \\
	& = \int \frac{\nu(t)}{G(t|z;\epsilon)} \left\{\frac{\partial }{\partial \epsilon} p_{Q,T,Z}(q,t,z;\epsilon) \right\} \ dq\ dt\ dz \\
	&\quad\quad\quad - \int \frac{\nu(t)}{G(t|z;\epsilon)^2} \left\{ \frac{\partial}{\partial \epsilon} G(t|z;\epsilon)\right\}  p_{Q,T,Z}(q,t,z;\epsilon) \ dq\ dt\ dz.
	\end{align}
	From \eqref{eq:dp=Sp} we have
	\begin{align}
	\left.\frac{\partial}{\partial \epsilon} p_{Q,T,Z}(q,t,z;\epsilon) \right|_{\epsilon = 0} = \Sc_{Q,T,Z}(q,t,z) p_{Q,T,Z}(q,t,z),
	\end{align}
	so
	\begin{align}
	\left.\frac{\partial}{\partial \epsilon}\E_\epsilon\left\{\frac{\nu(T)}{G(T|Z;\epsilon)}\right\} \right|_{\epsilon = 0}
	& = \int \frac{\nu(t)}{G(t|z)} \Sc_{Q,T,Z}(q,t,z) \ p_{Q,T,Z}(q,t,z) \ dq\ dt\ dz \\
	&\quad\quad\quad - \int \frac{\nu(t)}{G(t|z)^2} \left\{\left. \frac{\partial}{\partial \epsilon} G(t|z;\epsilon)\right|_{\epsilon = 0} \right\}  p_{Q,T,Z}(q,t,z) \ dq\ dt\ dz \\
	& = \E \left\{\frac{\nu(T) }{G(T|Z)} \Sc_{Q,T,Z}(Q,T,Z)\right\} -  \ee{\frac{\nu(T)}{G(T|Z)^2} \left\{ \left.\frac{\partial}{\partial \epsilon} G(T|Z;\epsilon)\right|_{\epsilon =0}\right\}}. \\
 \label{eq:lem2_2}
	\end{align}
	As a special case with $\nu\equiv 1$, we have
	\begin{align}
	\left.\frac{\partial}{\partial \epsilon}\E_\epsilon\left\{\frac{1}{G(T|Z;\epsilon)}\right\} \right|_{\epsilon = 0}
	& = \E \left\{\frac{1}{G(T|Z)} \Sc_{Q,T,Z}(Q,T,Z)\right\} -  \ee{\frac{1}{G(T|Z)^2} \left\{ \left.\frac{\partial}{\partial \epsilon} G(T|Z;\epsilon)\right|_{\epsilon =0}\right\}}. \\
 \label{eq:lem2_3}
	\end{align}
	From \eqref{eq:lem2_1}, \eqref{eq:lem2_2} and \eqref{eq:lem2_3} we have 
	\begin{align}
	\left.\frac{\partial }{\partial \epsilon}\theta(P_\epsilon)\right|_{\epsilon=0} 
	& = \beta  \E \left\{\frac{\nu(T) - \theta}{G(T|Z)} \Sc_{Q,T,Z}(Q,T,Z)\right\} - \beta \ee{\frac{\nu(T)-\theta}{G(T|Z)^2} \left\{\left.\frac{\partial}{\partial \epsilon} G(T|Z;\epsilon)\right|_{\epsilon =0}\right\}}. \label{eq:lem2_4}
	\end{align}
	If we can write the second term of the RHS in \eqref{eq:lem2_4} into $\beta  \E\left\{\varphi_2(Q,T,Z) \Sc_{Q,T,Z}(Q,T,Z)\right\}$ for some function $\varphi_2$, then 
	$\varphi(Q,T,Z) = \beta[\{\nu(T) - \theta\}/G(T|Z) -\varphi_2(Q,T,Z)]$
	satisfies \eqref{eq:goal1}. 
	\\
	
	We now focus on finding $\varphi_2$. 
	To compute $\left.\partial G(t|z;\epsilon)/\partial \epsilon\right|_{\epsilon =0}$
	recall from \eqref{eq:G_alpha}, 
	$G(t|z;\epsilon) = \exp\left\{-\int_t^\infty \alpha(v|z;\epsilon) dv\right\}$,
	so
	\begin{align}
	\frac{\partial}{\partial \epsilon} G(t|z;\epsilon)
	&= -\exp\left\{-\int_t^\infty \alpha(v|z;\epsilon) dv\right\} \frac{\partial}{\partial \epsilon}\int_t^\infty \alpha(v|z;\epsilon) dv\\
	& = -G(t|z)\int_t^\infty 
	\left\{\frac{\partial}{\partial \epsilon}\alpha(v|z;\epsilon) \right\} dv. \label{eq:IF_proof_1}
	\end{align}
	By \eqref{eq:identify_alpha}, we have 
	$$\alpha(v|z;\epsilon) = \frac{ p_{Q|Z}(v|z;\epsilon)}{\int_0^v\int_v^\infty 
	p_{Q,T|Z}(y,u|z;\epsilon)\ dy\ du} = \frac{p_{Q|Z}(v|z;\epsilon)}{ K(v|z;\epsilon)}.$$
	In addition, using \eqref{eq:dp=Sp} we have
	\begin{align}
	& \left.\frac{\partial}{\partial \epsilon}\alpha(v|z;\epsilon) \right|_{\epsilon=0} = \left.\frac{\partial}{\partial \epsilon} \left\{\frac{ p_{Q|Z}(v|z;\epsilon)}{K(v|z;\epsilon)}\right\}\right|_{\epsilon=0}  \\
	 = & \frac{ \partial p_{Q|Z}(v|z;\epsilon)/\partial \epsilon|_{\epsilon=0}}{K(v|z)} 
	- \frac{ p_{Q|Z}(v|z)  \int_0^v\int_v^\infty \left\{\partial p_{Q,T|Z}(y,u|z;\epsilon)/\partial\epsilon|_{\epsilon=0} \right\}\ dy\ du }{K(v|z)^2} \\
	 = & \frac{ \Sc_{Q|Z}(v|z)p_{Q|Z}(v|z)}{K(v|z)}   
	- \frac{ p_{Q|Z}(v|z)  \int_0^v\int_v^\infty \Sc_{Q,T|Z}(y,u|z) p_{Q,T|Z}(y,u|z)\ dy\ du }{K(v|z)^2}. 
	\end{align}
	Therefore, from \eqref{eq:IF_proof_1},
	\begin{align}
	\left.\frac{\partial}{\partial \epsilon} G(t|z;\epsilon)\right|_{\epsilon =0}
	&= -G(t|z)\int_t^\infty \frac{ \Sc_{Q|Z}(v|z)p_{Q|Z}(v|z)}{K(v|z)} dv \\
	&\quad\quad + G(t|z)\int_t^\infty \frac{ p_{Q|Z}(v|z)  \int_0^v\int_v^\infty \Sc_{Q,T|Z}(y,u|z) p_{Q,T|Z}(y,u|z)\ dy\ du }{K(v|z)^2} dv.
	\end{align}
	So we have
	\begin{align}
	\ee{\frac{\nu(T)-\theta}{G(T|Z)^2}\left\{ \left.\frac{\partial}{\partial \epsilon} G(T|Z;\epsilon)\right|_{\epsilon =0}\right\}}
	= E_1 + E_2
	\end{align}
	where $E_1,E_2$ are stated as follows. We have  
	\begin{align}
	E_1 & = -E\left[
	\int_0^\infty \frac{\nu(t)-\theta}{G(t|Z)} \left\{\int_t^\infty 
	\frac{\Sc_{Q|Z}(v|Z)p_{Q|Z}(v|Z)}{K(v|Z)}  dv \right\}  p_{T|Z}(t|Z) dt\right]  \\
	& = -E\left[ \int_0^\infty \left\{\int_0^v \frac{\nu(T)-\theta}{G(t|Z)}\ p_{T|Z}(t|Z) dt\right\}  \frac{1}{K(v|Z)} \ \Sc_{Q|Z}(v|Z) p_{Q|Z}(v|Z) dv 
	\right] \\
	&  = -E\left[ \int_0^\infty \frac{L(v|Z)}{K(v|Z)} \  \Sc_{Q|Z}(v|Z)p_{Q|Z}(v|Z) dv 
	\right] \\
	& = -\E\left\{\frac{L(Q|Z)}{K(Q|Z)}\ \Sc_{Q|Z}(Q|Z)\right\}
	\end{align}
	Because 
    	$\Sc_{Q,T,Z} = \Sc_{Q,T|Z} +\Sc_{Z}$, we have
	\begin{align}	
	E_1 & = -\E\left\{\frac{L(Q|Z)}{K(Q|Z)}\Sc_{Q,T|Z}(Q,T|Z)\right\}\label{eq:lem2_5} \\
	& = -\E\left\{\frac{L(Q|Z)}{K(Q|Z)}\Sc_{Q,T,Z}(Q,T,Z)\right\} + \E\left\{\frac{L(Q|Z) }{K(Q|Z)}\Sc_{Z}(Z)\right\} \label{eq:lem2_6} \\
	& = -\E\left\{\frac{L(Q|Z) }{K(Q|Z)}\Sc_{Q,T,Z}(Q,T,Z)\right\} + \ee{\E\left\{\left.\frac{L(Q|Z)}{K(Q|Z)}\right|Z\right\}\ \Sc_{Z}(Z)} \\
    & = -\E\left(\left[\frac{L(Q|Z)}{K(Q|Z)} -\E\left\{\left.\frac{L(Q|Z) }{K(Q|Z)}\right|Z\right\} \right]\Sc_{Q,T,Z}(Q,T,Z)\right) \\
    &\quad\quad\quad - \ee{\E\left\{\left.\frac{L(Q|Z)}{K(Q|Z)}\right|Z\right\}\ \Sc_{Q,T|Z}(Q,T|Z)}.
	\end{align}
 Since $\E\left\{\Sc_{Q,T|Z}(Q,T|Z)|Z\right\} = 0$, by the law of total expectation, we have
 \begin{align}
     \ee{\E\left\{\left.\frac{L(Q|Z)}{K(Q|Z)}\right|Z\right\}\ \Sc_{Q,T|Z}(Q,T|Z)} 
     & = \E\left(\E\left[\left.\E\left\{\left.\frac{L(Q|Z)}{K(Q|Z)}\right|Z\right\}\ \Sc_{Q,T|Z}(Q,T|Z)\right| Z \right] \right) \\
     & = \E\left[\E\left\{\left.\frac{L(Q|Z)}{K(Q|Z)}\right|Z\right\}\ \E\left\{\left.\Sc_{Q,T|Z}(Q,T|Z) \right| Z\right\} \right] \\
     & = 0.
 \end{align}
 Therefore, 
 \begin{align}
     E_1 & = -\E\left(\left[\frac{L(Q|Z)}{K(Q|Z)} -\E\left\{\left.\frac{L(Q|Z) }{K(Q|Z)}\right|Z\right\} \right]\Sc_{Q,T,Z}(Q,T,Z)\right). \label{eq:lem2_7} 
 \end{align}

	In addition, 
	\begin{align}
	E_2 & = \E\left[ \int_0^\infty \frac{\nu(t)-\theta}{G(t|Z)} \left\{ \int_t^\infty \frac{ p_{Q|Z}(v|Z)}{K(v|Z)^2} \right.\right. \\
    &\quad\quad\quad\quad\quad\quad\quad\quad\quad  \left. \left. \cdot \int_0^v \int_v^\infty \Sc_{Q,T|Z}(y,u|Z)\ p_{Q,T|Z}(y,u|Z)\ dy\ du\ dv \right\} p_{T|Z}(t|Z)\ dt\right]\\
	& = \E\left[ \int_0^\infty \int_0^\infty \left\{\int_y^u \frac{ L(v|Z) }{K(v|Z)^2} p_{Q|Z}(v|Z) dv \right\}   \Sc_{Q,T|Z}(y,u|Z)p_{Q,T|Z}(y,u|Z)\ dy\ du \right]\\
	& = \E\left[\int_0^\infty \int_0^\infty R(y,u|Z) \Sc_{Q,T|Z}(y,u|Z)p_{Q,T|Z}(y,u|Z)\ dy\ du \right]\\
	& = \E\left\{R(Q,T|Z) \Sc_{Q,T|Z}(Q,T|Z)\right\}.
	\end{align}
    Because $\mathcal{S}_{Q,T,Z} = \mathcal{S}_{Q,T|Z} +\mathcal{S}_{Z}$ and $\E\left\{\mathcal{S}_{Q,T|Z}(Q,T|Z)|Z\right\} = 0$, we have 
    \begin{align}
        E_2 & = \E\left\{R(Q,T|Z) \Sc_{Q,T,Z}(Q,T,Z)\right\} - \E\left\{R(Q,T|Z) \Sc_{Z}(Z)\right\}\\
        & = \E\left\{R(Q,T|Z) \Sc_{Q,T,Z}(Q,T,Z)\right\} - \E\left[\E\left\{\left.R(Q,T|Z)\right|Z\right\} \Sc_{Z}(Z) \right]\\
	    & = \E\left(\left[R(Q,T|Z) - \E\left\{R(Q,T|Z)|Z\right\} \right]\Sc_{Q,T,Z}(Q,T,Z)\right).
    \end{align}
	Putting the above together, we have 
	\begin{align} \ee{\frac{\nu(T)}{G(T|Z)^2}\left\{ \left.\frac{\partial}{\partial \epsilon} G(T|Z;\epsilon)\right|_{\epsilon =0}\right\}} = -\E\left\{\varphi_2(Q,T,Z)\Sc_{Q,T,Z}(Q,T,Z)\right\},
	\end{align}
	where
	\begin{align}
	\varphi_2(Q,T,Z) &=\left[\frac{L(Q|Z)}{K(Q|Z)} -\E\left\{\left.\frac{L(Q|Z) }{K(Q|Z)}\right|Z\right\} \right]
	- \big[R(Q,T|Z) - \E\left\{R(Q,T|Z)|Z\right\} \big] \\
	& = \left\{\frac{L(Q|Z)}{K(Q|Z)} -R(Q,T|Z) \right\}- e(Z), 
	\end{align}
	with
	$$
	e(Z) = \E\left\{\left.\frac{L(Q|Z)}{K(Q|Z)} - R(Q,T|Z)\right|Z\right\}.
	$$
	From \eqref{eq:LKR}	
	and by Proposition \ref{prop:int_dM},
	\begin{align}
	e(Z) 
	& = - \E\left\{\left. \int_0^\infty \frac{m(v,Z;F) - \theta F(v|Z)}{G(v|Z)\{1-F(v|Z)\}} d\bar M_Q(v;G) \right|Z\right\} =0. 
	\end{align}
	Therefore,
	\begin{align}
	\varphi_2(Q,T,Z) 
	&= - \int_0^\infty \frac{m(v,Z;F) - \theta F(v|Z)}{G(v|Z)\{1-F(v|Z)\}} d\bar M_Q(v;G).
	\end{align}
	
	From the above,
	\begin{align}
	\left.\frac{\partial }{\partial \epsilon}\theta(P_\epsilon)\right|_{\epsilon=0} 
	& = \ee{\beta \left\{\frac{\nu(T) - \theta}{G(T|Z)} + \varphi_2(Q,T,Z)\right\} \Sc_{Q,T,Z}(Q,T,Z)},
	\end{align}
 
	Finally recall that from the proof of Lemma \ref{lem:psi_identification}, we have shown that $\ee{\{\nu(T)-\theta\}/G(T|Z)} = 0$. In addition, 
	$\E\left\{\varphi_2(Q,T,Z)\right\} = 0$ by Proposition \ref{prop:int_dM}.
	Therefore,
	\begin{align}
	\ee{\beta\left\{\frac{\nu(T) - \theta}{G(T|Z)} + \varphi_2(Q,T,Z)\right\}} = 0.
	\end{align}
	Thus, 
	\begin{align}
	\varphi(O) &= \varphi(Q,T,Z)
	= \beta \left\{\frac{\nu(t)-\theta}{G(t|Z)} +
	\varphi_2(Q,T,Z) \right\}\\
	& = \beta \left\{\frac{\nu(t)-\theta}{G(t|Z)} - \int_0^\infty \frac{m(v,Z;F) - \theta F(v|Z)}{G(v|Z)\{1-F(v|Z)\}} d\bar M_Q(v;G) \right\}
	\end{align}
	is an influence curve.
  
\end{proof}

\bigskip 
 \begin{proof}[Proof of Proposition \ref{thm:EIF}]
	We will show that the influence curve $\varphi$ derived in Lemma \ref{lem:IF} lies in the tangent space $\dot P$.
	Therefore, $\varphi$ is the efficient influence curve.
	
	Recall that the influence curve is
	\begin{align}
	\varphi(O;\theta, F,G) 
	&= \beta \left [\frac{\nu(T)-\theta}{G(T|Z)} - \int_0^\infty \frac{m(v,Z;F) - \theta F(v|Z)}{G(v|Z)\{1-F(v|Z)\}} \{d\bar N_Q(v) - d\bar A_Q(v;G) \} \right] \\
	&= \beta \left[ \frac{\nu(T)-\theta}{G(T|Z)} + \frac{m(Q,Z;F) - \theta F(Q|Z)}{G(Q|Z)\{1-F(Q|Z)\}} \right.\\
	& \quad\quad\quad\quad  \left.  - \int_0^\infty \frac{ \{ m(v,Z;F) - \theta F(v|Z) \} \mathbbm{1}(Q\leq v<T)}{G(v|Z)^2\{1-F(v|Z)\}} d G(v|Z) \right] \\
    &= \beta \left[ \frac{\nu(T)-\theta}{G(T|Z)} + \frac{m(Q,Z;F) - \theta F(Q|Z)}{G(Q|Z)\{1-F(Q|Z)\}}  - \int_Q^T \frac{ m(v,Z;F) - \theta F(v|Z) }{G(v|Z)^2\{1-F(v|Z)\}} d G(v|Z) \right]. 
	\end{align}
	Denote 
	\begin{align}
	C(t, Z) = \int_0^t \frac{m(v,Z;F) - \theta F(v|Z)}{G(v|Z)^2\{1-F(v|Z)\}} d G(v|Z).
	\end{align}
	Then
	\begin{align}
	\varphi(O;\theta, F,G, H) 
	& = \beta \left\{\frac{\nu(T)-\theta}{G(T|Z)} - C(T,Z) \right\} + \beta\left\{\frac{m(Q,Z;F) - \theta F(Q|Z)}{G(Q|Z)\{1-F(Q|Z)\}} + C(Q,Z) \right\}\\
	&= a(T,Z) + b(Q,Z),
	\end{align}
	where
	\begin{align}
	a(T,Z) &= \beta \left\{\frac{\nu(T)-\theta}{G(T|Z)} - C(T,Z) \right\},\\
	b(Q,Z) & = \beta\left\{\frac{m(Q,Z;F) - \theta F(Q|Z)}{G(Q|Z)\{1-F(Q|Z)\}} + C(Q,Z) \right\}.
	\end{align}
	Since $\E\left\{\varphi(O;\theta, F,G,H)\right\} = 0$, we have $\E\{a(T,Z)\} + \E\{b(Q,Z)\} =0$, so
	\begin{align}
	\varphi(O;\theta, F,G,H) &= \left[ a(T,Z) - \E\left\{a(T,Z)\right\} \right] + \left[ b(Q,Z) -\E\left\{b(Q,Z)\right\} \right]\\
	& \in L_2^0(P_{T,Z})+L_2^0(P_{Q,Z}) = \dot P.
	\end{align}
	Therefore, by Theorem 4.3 on page 67 of \citet{tsiatis2006semiparametric},  the influence curve $\varphi$ given in \eqref{eq:EIC} is the efficient influence curve, and thus, the semiparametric efficiency bound for estimating $\theta$ is $\E(\varphi^2)$.
  
\end{proof}

\newpage
\section{Double robustness of the estimating function}

 \begin{proof}[Proof of Theorem \ref{thm:DR}]
	Recall that the estimating function is
	\begin{align}
	U(O; \theta, F, G)
	=& \frac{\nu(T)-\theta}{G(T|Z)} - \int_0^\infty \frac{m(v,Z;F) - \theta F(v|Z)}{G(v|Z)\{1-F(v|Z)\}} d\bar M_Q(v; G).  \label{eq:eq1}
	\end{align}
	1) We will first show that $\E\left\{U(O; \theta_0, F, G_0)\right\} = 0$.
 Since  
 $\{\bar M_Q(t;G_0)\}_{t\geq 0}$ is a backwards martingale with respect to the filtration $\{\mathcal{\bar F}_t\}_{t\geq 0}$. 
	By Assumption \ref{ass:overlap}, the integrand for the integral in \eqref{eq:eq1} is bounded on $[0, \tau_2]$. Therefore, by Proposition \ref{prop:int_dM}, 
	\begin{align}
	\ee{\int_0^\infty \frac{m(v,Z;F) - \theta F(v|Z)}{G_0(v|Z)\{1-F(v|Z)\}} d\bar M_Q(v; G_0)} = \ee{\int_0^{\tau_2} \frac{m(v,Z;F) - \theta F(v|Z)}{G_0(v|Z)\{1-F(v|Z)\}} d\bar M_Q(v; G_0)} = 0.
	\end{align}
	In addition, we have shown that $\ee{\{\nu(T) - \theta_0\}/G_0(T|Z)} = 0$ in the proof of Lemma \ref{lem:psi_identification}. Thus, 
	\begin{align}
	\E\left\{U(O;\theta_0, F,G_0)\right\} = 0.
	\end{align}
	
	\bigskip
	2)	We will next show that $\E\left\{U (O; \theta_0, F_0, G)\right\} = 0$. 
	We begin by rewriting the estimating function:
	\begin{align}
	U(O; \theta, F, G) 
	& = \frac{\nu(T) - \theta}{G(T|Z)} -  \int_0^\infty \frac{m(v,Z; F) - \theta F(v|Z)}{G(v|Z)\{1-F(v|Z)\}} \{d\bar N_Q(v) - d\bar A_Q(v;G) \} \nonumber \\
	& = \frac{\nu(T) - \theta}{G(T|Z)} + \frac{m(Q,Z; F) - \theta F(Q|Z)}{G(Q|Z)\{1-F(Q|Z)\}} \nonumber \\
	&\quad\quad\quad - \int_Q^T \frac{m(v,Z; F) - \theta F(v|Z)}{1-F(v|Z)} 
	\cdot \frac{d G(v|Z)}{G(v|Z)^2}.  \label{eq:DR_proof_2}
	\end{align}
	Note that 
	\begin{align}
	\int_Q^T \frac{ dG(v|Z) }{G(v|Z)^2}
	= \frac{1}{G(Q|Z)} - \frac{1}{G(T|Z)}. 
	\end{align}
	Therefore 
	\begin{align}
	\frac{\nu(T)-\theta}{ G(T|Z)}  = \frac{\nu(T)-\theta}{ G(Q|Z)}  - \{\nu(T)-\theta\} \int_Q^T \frac{ dG(v|Z) }{G(v|Z)^2}
	\label{eq:2expr_10}.
	\end{align}
	Replacing the first term $\{\nu(T)-\theta\}/G(T|Z)$ in \eqref{eq:DR_proof_2} by the right hand side of \eqref{eq:2expr_10} and after some algebra, we have $ U(O;\theta,F,G) = \A_1-\A_2$, where 
	\begin{align}
	\A_1 & = \frac{\nu(T)-\theta}{G(Q|Z)} + \frac{m(Q,Z; F) - \theta F(Q|Z)}{G(Q|Z)\{1-F(Q|Z)\}}, \\
	\A_2 & = \int_Q^T \left[\{\nu(T)-\theta\}+ \frac{m(v,Z; F)- \theta F(v|Z)}{1-F(v|Z)} \right] \frac{dG(v|Z)}{G(v|Z)^2}.
	\end{align}
	
	In the following, we will first compute $\E\{U(O;\theta,F_0,G)|Q,Z\}$ and then show that 
 $\E\{U(O;\theta_0,F_0,G)\} = \ee{\E\{U(O;\theta_0,F_0,G)|Q,Z\}}=0$. 
	The key quantities involved are \\
	(a) $\E\{\nu(T)|Q,Z\}$, (b) $\E\{\mathbbm{1}(T>v)|Q,Z\}$ and (c) $\E\{\nu(T)\mathbbm{1}(T>v)|Q,Z\}$ for $v\geq Q$.
	
	Recall that
	\begin{align}
	\tilde m(v,Z;F) = \int_v^\infty \nu(t) dF(t|Z), \quad \mu(Z;F) = \int_0^\infty \nu(t) dF(t|Z).
	\end{align}
	So $m(v,Z; F) + \tilde m(v,Z; F) = \mu(Z;F)$ for all $v>0$. 
	
	Recall the observed data density of $T$ given $(Q,Z)$ from \eqref{eq:p(t|q,z)}:
	\begin{align}
	p_{T|Q,Z}(t|q,z) = \frac{\mathbbm{1}(q<t)}{1-F_0(q|z)} dF_0(t|z).
	\end{align}
	We have 
	\begin{align}
	&\quad \mbox{(a) } \E\left\{\left.\nu(T) \right| Q,Z\right\} 
	=\int_0^\infty \nu(t)   \frac{\mathbbm{1}(Q<t)}{1-F_0(Q|Z)} dF_0(t|Z)
	= \frac{\tilde m(Q,Z;F_0)}{1-F_0(Q|Z)}.
	\end{align}
	Moreover, for $v\geq Q$,
	\begin{align}
	\mbox{(b) } \E\left\{\left.\mathbbm{1}(T>v)\right| Q,Z\right\} 
	&= \int_0^\infty \mathbbm{1}(t>v) \frac{\mathbbm{1}(Q<t)}{1-F_0(Q|Z)}\ dF
	_0(t|Z) \\
	& = \frac{1}{1-F_0(Q|Z)} \int_v^\infty dF_0(t|Z) \\
	& = \frac{1-F_0(v|Z)}{1-F_0(Q|Z)},
	\end{align}
	and
	\begin{align}
	\quad \mbox{(c) } \E\left\{\left.\nu(T) \mathbbm{1}(T>v) \right| Q,Z\right\} 
	&=\int_0^\infty \nu(t) \mathbbm{1}(t>v)  \frac{\mathbbm{1}(Q<t)}{1-F_0(Q|Z)} dF_0(t|Z)\\
	&= \frac{1}{1-F_0(Q|Z)} \int_v^\infty \nu(t) dF_0(t|Z)\\
	&= \frac{\tilde m(v,Z,F_0)}{1-F_0(Q|Z)}.
	\end{align}
	
	Therefore, when $F = F_0$,
	\begin{align}
	\E(A_1|Q,Z) 
	& = \left\{ \frac{\tilde m(Q,Z;F_0)}{1-F_0(Q|Z)} - \theta \right\}  \frac{1}{G(Q|Z)}
	+ \frac{m(Q,Z; F_0) - \theta F_0(Q|Z)}{G(Q|Z)\{1-F_0(Q|Z)\}}\\
	& = \frac{\mu(Z;F_0) - \theta}{G(Q|Z)\{1-F_0(Q|Z)\}}, 
	\end{align}
	and 
	\begin{align}
	\E(A_2|Q,Z)
	& = \int_Q^\infty \left\{\frac{\tilde m(v,Z,F_0)}{1-F_0(Q|Z)}-\theta \frac{1-F_0(v|Z)}{1-F_0(Q|Z)}\right. \\
	& \quad\quad\quad\quad\quad\quad \left. + \frac{m(v,Z; F_0)- \theta F_0(v|Z)}{1-F_0(v|Z)}\cdot \frac{1-F_0(v|Z)}{1-F_0(Q|Z)}\right\} \frac{dG(v|Z)}{G(v|Z)^2} \\
	& = \int_Q^\infty \frac{\mu(Z;F_0) - \theta}{1-F_0(Q|Z)}  
	\cdot \frac{dG(v|Z)}{G(v|Z)^2} \\
	& = \frac{\mu(Z;F_0) - \theta}{1-F_0(Q|Z)} \int_Q^\infty \frac{dG(v|Z)}{G(v|Z)^2} \\
	& = \frac{\mu(Z;F_0) - \theta}{1-F_0(Q|Z)} \left\{\frac{1}{G(Q|Z)} - 1 \right\} \label{eq:DR_proof_1},
	\end{align}
	where \eqref{eq:DR_proof_1} holds because $\lim_{v\to\infty} G(v|z) = 1$ for all $z$. 
	
	Combining the above we have
	\begin{align}
	\E\left\{U(O;\theta,F_0, G)|Q,Z\right\}
	= \frac{\mu(Z;F_0) - \theta}{1-F_0(Q|Z)}.
	\end{align}
	Recall the observed data density of $(Q,Z)$ from \eqref{eq:p(q,z)}:
	\begin{align}
	p_{Q,Z}(q,z) = \frac{1-F_0(q|z)}{\beta} g_0(q|z)h_0(z).
	\end{align}
	So we have
	\begin{align}
	\E\left\{\frac{\mu(Z;F_0)- \theta}{1-F_0(Q|Z)}\right\} 
	&= \int_{\cal Z} \int_0^\infty \frac{\mu(z;F_0)- \theta}{1-F_0(q|z)} \cdot \frac{1-F_0(q|z)}{\beta_0} \ dG_0(q|z)\ dH_0(z)
	= \frac{\theta_0- \theta}{\beta_0}.
	\end{align}
	Therefore, when $\theta = \theta_0$, we have that
	$\E\left\{U(O;\theta_0,F_0,G)\right\} = 0$.
  
\end{proof}

\newpage
\section{Asymptotics} \label{app:asymptotics}

\subsection{Proofs}\label{sec:asymptotic_proofs}


We assume without loss of generality (WLOG) that $\hat F(t|Z)\equiv 0$ for all $t<\tau_1$, and 
$\hat G(t|Z)\equiv 1$ for all $t>\tau_2$. 

Therefore $m(t,Z; \hat F) = m(t,Z; F^\divideontimes)\equiv 0$ for all $t<\tau_1$, 
since $F^\divideontimes$ and $G^\divideontimes$ 
satisfy Assumption \ref{ass:overlap}. Furthermore, for $F=F^\divideontimes$ or $\hat F$, and $G=G^\divideontimes$ or $\hat G$, we have
\begin{align}
    &\int_0^\infty \frac{m(v,Z; F) - \theta_0 F(v|Z)}{1 - F(v|Z)}\mathbbm{1}(Q\leq v<T) \frac{d G(v|Z)}{G(v|Z)^2} \\
    =&  \int_{\tau_1}^{\tau_2} \frac{m(v,Z; F) - \theta_0 F(v|Z)}{1-F(v|Z)}\mathbbm{1}(Q\leq v<T) \frac{d G(v|Z)}{G(v|Z)^2}.\label{eq:asymp_int_mF} \\ 
\end{align}
We have also
\begin{align}
    &\left|m(Q,Z; \hat F) - m(Q,Z; F^\divideontimes)\right| \leq \sup_{t\in[\tau_1,\tau_2]}\left|m(t,Z; \hat F) - m(t,Z; F^\divideontimes)\right|, \ \ a.s., \label{eq:asymp_mQ}\\
    &\left|\hat F(Q|Z) - F^\divideontimes(Q|Z)\right| \leq \sup_{t\in[\tau_1,\tau_2]}\left|\hat F(t|Z) -  F^\divideontimes(t|Z)\right|, \ \ a.s.. \label{eq:asymp_FQ}
\end{align}

\bigskip
The following asymptotic results of $U$-statistics from \citet{van2000asymptotic} are used in proof of Theorem \ref{thm:mdr} for the asymptotic normality of $\hat\theta_{dr}$.
Let $X_1,...,X_n$ be a random sample from an unknown distribution. Given a known function $h$, consider the estimation of the parameter $\theta = \E h(X_1,...,X_r)$. A $U$-statistic with kernel $h$ is defined as
\begin{align}
U = \frac{1}{{n \choose r}} \sum_{\gamma} h(X_{\gamma_1},...,X_{\gamma_r}),
\end{align}
where the sum is taken over the set of all unordered subsets $\gamma$ of $r$ different integers chosen from $\{1,...,n\}$. The projection of $U-\theta$ onto the set of all statistics of the form $\sum_{i=1}^n g_i(X_i)$ is given by
\begin{align}
\hat U = \sum_{i=1}^n \E(U-\theta|X_i) = \frac{r}{n}\sum_{i=1}^n h_1(X_i),
\end{align}
where the function $h_1$ is given by
\begin{align}
h_1(x) & = \E h(x,X_2,...,X_r) - \theta. 
\end{align}
\begin{theorem}[Theorem 12.3 in \citet{van2000asymptotic}]\label{lem:Ustats}
	If $\E h^2(X_1,...,X_r)<\infty$, then $n^{1/2} (U- \theta - \hat U)\convp 0$. Consequently, the sequence $n^{1/2}(U-\theta)$ is asymptotically normal with mean zero and covariance $r^2\xi_1$, where, with $X_1,...,X_r, X_1',...,X_r'$ denoting i.i.d. variables, 
 $$\xi_1 = \text{cov}\{h(X_1,X_2,...,X_r), h(X_1,X_2',...,X_r')\}.$$
\end{theorem}

\medskip
The following  is a technical regularity assumption in order to apply Theorem \ref{lem:Ustats}. 
From the asymptotic linearity of $\hat F$ and $\hat G$ in Assumption \ref{assump:if}, 
it can be shown that for each $(v,z)\in[\tau_1,\tau_2]\times \Z$, $\int_0^v\{\nu(t)-\theta\} d\hat F(t|Z) / \{1-\hat F(v|Z)\}$ is asymptotically linear for $\int_0^v\{\nu(t)-\theta\} dF^\divideontimes(t|Z)/\{1-F^\divideontimes(v|Z)\}$ with influence function $\phi_1(t,z,O;F^\divideontimes)$, 
and $1/\hat G(v|Z)$ is asymptotically linear for $1/G^\divideontimes$ with influence function $\phi_2(t,z,O;G^\divideontimes)$. In addition, denote 
\begin{align}
\tilde R_{1}(v,z) &= \frac{\int_0^v\{\nu(t)-\theta\} d\hat F(t|z) }{1-\hat F(v|z)} - \frac{\int_0^v\{\nu(t)-\theta\} dF^\divideontimes(t|z)}{1-F^\divideontimes(v|z)} -  \frac{1}{n}\sum_{i=1}^n\phi_1(v,z,O_i;F^\divideontimes),\label{eq:tilde_R1}\\
\tilde R_{2}(v,z) & = \frac{1}{\hat G(v|z)} - \frac{1}{G^\divideontimes(v|z)} - \frac{1}{n}\sum_{i=1}^n \phi_2(v,z,O_i; G^\divideontimes). \label{eq:tilde_R2}
\end{align}
We have that $\|\tilde R_1(\cdot,Z)\|_{\sup,2} = o(n^{-1/2})$, $\|\tilde R_2(\cdot,Z)\|_{\sup,2} = o(n^{-1/2})$, and either $\|\tilde R_1(\cdot,Z)\|_{\TV,2} = o(1)$ or $\|\tilde R_2(\cdot,Z)\|_{\TV,2} = o(1)$ (depending on whether $\|R_1(\cdot,Z)\|_{\TV,2} = o(1)$ or $\|R_2(\cdot,Z)\|_{\TV,2} = o(1)$ in Assumption \ref{assump:if}, respectively).

\bigskip
\begin{assumption}\label{ass:U_stat}
    Suppose that 
	\begin{align}
	&\E\left[\left\{\int_{0}^{\infty} \phi_1(v,Z_1,O_2; F^\divideontimes) \frac{d\bar M_{Q,1}(v; G^\divideontimes)}{G^\divideontimes(v|Z_1)}  
	+\int_{0}^{\infty} \phi_1(v,Z_2,O_1; F^\divideontimes) \frac{d\bar M_{Q,2}(v; G^\divideontimes)}{G^\divideontimes(v|Z_2)}\right\}^2\right] < \infty, \\
	&\E\left\{\left(\int_{0}^{\infty} \phi_2(v,Z_1,O_2; G^\divideontimes)\mathbbm{1}(Q_1\leq v<T_1) \ d\left[\frac{\int_0^v\{\nu(t)-\theta\} dF^\divideontimes(t|Z_1)}{1-F^\divideontimes(v|Z_1)} \right] \right.\right.\\
	&\quad\quad\quad\quad \left.\left.
	+\int_{0}^{\infty} \phi_2(v,Z_2,O_1; G^\divideontimes) \mathbbm{1}(Q_2\leq v<T_2) \ d\left[\frac{\int_0^v\{\nu(t)-\theta\} dF^\divideontimes(t|Z_2)}{1-F^\divideontimes(v|Z_2)} \right] \right)^2\right\} < \infty.
	\end{align}
\end{assumption}

\newpage
 \begin{proof}[Proof of Theorem \ref{thm:mdr}]
	(1) First, we show the consistency of $\hat\theta_{dr}$, i.e., $\hat\theta_{dr} - \theta_0 \convp 0$.
    
    We will use expression \eqref{eq:consistency_proof_2} of $\hat\theta_{dr} - \theta_0$ below: 
    \begin{align}
	\hat\theta_{dr} - \theta_0
	&=  \left[\sum_{i=1}^n \left\{\frac{1}{\hat G(T_i|Z_i)} - \int_0^\infty \frac{\hat F(v|Z_i)}{1-\hat F(v|Z_i)} \cdot \frac{d\bar M_{Q,i}(v; \hat G)}{\hat G(v|Z_i)} \right\} \right]^{-1}\\
	&\quad\quad\quad\quad \times \left[ \sum_{i=1}^n \left\{\frac{\nu(T_i) - \theta_0}{\hat G(T_i|Z_i)} - \int_0^\infty \frac{m(v,Z_i;\hat F) - \theta_0\hat F(v|Z_i)}{1-\hat F(v|Z_i)} \cdot \frac{d\bar M_{Q,i}(v; \hat G)}{\hat G(v|Z_i)} \right\} \right] 
	\\
	& = \left. \left\{\frac{1}{n}\sum_{i=1}^n U_{i}(\theta_0,\hat F,\hat G) \right\} \right/\left\{\frac{1}{n}\sum_{i=1}^n D_{i}(\hat F,\hat G) \right\}, \label{eq:consistency_proof_2}
	\end{align}
 where $U_i$'s are i.i.d.~copies of $U$ defined in \eqref{eq:score}, and 
 \begin{align}
	D_{i}(\hat F, \hat G) & = \frac{1}{\hat G(T_i|Z_i)} - \int_0^\infty \frac{\hat F(v|Z_i)}{1-\hat F(v|Z_i)} \frac{ d\bar M_{Q,i}(v; \hat G)}{\hat G(v|Z_i)}.  \label{eq:D_i}
	\end{align}
	We will show that (i) the numerator in \eqref{eq:consistency_proof_2} converges to 0 in probability and (ii) the denominator converges to $1/\beta_0$ in probability. 
	For (i), we consider the decomposition \eqref{eq:consistency_proof_3} and show that the leading term $\A_1$ converges to zero in probability by the law of large numbers, and the remainder term $ \A_2$ is $o_p(1)$. With the same proof techniques, we will show (ii). 
	
	(i) First consider $n^{-1}\sum_{i=1}^n U_{i}(\theta_0, \hat F, \hat G)$. 
	\begin{align}
	\frac{1}{n}\sum_{i=1}^n U_{i}(\theta_0, \hat F, \hat G) 
	& = \A_1+\A_2, \label{eq:consistency_proof_3}
	\end{align}
	where
	\begin{align}
	\A_1 & = \frac{1}{n}\sum_{i=1}^n U_{i}(\theta_0,F^\divideontimes, G^\divideontimes), \\
	\A_2 & = \frac{1}{n}\sum_{i=1}^n \left\{U_{i}(\theta_0,\hat F, \hat G) - U_{i}(\theta_0, F^\divideontimes, G^\divideontimes) \right\}.
	\end{align}
	By law of large numbers, the first term
	\begin{align}
	\A_1 = \frac{1}{n}\sum_{i=1}^n U_{i}(\theta_0, F^\divideontimes, G^\divideontimes) \convp \E\left\{U(\theta_0, F^\divideontimes, G^\divideontimes)\right\} =0 
	\end{align}
	when $G^\divideontimes = G_0$ or $F^\divideontimes = F_0$
	by Theorem \ref{thm:DR}. 
	
	Next, we will show that $\A_2 = o_p(1)$. We consider the decomposition 
	\begin{align}
	\A_2 & = \A_{21} + \A_{22} + \A_{23},
	\end{align}
	where
	\begin{align}
	\A_{21} & =  \frac{1}{n} \sum_{i=1}^n \left\{\frac{\nu(T_i) - \theta_0}{\hat G(T_i|Z_i)} - \frac{\nu(T_i) - \theta_0}{G^\divideontimes(T_i|Z_i)} \right\} ,\\
	\A_{22} & = \frac{1}{n} \sum_{i=1}^n \left[\frac{m(Q_i,Z_i;\hat F) - \theta_0\hat F(Q_i|Z_i)}{\hat G(Q_i|Z_i)\{1-\hat F(Q_i|Z_i)\}} - \frac{m(Q_i,Z_i;F^\divideontimes) - \theta_0 F^\divideontimes(Q_i|Z_i)}{G^\divideontimes(Q_i|Z_i)\{1-F^\divideontimes(Q_i|Z_i)\}} \right] ,\\
	\A_{23} & = \frac{1}{n} \sum_{i=1}^n \left[\int_0^\infty \frac{m(v,Z_i;\hat F) - \theta_0 \hat F(v|Z_i)}{1-\hat F(v|Z_i)}\mathbbm{1}(Q_i\leq v<T_i) \frac{d\hat G(v|Z_i)}{\hat G(v|Z_i)^2} \right. \\
	&\quad\quad\quad\quad\quad  \left.- \int_0^\infty \frac{m(v,Z_i;F^\divideontimes) - \theta_0 F^\divideontimes(v|Z_i)}{1-F^\divideontimes(v|Z_i)} \mathbbm{1}(Q_i\leq v<T_i) \frac{dG^\divideontimes(v|Z_i)}{G^\divideontimes(v|Z_i)^2} \right]. \\
	\end{align}
	We will show that each term is $o_p(1)$. First consider $\A_{21}$. 
	\begin{align}
	\left|\A_{21}\right| 
	& \leq\frac{1}{n}\sum_{i=1}^n\left|\frac{\nu(T_i)}{\hat G(T_i|Z_i)} - \frac{\nu(T_i)}{G^\divideontimes(T_i|Z_i)} \right| 
	= \frac{1}{n}\sum_{i=1}^n \left|\frac{\nu(T_i)}{G^\divideontimes(T_i|Z_i)\hat G(T_i|Z_i)}\right| \left|\hat G(T_i|Z_i) - G^\divideontimes(T_i|Z_i)\right|.
	\end{align}
	By Cauchy-Schwartz inequality, 
	\begin{align}
	\E\left\{|\A_{21}|\right\}
	& \leq \E\left\{\left|\frac{\nu(T)}{G^\divideontimes(T|Z) \hat G(T_i|Z_i)}\right| \left|G^\divideontimes(T_i|Z_i) -\hat G(T_i|Z_i)\right|\right\}\\
	&\leq  \E\left\{\left|\frac{\nu(T)}{G^\divideontimes(T|Z)\hat G(T_i|Z_i)}\right|^2\right\}^{1/2} \E\left\{\left|\hat G(T_i|Z_i) - G^\divideontimes(T_i|Z_i)\right|^2\right\}^{1/2} \\
	&\lesssim  \left\|\hat G(t|Z) - G^\divideontimes(t|Z)\right\|_{\sup, 2} \\
	& = o(1).
	\end{align}
	Therefore, by Markov's inequality, $\A_{21} = o_p(1)$.
	
	Now consider $\A_{22}$.
	\begin{align}
	|\A_{22}|
	& \leq \frac{1}{n} \sum_{i=1}^n \left|\frac{m(Q_i,Z_i;\hat F) - \theta_0\hat F(Q_i|Z_i)}{\hat G(Q_i|Z_i)\{1-\hat F(Q_i|Z_i)\}} - \frac{m(Q_i,Z_i;F^\divideontimes) - \theta_0 F^\divideontimes(Q_i|Z_i)}{G^\divideontimes(Q_i|Z_i)\{1-F^\divideontimes(Q_i|Z_i)\}} \right| \\
	&\leq \A_{221} + \A_{222}.
	\end{align}
 where
 \begin{align}
     \A_{221} & = \frac{1}{n} \sum_{i=1}^n \left|\frac{1}{\hat G(Q_i|Z_i)} \left\{\frac{m(Q_i,Z_i;\hat F) - \theta_0\hat F(Q_i|Z_i)}{1-\hat F(Q_i|Z_i)} - \frac{m(Q_i,Z_i;F^\divideontimes) - \theta_0 F^\divideontimes(Q_i|Z_i)}{1-F^\divideontimes(Q_i|Z_i)} \right\} \right|, \\
     \A_{222} & = \frac{1}{n} \sum_{i=1}^n \left|\frac{m(Q_i,Z_i;F^\divideontimes) - \theta_0 F^\divideontimes(Q_i|Z_i)}{1-F^\divideontimes(Q_i|Z_i)} \left\{\frac{1}{\hat G(Q_i|Z_i)} - \frac{1}{G^\divideontimes(Q_i|Z_i)} \right\} \right|.
 \end{align}
Integrating by parts we have
	\begin{align}
	m(v,Z; F) 
	= \int_0^v \nu(t) d F(t|Z) 
	= \nu(v)F(v|Z) - \int_0^v F(t|Z) d\nu(t)
	\end{align}
 because $F(0|Z) = 0$.
	Therefore, 
	\begin{align}
	&\quad \sup_{v\in[\tau_1,\tau_2]} \left| m(v,Z;\hat F)  - m(v,Z;F^\divideontimes) \right| \\
    & \leq \sup_{v\in[\tau_1,\tau_2]} \left| \nu(v)\hat F(v|Z) - \nu(v)F^\divideontimes(v|Z)\right| 
	+ \sup_{v\in[\tau_1,\tau_2]} \int_0^v \left|\hat F(t|Z) - F^\divideontimes(t|Z) \right| d\nu(t) \\
	& =\sup_{v\in[\tau_1,\tau_2]} \left| \nu(v)\hat F(v|Z) - \nu(v)F^\divideontimes(v|Z)\right| 
	+ \sup_{v\in[\tau_1,\tau_2]} \int_{\tau_1}^v \left|\hat F(t|Z) - F^\divideontimes(t|Z) \right| d\nu(t) \\
	& \lesssim \sup_{v\in[\tau_1,\tau_2]} \left|\hat F(v|Z)  - F^\divideontimes(v|Z)\right|. \label{eq:mF_hat-mF}
	\end{align}
	In addition, by \eqref{eq:asymp_mQ} and \eqref{eq:asymp_FQ}, 
\begin{align}
    &\quad \E\left(|\A_{221}|\right) \\
    & \leq \ee{\left|\frac{1}{\hat G(Q|Z)\{1- \hat F(Q|Z)\}} \right|  \left|\{m(Q,Z;\hat F) - \theta_0\hat F(Q|Z)\}- \{m(Q,Z;F^\divideontimes) - \theta_0 F^\divideontimes(Q|Z)\}\right|} \\
    &\quad + \ee{\left|\frac{m(Q,Z;F^\divideontimes) - \theta_0 F^\divideontimes(Q|Z)}{\hat G(Q|Z)\{1-\hat F(Q|Z)\}\{1-F^\divideontimes(Q|Z)\}} \right|  \left|F^\divideontimes(Q|Z) - \hat F(Q|Z) \right|} \\
    &\lesssim \E\left\{ \sup_{v\in[\tau_1,\tau_2]} \left|\hat F(v|Z)  - F^\divideontimes(v|Z)\right| \right\} + \E\left\{\sup_{v\in[\tau_1,\tau_2]} \left| m(v,Z;\hat F)  - m(v,Z;F^\divideontimes) \right|\right\} \\
     &\lesssim \E\left\{ \sup_{v\in[\tau_1,\tau_2]} \left|\hat F(v|Z)  - F^\divideontimes(v|Z)\right| \right\} 
\end{align}
By Cauchy-Schwartz inequality,
\begin{align}
    \E\left\{ \sup_{v\in[\tau_1,\tau_2]} \left|\hat F(v|Z)  - F^\divideontimes(v|Z)\right| \right\} 
    \leq \left\|\hat F(t|Z) - F^\divideontimes(t|Z)\right\|_{\sup, 2} = o(1), 
\end{align}
so $\E\left(|\A_{221}|\right) = o(1)$.
Therefore, by Markov's inequality, $\A_{221} = o_p(1)$. Likewise, we can show that $\A_{222} = o_p(1)$. Thus, $A_{22} = A_{221}+A_{222} = o_p(1)$.

Next, consider $\A_{23}= \A_{231} + \A_{232}$, where
 \begin{align}
     \A_{231} & = \frac{1}{n} \sum_{i=1}^n \int_{Q_i}^{T_i} \left\{\frac{m(v,Z_i;\hat F) - \theta_0\hat F(v|Z_i)}{1-\hat F(v|Z_i)} - \frac{m(v,Z_i;F^\divideontimes) - \theta_0 F^\divideontimes(v|Z_i)}{1-F^\divideontimes(v|Z_i)} \right\} \frac{d\hat G(v|Z_i)}{\hat G(v|Z_i)^2}, \\
     \A_{232} & = \frac{1}{n} \sum_{i=1}^n \int_{Q_i}^{T_i} \frac{m(v,Z_i;F^\divideontimes)- \theta_0 F^\divideontimes(v|Z_i)}{1-F^\divideontimes(v|Z_i)} \  d \left\{\frac{1}{\hat G(v|Z_i)} - \frac{1}{G^\divideontimes(v|Z_i)}\right\}.
 \end{align}
 Since $\hat G(t|Z)$ is a non-decreasing function a.s., by \eqref{eq:asymp_int_mF} \eqref{eq:mF_hat-mF}
 and Cauchy-Schwartz inequality, 
	\begin{align}
	\E(|\A_{231}|) 
    &\leq \E\left\{\int_{Q}^{T} \left|\frac{m(v,Z;\hat F) - \theta_0\hat F(v|Z)}{1-\hat F(v|Z)} - \frac{m(v,Z;F^\divideontimes) - \theta_0 F^\divideontimes(v|Z)}{1-F^\divideontimes(v|Z)} \right|  \frac{d\hat G(v|Z)}{\hat G(v|Z)^2}\right\}\\
    &= \E\left\{\int_{Q\vee\tau_1}^{T\wedge\tau_2} \left|\frac{m(v,Z;\hat F) - \theta_0\hat F(v|Z)}{1-\hat F(v|Z)} - \frac{m(v,Z;F^\divideontimes) - \theta_0 F^\divideontimes(v|Z)}{1-F^\divideontimes(v|Z)} \right|   \frac{d\hat G(v|Z)}{\hat G(v|Z)^2}\right\}\\
    &\lesssim \E\left\{ \sup_{v\in[\tau_1,\tau_2]} \left|\hat F(v|Z)  - F^\divideontimes(v|Z)\right| \right\} + \E\left\{ \sup_{v\in[\tau_1,\tau_2]} \left|m(v,Z;\hat F)  - m(v,Z;F^\divideontimes)\right| \right\} \\
    &\lesssim \E\left\{ \sup_{v\in[\tau_1,\tau_2]} \left|\hat F(v|Z)  - F^\divideontimes(v|Z)\right| \right\} \\
	& \leq \left\|\hat F(t|Z) - F^\divideontimes(t|Z)\right\|_{\sup, 2} \\
	& = o(1).
	\end{align}
    So $\A_{231} = o_p(1)$ by Markov's inequality.
In addition, with integration by parts, we can show that $\A_{232} = o_p(1)$ with a similar argument as above. Therefore, $\A_{23} = \A_{232} + \A_{232} = o_p(1)$.
 
	Thus, we have that $\A_{2} = o_p(1)$. 
	
	(ii) Likewise, we can show that
	\begin{align}
	\frac{1}{n}\sum_{i=1}^n D_{i}(\hat F, \hat G) = \frac{1}{n}\sum_{i=1}^n D_{i}(F^\divideontimes, G^\divideontimes) + o_p(1). 
	\end{align}
	
	Similarly to the proof of Theorem \ref{thm:DR}, we can show that $\E\{D_i(F_0, G^\divideontimes)\} = 1/\beta_0$ and $\E\{D_i(F^\divideontimes, G_0)\} = 1/\beta_0$. 
 Therefore, when $F^\divideontimes = F_0$ or $G^\divideontimes = G_0$, $n^{-1}\sum_{i=1}^n D_{i}(F^\divideontimes, G^\divideontimes)  \convp 1/\beta_0$ by the law of large numbers. Thus, $n^{-1}\sum_{i=1}^n D_{i}(\hat F, \hat G)  \convp 1/\beta_0$.
	
	Above all, if $F^\divideontimes = F_0$ or $G^\divideontimes = G_0$, we have that $\hat\theta_{dr} - \theta_0 \convp 0$, 
	that is, $\hat\theta_{dr} \convp \theta_0$. 
	
	\bigskip
	(2) We now show the asymptotic normality of $\hat\theta_{dr}$ with the additional Assumption \ref{assump:if}. 
	For the decomposition \eqref{eq:AN_err_decomp} below, we show that $\B_2$
    is $o_p(n^{-1/2})$ by `counting' small terms using combinatorics.
    The terms $\B_3$ and $\B_4$ can be shown to be asymptotically normal with the help of  $U$-statistics in Theorem \ref{lem:Ustats}.

    Without loss of generality, we consider the case where Assumption \ref{assump:if} holds with $\|R_2(\cdot,Z)\|_{\TV,2} = o(1)$. This implies that $\|\tilde R_2(\cdot,Z)\|_{\TV,2} = o(1)$, where $\tilde R_2$ is defined in \eqref{eq:tilde_R2}. 
	
	In part (1) of the proof, we have shown that 
	\begin{align}
	\hat\theta_{dr} - \theta_0
	& = \left. \left\{\frac{1}{n}\sum_{i=1}^n U_{i}(\theta_0,\hat F,\hat G) \right\} \right/\left\{\frac{1}{n}\sum_{i=1}^n D_{i}(\hat F,\hat G) \right\},
	\end{align}
	and 
	\begin{align}
	\frac{1}{n}\sum_{i=1}^n D_{i}(\hat F,\hat G) \convp \frac{1}{\beta_0}, 
	\end{align}
	where $D_i$ is defined in \eqref{eq:D_i}.
	By Slutsky's Theorem, it suffices to show that $n^{-1/2}\sum_{i=1}^n U_i(\theta_0,\hat F,\hat G)$ converges in distribution to a normal distribution with mean zero. 
	
	Write
	\begin{align}
	\frac{1}{n}\sum_{i=1}^n U_i(\theta_0, \hat F,\hat G)
	&= \B_1+\B_2 + \B_3 + \B_4, \label{eq:AN_err_decomp}
	\end{align}
	where
	\begin{align}
	\B_1 & = \frac{1}{n}\sum_{i=1}^n U_i(\theta_0, F^\divideontimes, G^\divideontimes),\\
	\B_2 &= \frac{1}{n}\sum_{i=1}^n \left[ \left\{ U_i(\theta_0, \hat F,\hat G) -U_i(\theta_0, F^\divideontimes, \hat G)\right\} 
	- \left\{ U_i(\theta_0, \hat F,G^\divideontimes) -U_i(\theta_0, F^\divideontimes, G^\divideontimes)\right\} \right],\\
	\B_3 &= \frac{1}{n}\sum_{i=1}^n \left\{ U_i(\theta_0, \hat F,G^\divideontimes) -U_i(\theta_0, F^\divideontimes, G^\divideontimes)\right\}, \\
	\B_4 & = \frac{1}{n}\sum_{i=1}^n  \left\{ U_i(\theta_0, F^\divideontimes,\hat G) -U_i(\theta_0, F^\divideontimes, G^\divideontimes)\right\}.
	\end{align}
	
	For $\B_1$, when  $F^\divideontimes = F_0$ or $G^\divideontimes = G_0$, we have  $\E\left\{U(\theta_0,F^\divideontimes, G^\divideontimes) \right\} = 0$ by Theorem \ref{thm:DR}. In addition, it can be shown that $\E\left\{U(\theta_0,F^\divideontimes, G^\divideontimes)^2  \right\}<\infty$, so by central limit theorem,
	\begin{align}
	n^{1/2} \B_1 = n^{-1/2}\sum_{i=1}^n U_{i}(\theta_0,F^\divideontimes,G^\divideontimes) \convd N\left(0, \E\left\{U(\theta_0,F^\divideontimes, G^\divideontimes)^2  \right\}\right). 
	\end{align}
 
	Now consider $\B_2$. 
	Recall  $\phi_1$, $\phi_2$, $\tilde R_1$, $\tilde R_2$ as defined earlier in this subsection.
	By \eqref{eq:asymp_int_mF}, 
    we have that 
	\begin{align}
	\B_2 
	& = \frac{1}{n}\sum_{i=1}^n \int_{\tau_1\vee Q_i}^{\tau_2\wedge T_i} \left\{\frac{m(v,Z_i;\hat F) - \theta_0\hat F(v|Z_i)}{1-\hat F(v|Z_i)} - \frac{m(v,Z_i;F^\divideontimes) - \theta_0 F^\divideontimes(v|Z_i)}{1-F^\divideontimes(v|Z_i)}\right\} d\left\{ \frac{1}{G^\divideontimes(v|Z_i)} - \frac{1}{\hat G(v|Z_i)} \right\} \\
    & =  \frac{1}{n}\sum_{i=1}^n \int_{\tau_1\vee Q_i}^{\tau_2\wedge T_i} \left\{\frac{1}{n}\sum_{j=1}^n\phi_1(v,Z_i,O_j; F^\divideontimes) + \tilde R_{1}(v,Z_i)\right\} d\left\{\frac{1}{n}\sum_{k=1}^n \phi_2(v,Z_i,O_k; G^\divideontimes) + \tilde R_{2}(v,Z_i) \right\} \\
	& = \B_{20}+\B_{21},
	\end{align}
	where 
	\begin{align}
	\B_{20} &=  \frac{1}{n}\sum_{i=1}^n \int_{\tau_1\vee Q_i}^{\tau_2\wedge T_i} \left\{\frac{1}{n}\sum_{j=1}^n\phi_1(v,Z_i,O_j; F^\divideontimes) \right\}d\left\{\frac{1}{n}\sum_{k=1}^n \phi_2(v,Z_i,O_k; G^\divideontimes) \right\},\\
	\B_{21} & = \frac{1}{n}\sum_{i=1}^n \int_{\tau_1\vee Q_i}^{\tau_2\wedge T_i} \left\{\frac{1}{n}\sum_{j=1}^n\phi_1(v,Z_i,O_j; F^\divideontimes) \right\}d\tilde R_{2}(v,Z_i) \\
	&\quad +\int_{\tau_1\vee Q_i}^{\tau_2\wedge T_i} \tilde R_{1}(v,Z_i)d\left\{\frac{1}{n}\sum_{k=1}^n \phi_2(v,Z_i,O_k; G^\divideontimes) + \tilde R_{2}(v,Z_i) \right\}.
	\end{align}
	Since we assume that Assumption \ref{assump:if} holds with $\|R_2(\cdot,Z)\|_{\TV,2} = o(1)$,
	we have  $\|\tilde R_1(\cdot,Z)\|_{\sup,2} = o(n^{-1/2})$, $\|\tilde R_2(\cdot,Z)\|_{\sup,2} = o(n^{-1/2})$, and $\|\tilde R_2(\cdot,Z)\|_{\TV,2} = o(1)$. 
 Therefore, 
 $\E\left( |\B_{21}| \right) = o(n^{-1/2})$, so $\B_{21} = o_p(n^{-1/2})$ by Markov's inequality. 
 
	We will now show that $\B_{20} = o_p(n^{-1/2})$. For the simplification of notations, we will suppress the dependency of $\phi_1$ on $F^\divideontimes$ and $\phi_2$ on $G^\divideontimes$.
	\begin{align}
	&\quad \E(|\B_{20}|^2)\\
	& = \frac{1}{n^6}\E\left[ \sum_{i,j,k,i',j',k'=1}^n \int_{\tau_1\vee Q_i}^{\tau_2\wedge T_i} \phi_1(v, Z_{i}, O_{j}) d\phi_2(v, Z_{i}, O_{k}) \int_{\tau_1\vee Q_{i'}}^{\tau_2\wedge T_{i'}} \phi_1(v, Z_{i'}, O_{j'}) d\phi_2(v, Z_{i'}, O_{k'}) \right] \\
	& = \frac{1}{n^6}\E\left[ \sum_{\mathcal{J}} \int_{\tau_1\vee Q_i}^{\tau_2\wedge T_i} \phi_1(v, Z_{i}, O_{j}) d\phi_2(v, Z_{i}, O_{k}) \int_{\tau_1\vee Q_{i'}}^{\tau_2\wedge T_{i'}} \phi_1(v, Z_{i'}, O_{j'}) d\phi_2(v, Z_{i'}, O_{k'}) \right] \\
	& \quad +\frac{1}{n^6}\E\left[ \sum_{\mathcal{K}} \int_{\tau_1\vee Q_i}^{\tau_2\wedge T_i} \phi_1(v, Z_{i}, O_{j}) d\phi_2(v, Z_{i}, O_{k}) \int_{\tau_1\vee Q_{i'}}^{\tau_2\wedge T_{i'}} \phi_1(v, Z_{i'}, O_{j'}) d\phi_2(v, Z_{i'}, O_{k'}) \right],
	\end{align}
	where
	\begin{align}
	\mathcal{J} &= \mathcal{J}_1 \cup \mathcal{J}_2 \cup \mathcal{J}_3 \cup \mathcal{J}_4,\\
	\mathcal{J}_1 &= \{(i,j,k,i',j',k'): j \text{ is different from the other indices} \},\\
	\mathcal{J}_2 &= \{(i,j,k,i',j',k'): k \text{ is different from the other indices} \},\\
	\mathcal{J}_3 &= \{(i,j,k,i',j',k'): j' \text{ is different from the other indices} \},\\
	\mathcal{J}_4 &= \{(i,j,k,i',j',k'): k' \text{ is different from the other indices} \},\\
	\mathcal{K} &= \mathcal{J}^c.
	\end{align}
	We will show that for any $(i,j,k,i',j',k')\in \mathcal{J}$, we have 
	\begin{align}
	\E\left\{\int_{\tau_1\vee Q_i}^{\tau_2\wedge T_i} \phi_1(v, Z_{i}, O_{j}) d\phi_2(v, Z_{i}, O_{k}) \int_{\tau_1\vee Q_{i'}}^{\tau_2\wedge T_{i'}} \phi_1(v, Z_{i'}, O_{j'}) d\phi_2(v, Z_{i'}, O_{k'}) \right\} = 0. 
	\end{align}
	This is because for $(i,j,k,i',j',k')\in \mathcal{J}_1$, $O_j$ is independent of $(O_i, O_k,O_{i'}, O_{j'}, O_{k'})$, so 
 $$\E\left\{\left.\phi_1(v, Z_{i}, O_{j})\right|O_i, O_k,O_{i'}, O_{j'}, O_{k'} \right\} = \E\left\{\left.\phi_1(v, Z_{i}, O_{j})\right|O_i\right\} = 0$$ 
 by the property of influence functions. 
 Therefore, 
	\begin{align}
	&\quad \E\left\{\int_{\tau_1\vee Q_i}^{\tau_2\wedge T_i} \phi_1(v, Z_{i}, O_{j}) d\phi_2(v, Z_{i}, O_{k}) \int_{\tau_1\vee Q_{i'}}^{\tau_2\wedge T_{i'}} \phi_1(v, Z_{i'}, O_{j'}) d\phi_2(v, Z_{i'}, O_{k'}) \right\}\\
	& = \ee{\E\left\{\left.\int_{\tau_1\vee Q_i}^{\tau_2\wedge T_i} \phi_1(v, Z_{i}, O_{j}) d\phi_2(v, Z_{i}, O_{k}) \int_{\tau_1\vee Q_{i'}}^{\tau_2\wedge T_{i'}} \phi_1(v, Z_{i'}, O_{j'}) d\phi_2(v, Z_{i'}, O_{k'}) \right|O_i,O_k,O_{i'},O_{j'},O_{k'} \right\}} \\
	& = \ee{\int_{\tau_1\vee Q_i}^{\tau_2\wedge T_i} \E\left\{\left.\phi_1(v, Z_{i}, O_{j})\right|O_i,O_k,O_{i'},O_{j'},O_{k'} \right\} d\phi_2(v, Z_{i}, O_{k}) \int_{\tau_1\vee Q_{i'}}^{\tau_2\wedge T_{i'}} \phi_1(v, Z_{i'}, O_{j'}) d\phi_2(v, Z_{i'}, O_{k'}) }  \\
	& = 0. 
	\end{align}
	
	Likewise, we can show that it has expectation zero when $(i,j,k,i',j',k')\in \mathcal{J}_3$. In addition, by integration by parts, we can show that it has expectation zero when $(i,j,k,i',j',k')\in \mathcal{J}_2$ or $\mathcal{J}_4$.

	Denote $|\mathcal{J}|$ the total number of elements in set $\mathcal{J}$. 
	Then
	\begin{align}
	|\mathcal{J}| &= {4 \choose 1}n(n-1)^5 - {4\choose 2} n(n-1)(n-2)^4 + {4 \choose 3}n(n-1)(n-2)(n-3)^3 \\
	&\quad - {4 \choose 4} n(n-1)(n-2)(n-3)(n-4)^2\\
	& = \left\{4n^6-20n^5+ O(n^4)\right\} - \left\{6n^6 - 54n^5 + O(n^4)\right\} \\
	& \quad +\left\{4n^6-48n^5 + O(n^4)\right\} - \left\{n^6 + 14n^5+ O(n^4)\right\} \\
	& = n^6 + O(n^4). 
	\end{align}
	So
	\begin{align}
	|\mathcal{K}| = n^6 - |\mathcal{J}| = O(n^4).
	\end{align}
	Therefore, $\E(|\B_{20}|^2)  = O(n^{-2})$. 
	By Markov's inequality, $\B_{20} = O_p(n^{-1}) = o_p(n^{-1/2})$.    
 
 \bigskip 
	Now consider $\B_{3}$. By \eqref{eq:asymp_int_mF}, 
	\begin{align}
	\B_3 & = \frac{1}{n}\sum_{i=1}^n \int_0^\infty \left\{\frac{m(v,Z_i;\hat F) - \theta_0\hat F(v|Z_i)}{1-\hat F(v|Z_i)} - \frac{m(v,Z_i;F^\divideontimes) - \theta_0 F^\divideontimes(v|Z_i)}{1-F^\divideontimes(v|Z_i)}\right\}  \frac{d\bar M_{Q,i}(v; G^\divideontimes)}{G^\divideontimes(v|Z_i)} \\
	& =  \frac{1}{n}\sum_{i=1}^n \int_{\tau_1}^{\tau_2} \left\{\frac{1}{n}\sum_{j=1}^n\phi_1(v,Z_i,O_j) + \tilde R_{1}(v,Z_i)\right\}\frac{d\bar M_{Q,i}(v; G^\divideontimes)}{G^\divideontimes(v|Z_i)} \\
	& = \B_{30}+ \B_{31},
	\end{align}
	where
	\begin{align}
	\B_{30} & = \frac{1}{n}\sum_{i=1}^n \int_{\tau_1}^{\tau_2} \left\{\frac{1}{n}\sum_{j=1}^n\phi_1(v,Z_i,O_j) \right\}\frac{d\bar M_{Q,i}(v; G^\divideontimes)}{G^\divideontimes(v|Z_i)}, \\
	\B_{31} & = \frac{1}{n}\sum_{i=1}^n \int_{\tau_1}^{\tau_2}  \tilde R_{1}(v,Z_i)\frac{d\bar M_{Q,i}(v; G^\divideontimes)}{G^\divideontimes(v|Z_i)} .
	\end{align}
	Since $G^\divideontimes$ satisfies the overlap assumption, $\left\|\bar M_{Q}(\cdot; G^\divideontimes)/G^\divideontimes(\cdot|Z)\right\|_{\TV,2} = O(1)$, so by Cauchy-Schwartz inequality, 
	\begin{align}
	\E(|\B_{31}|) &\lesssim \|\tilde R_1(\cdot,Z)\|_{\sup,2} 
 \cdot \left\|\frac{\bar M_{Q}(\cdot; G^\divideontimes)}{G^\divideontimes(\cdot|Z)}\right\|_{\TV,2}= o(n^{-1/2}) \cdot O(1) = o(n^{-1/2}), 
	\end{align}
	we have that $\B_{31} = o_p(n^{-1/2})$ by Markov's inequality. 
	
	Next, we utilize the asymptotic result for $U$-statistics stated in Theorem \ref{lem:Ustats} (Theorem 12.3 in \citet{van2000asymptotic}) to show that $\B_{30}$ is asymptotically linear.
 We consider the $U$-statistics
 \begin{align}
     Y &= \frac{1}{{n\choose 2}}\sum_{i<j} \left\{\int_{\tau_1}^{\tau_2} \phi_1(v,Z_j,O_i) \frac{d\bar M_{Q,j}(v; G^\divideontimes)}{G^\divideontimes(v|Z_j)}  
	+\int_{\tau_1}^{\tau_2} \phi_1(v,Z_i,O_j) \frac{d\bar M_{Q,i}(v; G^\divideontimes)}{G^\divideontimes(v|Z_i)}\right\}.
 \end{align}
 In the following, we will show that $\B_{30} = Y/2 + o_p(n^{-1/2})$. Note that
 \begin{align}
     \frac{1}{n^2} = \frac{1}{2} \frac{1}{{n\choose 2}} - \frac{1}{n^2(n-1)}.
 \end{align}
 We have
	\begin{align}
	\B_{30} 
	& = \frac{1}{n^2}\sum_{i=1}^n\sum_{j=1}^n \int_{\tau_1}^{\tau_2} \phi_1(v,Z_j,O_i) \frac{d\bar M_{Q,j}(v; G^\divideontimes)}{G^\divideontimes(v|Z_j)}\\
	& = \frac{1}{n^2}\sum_{i<j} \left\{\int_{\tau_1}^{\tau_2} \phi_1(v,Z_j,O_i) \frac{d\bar M_{Q,j}(v; G^\divideontimes)}{G^\divideontimes(v|Z_j)}  
	+\int_{\tau_1}^{\tau_2} \phi_1(v,Z_i,O_j) \frac{d\bar M_{Q,i}(v; G^\divideontimes)}{G^\divideontimes(v|Z_i)}\right\} \\
	&\quad 
	+ \frac{1}{n^2} \sum_{i=1}^n \int_{\tau_1}^{\tau_2} \phi_1(v,Z_i,O_i) \frac{d\bar M_{Q,i}(v; G^\divideontimes)}{G^\divideontimes(v|Z_i)} \\
	& = \left\{ \frac{1}{2} \frac{1}{{n\choose 2}} - \frac{1}{n^2(n-1)} \right\} \\
 &\quad\quad\quad \cdot \sum_{i<j} \left\{\int_{\tau_1}^{\tau_2} \phi_1(v,Z_j,O_i) \frac{d\bar M_{Q,j}(v; G^\divideontimes)}{G^\divideontimes(v|Z_j)} 
	+\int_{\tau_1}^{\tau_2} \phi_1(v,Z_i,O_j) \frac{d\bar M_{Q,i}(v; G^\divideontimes)}{G^\divideontimes(v|Z_i)}\right\} \\
	&\quad + O_p(n^{-1}) \\
	& = \frac{Y}{2} - \frac{1}{n^2(n-1)} \sum_{i<j} \left\{\int_{\tau_1}^{\tau_2} \phi_1(v,Z_j,O_i) \frac{d\bar M_{Q,j}(v; G^\divideontimes)}{G^\divideontimes(v|Z_j)} 
	+\int_{\tau_1}^{\tau_2} \phi_1(v,Z_i,O_j) \frac{d\bar M_{Q,i}(v; G^\divideontimes)}{G^\divideontimes(v|Z_i)}\right\} \\
 &\quad +O_p(n^{-1}) \\
 & = \frac{Y}{2} +O_p(n^{-1}) +O_p(n^{-1})\\
 & = \frac{Y}{2} +o_p(n^{-1/2}). \label{eq:B_30}
	\end{align}
 Denote $o = (q,t,z)$. 
 Define
 \begin{align}
	h_1(o;F^\divideontimes, G^\divideontimes) 
	& = \E\left\{ \frac{\phi_1(q,z,O; F^\divideontimes)}{G^\divideontimes(q|z)} -  \int_{\tau_1}^{\tau_2} \phi_1(v,z,O; F^\divideontimes) \mathbbm{1}(q\leq v<t) \frac{dG^\divideontimes(v|z)}{G^\divideontimes(v|z)^2} \right. \\
	&\quad\quad\quad \left.
	+ \frac{\phi_1(Q,Z,o;F^\divideontimes)}{G^\divideontimes(Q|Z)} - \int_{\tau_1}^{\tau_2} \phi_1(v,Z,o; F^\divideontimes) \mathbbm{1}(Q\leq v<T)\frac{dG^\divideontimes(v|Z)}{G^\divideontimes(v|Z)^2}\right\}.
\end{align}
Since $\E\{\phi_1(v,z,O; F^\divideontimes)\} = 0$ for all $v\geq 0$ and $z\in\Z$ by the property of influence functions, we have
\begin{align}
	h_1(o;F^\divideontimes, G^\divideontimes) 
	& = \frac{\E\left\{ \phi_1(q,z,O; F^\divideontimes) \right\} }{G^\divideontimes(q|z)} -  \int_{\tau_1}^{\tau_2} \E\left\{\phi_1(v,z,O; F^\divideontimes)\right\} \mathbbm{1}(q\leq v<t) \frac{dG^\divideontimes(v|z)}{G^\divideontimes(v|z)^2} \\
	&\quad + \E\left\{\frac{\phi_1(Q,Z,o;F^\divideontimes)}{G^\divideontimes(Q|Z)} - \int_{\tau_1}^{\tau_2} \phi_1(v,Z,o; F^\divideontimes) \mathbbm{1}(Q\leq v<T)\frac{dG^\divideontimes(v|Z)}{G^\divideontimes(v|Z)^2}\right\} \\
 & = \E\left\{\frac{\phi_1(Q,Z,o;F^\divideontimes)}{G^\divideontimes(Q|Z)} - \int_{\tau_1}^{\tau_2} \phi_1(v,Z,o; F^\divideontimes) \mathbbm{1}(Q\leq v<T)\frac{dG^\divideontimes(v|Z)}{G^\divideontimes(v|Z)^2}\right\},
\end{align}
Consider
 \begin{align}
	\hat Y = \frac{2}{n}\sum_{i=1}^n h_1(O_i, F^\divideontimes, G^\divideontimes). 
\end{align}
We will show that $Y = \hat Y + o_P(n^{-1/2})$ by Theorem \ref{lem:Ustats}. Before applying Theorem \ref{lem:Ustats}, we will first show in the following that $\E\left\{h_1(O;F^\divideontimes, G^\divideontimes)\right\} = 0$.

 Let $O' = (Q',T',Z')$ be an i.i.d copy of $O$. Let $\E'(\cdot|O)$ and $\E(\cdot|O')$ denote the conditional expectation taken with respect to $O'$ given $O$ and $O$ given $O'$, respectively. Then 
 \begin{align}
     &\quad \E\left\{h_1(O;F^\divideontimes, G^\divideontimes) \right\} \\
     & = \E\left[\E'\left\{\left.\frac{\phi_1(Q',Z',O;F^\divideontimes)}{G^\divideontimes(Q'|Z')} - \int_{\tau_1}^{\tau_2} \phi_1(v,Z',O; F^\divideontimes) \mathbbm{1}(Q'\leq v<T')\frac{dG^\divideontimes(v|Z')}{G^\divideontimes(v|Z')^2}\right|O\right\} \right]\\
     & = \E'\left[\E\left\{\left.\frac{\phi_1(Q',Z',O;F^\divideontimes)}{G^\divideontimes(Q'|Z')} - \int_{\tau_1}^{\tau_2} \phi_1(v,Z',O; F^\divideontimes) \mathbbm{1}(Q'\leq v<T')\frac{dG^\divideontimes(v|Z')}{G^\divideontimes(v|Z')^2}\right|O'\right\} \right] \\
     & = \E'\left[\frac{\E\left\{\left.\phi_1(Q',Z',O;F^\divideontimes)\right|O'\right\}}{G^\divideontimes(Q'|Z')} - \int_{\tau_1}^{\tau_2} \E\left\{\left.\phi_1(v,Z',O; F^\divideontimes)\right|O'\right\} \mathbbm{1}(Q'\leq v<T')\frac{dG^\divideontimes(v|Z')}{G^\divideontimes(v|Z')^2} \right] \\
 \end{align}
 Since $\E\left\{\left.\phi_1(Q',Z',O;F^\divideontimes)\right|O'\right\} =0$ and $\E\left\{\left.\phi_1(v,Z',O; F^\divideontimes)\right|O'\right\}$ for all $v\in[\tau_1,\tau_2]$  by the property of influence functions, we have 
$\E\left\{h_1(O;F^\divideontimes, G^\divideontimes) \right\}  = 0$. 
 In addition, it can be verified that $\E\left\{h_1(O;F^\divideontimes, G^\divideontimes)^2\right\}<\infty$. 
Therefore, by Theorem \ref{lem:Ustats},  
	\begin{align}
	Y = \hat Y + o_p(n^{-1/2}) = \frac{2}{n}\sum_{i=1}^n h_1(O_i, F^\divideontimes, G^\divideontimes) + o_p(n^{-1/2})  .
	\end{align}
	Therefore by \eqref{eq:B_30},
	\begin{align}
	\B_{30}
	& = \frac{1}{n}\sum_{i=1}^n h_{1}(O_i, F^\divideontimes, G^\divideontimes) + o_p(n^{-1/2}),
	\end{align}
 
	Likewise, by integration by parts, we can show that
	\begin{align}
	\B_{4} & = \frac{1}{n}\sum_{i=1}^n h_{2}(O_i,F^\divideontimes, G^\divideontimes) + o_p(n^{-1/2})
	\end{align}
 for some function $h_2$ such that $h_{2}(O,F^\divideontimes, G^\divideontimes)$ has mean zero and finite second moment.
	
	Putting the above together, we have 
	\begin{align}
	\frac{1}{n}\sum_{i=1}^n U_i(\theta_0, \hat F,\hat G) 
	& = \frac{1}{n}\sum_{i=1}^n \left\{U_i(\theta_0, F^\divideontimes, G^\divideontimes) + h_{1}(O_i,F^\divideontimes, G^\divideontimes) + h_{2}(O_i,F^\divideontimes, G^\divideontimes) \right\} + o_p(n^{-1/2}). 
	\end{align}
	By central limit theorem,
	\begin{align}
	n^{-1/2}\sum_{i=1}^n \left\{U_i(\theta_0, F^\divideontimes, G^\divideontimes) + h_{1}(O_i, F^\divideontimes, G^\divideontimes) + h_{2}(O_i, F^\divideontimes, G^\divideontimes) \right\}  \convd N(0, \sigma_1^2),
	\end{align}
	where
	\begin{align}
	\sigma_1^2 = \ee{\left\{U_i(\theta_0, F^\divideontimes, G^\divideontimes) + h_{1}(O_i, F^\divideontimes, G^\divideontimes) + h_{2}(O_i, F^\divideontimes, G^\divideontimes) \right\}^2}.
	\end{align}

	\bigskip 
	Now consider the special case when $F^\divideontimes = F_0$ and $G^\divideontimes = G_0$. We will show that $\B_{3} = o_p(n^{-1/2})$ and $\B_{4} = o_p(n^{-1/2})$ using similar techniques as in proving $\B_{20} = o(n^{-1/2})$.
 First consider $\B_{3}$. It suffices to show that $\B_{30} = o_p(n^{-1/2})$. We have that 
	\begin{align}
	|\B_{30}|^2 & = \frac{1}{n^4}\sum_{i=1}^n \sum_{j=1}^n\sum_{i'=1}^n \sum_{j'=1}^n  \int_{\tau_1}^{\tau_2} \phi_1(v,Z_j,O_i) \frac{d\bar M_{Q,j}(v; G_0)}{G_0(v|Z_j)}
	\int_{\tau_1}^{\tau_2} \phi_1(v,Z_{j'},O_{i'}) \frac{d\bar M_{Q,j'}(v; G_0)}{G_0(v|Z_{j'})}.
	\end{align}
 In the following we show that if one of $i,j,i',j'$ is different from the other indices, then
	\begin{align}
	\E\left\{\int_{\tau_1}^{\tau_2} \phi_1(v,Z_j,O_i) \frac{d\bar M_{Q,j}(v; G_0)}{G_0(v|Z_j)}
	\int_{\tau_1}^{\tau_2} \phi_1(v,Z_{j'},O_{i'}) \frac{d\bar M_{Q,j'}(v; G_0)}{G_0(v|Z_{j'})} \right\} = 0 \label{eq:EB30}.
	\end{align}
Assume that $i$ is not the same as any of $\{j,i',j'\}$, then $O_i$ is independent of $(O_j,O_{i'},O_{j'})$, so $\E\{\phi_1(v,Z_j,O_i)|O_j,O_{i'},O_{j'} \} = \E\{\phi_1(v,Z_j,O_i)|O_j \} =0$ by the property of influence functions. Therefore, 
\begin{align}
    &\quad \E\left\{\int_{\tau_1}^{\tau_2} \phi_1(v,Z_j,O_i) \frac{d\bar M_{Q,j}(v; G_0)}{G_0(v|Z_j)}
	\int_{\tau_1}^{\tau_2} \phi_1(v,Z_{j'},O_{i'}) \frac{d\bar M_{Q,j'}(v; G_0)}{G_0(v|Z_{j'})} \right\}  \\
 & = \E\left[\E\left\{\left. \int_{\tau_1}^{\tau_2} \phi_1(v,Z_j,O_i) \frac{d\bar M_{Q,j}(v; G_0)}{G_0(v|Z_j)}
	\int_{\tau_1}^{\tau_2} \phi_1(v,Z_{j'},O_{i'}) \frac{d\bar M_{Q,j'}(v; G_0)}{G_0(v|Z_{j'})} \right|O_j,O_{i'},O_{j'} \right\} \right] \\
 & = \E\left[ \int_{\tau_1}^{\tau_2} \E\left\{\left.\phi_1(v,Z_j,O_i)\right|O_j,O_{i'},O_{j'} \right\}  \frac{d\bar M_{Q,j}(v; G_0)}{G_0(v|Z_j)}
	\int_{\tau_1}^{\tau_2} \phi_1(v,Z_{j'},O_{i'}) \frac{d\bar M_{Q,j'}(v; G_0)}{G_0(v|Z_{j'})} \right]\\
 & = 0.
\end{align}
On the other hand, if $j$ is not the same as any of $\{i,i',j'\}$, then $O_j$ is independent of $(O_i,O_{i'},O_{j'})$, and by Proposition \ref{prop:int_dM}
we have
\begin{align}
    \E\left\{\left.\int_{\tau_1}^{\tau_2} \phi_1(v,Z_j,O_i) \frac{d\bar M_{Q,j}(v; G_0)}{G_0(v|Z_j)} \right|O_i,O_{i'},O_{j'}\right\} = 0.
\end{align}
Therefore, 
\begin{align}
    &\quad \E\left\{\int_{\tau_1}^{\tau_2} \phi_1(v,Z_j,O_i) \frac{d\bar M_{Q,j}(v; G_0)}{G_0(v|Z_j)}
	\int_{\tau_1}^{\tau_2} \phi_1(v,Z_{j'},O_{i'}) \frac{d\bar M_{Q,j'}(v; G_0)}{G_0(v|Z_{j'})} \right\}  \\
 & = \E\left[\E\left\{\left. \int_{\tau_1}^{\tau_2} \phi_1(v,Z_j,O_i) \frac{d\bar M_{Q,j}(v; G_0)}{G_0(v|Z_j)}
	\int_{\tau_1}^{\tau_2} \phi_1(v,Z_{j'},O_{i'}) \frac{d\bar M_{Q,j'}(v; G_0)}{G_0(v|Z_{j'})} \right|O_i,O_{i'},O_{j'} \right\} \right] \\
 & = \E\left[\E\left\{\left. \int_{\tau_1}^{\tau_2} \phi_1(v,Z_j,O_i) \frac{d\bar M_{Q,j}(v; G_0)}{G_0(v|Z_j)}\right|O_i,O_{i'},O_{j'} \right\}
	\int_{\tau_1}^{\tau_2} \phi_1(v,Z_{j'},O_{i'}) \frac{d\bar M_{Q,j'}(v; G_0)}{G_0(v|Z_{j'})}  \right] \\
 & = 0.
\end{align}
Similarly we have \eqref{eq:EB30} when $i'$ is not the same as any of $\{i,j,j'\}$ or $j'$ is not the same as any of $\{i,j,i'\}$.
 
	The total number of such $(i,j,i',j')$'s is
	\begin{align}
	& \quad {4\choose 1} n(n-1)^3 - {4\choose 2} n(n-1)(n-2)^2 +{4\choose 3} n(n-1)(n-2)(n-3) \\
    &\quad\quad\quad - {4\choose 4} n(n-1)(n-2)(n-3) \\
	& = 4\{n^4 -3n^3 +O(n^2)\} - 6\{n^4 - 5n^3+O(n^2)\} + 4\{n^4 - 6n^2+O(n^2)\} - \{n^4 - 6n^3+O(n^2)\} \\
	& = n^4 + O(n^2).
	\end{align}
	Therefore, 
	\begin{align}
	\E\{|\B_{30}|^2\} & = \frac{1}{n^4} \cdot O(n^2) = O(n^{-2}),
	\end{align}
	so by Markov's inequality, $\B_{30} = O_p(n^{-1}) = o_p(n^{-1/2})$, so $\B_{3} = o_p(n^{-1/2})$.
	Likewise, we can show that $\B_{4} = o_p(n^{-1/2})$. 
	Therefore, 
	\begin{align}
	\frac{1}{n}\sum_{i=1}^n U_{i}(\theta_0,\hat F,\hat G)  = \frac{1}{n}\sum_{i=1}^n U_{i}(\theta_0,F_0, G_0) + o_p(n^{-1/2}).  \label{eq:AN_proof_1}
	\end{align}
	
	(3) Finally, in the proof above we have shown that 
 $\sigma^2 = \beta_0^2 \cdot \E\{U(\theta_0,F_0, G_0)^2\}$. 
 Similar to the consistency proof in (1), we can show that 
 $\hat\sigma^2 \convp \beta_0^2\cdot \E\{U(\theta_0,F_0, G_0)^2\}$.

\end{proof}

\newpage
 \begin{proof}[Proof of Theorem \ref{thm:rdr}]
	(1) The consistency of $\hat\theta_{cf}$ can be shown similarly as the consistency proof of Theorem \ref{thm:mdr}. 
	
	(2) Now we show that the asymptotic normality. The conditional independence introduced by the cross-fitting procedure plays a central role for showing that the terms $\B_2'$, $\B_3'$ and $\B_4'$ in the error decomposition \eqref{eq:AN_cf_err_decomp} are all $o_p(n^{-1/2})$. Then the asymptotic normality follows from the central limit theorem. 

 Recall that $K$ is the total number of folds in Algorithm \ref{alg:cf}, which is  fixed.
 Without loss of generality, assume $n = Km$. Then $|\I_k| = m$ for all $k = 1,...,K$, and we have
	\begin{align}
	n^{1/2}(\hat\theta_{cf} - \theta_0 )
	& = \left. n^{1/2}\left[\frac{1}{K}\sum_{k=1}^K\frac{1}{m}\sum_{i\in\I_k} U_{i}\{\theta_0,\hat F^{(-k)},\hat G^{(-k)}\} \right] \right/\left[\frac{1}{K}\sum_{k=1}^K\frac{1}{m}\sum_{i\in\I_k} D_{i}\{\hat F^{(-k)},\hat G^{(-k)}\} \right]. \label{eq:cf_est_err}
	\end{align}
	Similar to the proof of Theorem \ref{thm:mdr}, we can show that the denominator of the right hand side (R.H.S.) 
 of \eqref{eq:cf_est_err} converges to $1/\beta_0$ in probability. 
	
	In the following,  we will show that for any $k\in\{1,...,K\}$,
	\begin{align}
	\frac{1}{m}\sum_{i\in\I_k} U_{i}\{\theta_0, \hat F^{(-k)},\hat G^{(-k)}\}  
	&=  \frac{1}{m}\sum_{i\in\I_k} U_{i}(\theta_0, F_0, G_0) + o_p(m^{-1/2}) \label{eq:ANproof_4} \\
	&=  \frac{1}{m}\sum_{i\in\I_k} U_{i}(\theta_0, F_0, G_0) + o_p(n^{-1/2}). 
	\end{align} 
	This implies that the numerator of the R.H.S of \eqref{eq:cf_est_err} can be written as $n^{-1/2}\sum_{i=1}^n U_{i}(\theta_0, F_0, G_0) + o_p(1)$, which converges in distribution to $N\left(0, \E\left\{U(\theta_0, F_0, G_0)^2\right\}\right)$ by central limit theorem. 
	Therefore, by Slutsky's Theorem, 
	$
	n^{1/2}(\hat\theta_{cf} - \theta_0) \convd N \left(0, \beta_0^2 \E\{U(\theta_0, F_0, G_0)^2\}\right).
	$
	
	We now show \eqref{eq:ANproof_4}. Without loss of generality, consider $k=1$ and suppose that $\I_1 = \{1,...,m\}$.
	For simplification of the notation, we will denote $\hat F' = \hat F^{(-1)}$ and $\hat G' = \hat G^{(-1)}$ the estimates of $F$ and $G$ using the out-of-1-fold sample, i.e., data indexed by $\{m+1,...,n\}$.
	Like the asymptotic normality proof in Theorem \ref{thm:mdr}, we consider the decomposition,
	\begin{align}
	\frac{1}{m}\sum_{i=1}^m U_i(\theta_0, \hat F',\hat G')
	&= \B_1'+\B_2' + \B_3' + \B_4', \label{eq:AN_cf_err_decomp}
	\end{align}
	where
	\begin{align}
	\B_1' & = \frac{1}{m}\sum_{i=1}^m U_i(\theta_0, F_0, G_0), \\
	\B_2' & = \frac{1}{m}\sum_{i=1}^m \left[ \left\{ U_i(\theta_0, \hat F',\hat G') -U_i(\theta_0, F_0, \hat G')\right\}  - \left\{ U_i(\theta_0, \hat F',G_0) -U_i(\theta_0, F_0, G_0)\right\} \right], \\
	\B_3' & = \frac{1}{m}\sum_{i=1}^m \left\{ U_i(\theta_0, \hat F',G_0) -U_i(\theta_0, F_0, G_0)\right\}, \\
	\B_4' & =  \frac{1}{m}\sum_{i=1}^m  \left\{ U_i(\theta_0, F_0,\hat G') -U_i(\theta_0, F_0, G_0)\right\}.
	\end{align}
 
 We first consider $\B_2'$. By \eqref{eq:asymp_int_mF}, 
	\begin{align}
	\B_2'
	& = \frac{1}{m}\sum_{i=1}^m \int_{\tau_1}^{\tau_2} \left\{\frac{m(v,Z_i;\hat F') - \theta_0\hat F'(v|Z_i)}{1-\hat F'(v|Z_i)} - \frac{m(v,Z_i;F_0) - \theta_0 F_0(v|Z_i)}{1-F_0(v|Z_i)}\right\}\\
	&\quad\quad\quad\quad\quad\quad\quad  \cdot\mathbbm{1}(Q_i\leq v<T_i) \ d\left\{\frac{1}{\hat G'(v|Z_i)} - \frac{1}{G_0(v|Z_i)} \right\}. 
	\end{align}
 By Assumption \ref{ass:prodrate}, $\E(\B_2') = \D_{\dagger}(\hat F,\hat G; F_0,G_0)  = o(n^{-1/2})$, so $\B_2' = o_p(m^{-1/2})$ by Markov's inequality. 
	In addition,  by \eqref{eq:asymp_int_mF}, 
	\begin{align}
	\B_3' & = \frac{1}{m}\sum_{i=1}^m \int_{\tau_1}^{\tau_2} \left\{\frac{m(v,Z_i;\hat F') - \theta_0\hat F'(v|Z_i)}{1-\hat F'(v|Z_i)} - \frac{m(v,Z_i;F_0) - \theta_0 F_0(v|Z_i)}{1-F_0(v|Z_i)}\right\}  \frac{d\bar M_{Q,i}(v; G_0)}{G_0(v|Z_i)}. 
	\end{align}
	Let $O'$ denote the out-of-1-fold sample $\{O_i:\ i = {m+1},...,n\}$, where $O_i = (Q_i, T_i, Z_i)$. Recall that $\{O_1,...,O_m\}$ is the data in the first fold. 
	By the law of total variance, 
	\begin{align}
	\var(\B_3') = \E\{\var(\B_3'|O')\} + \var\{\E(\B_3'|O')\}.
	\end{align}
	When conditioning on $O'$, $\hat F'$ is fixed. In addition, $O_1$ and $O'$ are independent, so by Proposition \ref{prop:int_dM}, 
 we have
	\begin{align}
	&\quad \E(\B_3'|O')  \\
	&= \E\left[\left.\int_{\tau_1}^{\tau_2} \left\{\frac{m(v,Z_1;\hat F') - \theta_0\hat F'(v|Z_1)}{1-\hat F'(v|Z_1)} - \frac{m(v,Z_1;F_0) - \theta_0 F_0(v|Z_1)}{1-F_0(v|Z_1)}\right\}  \frac{d\bar M_{Q,1}(v; G_0)}{G_0(v|Z_1)}\right|O'\right] \\
	& = 0.
	\end{align} 
	So 
	\begin{align}
	&\quad \var(\B_3'|O') \\
 & = \frac{1}{m} \var\left(\left.\int_{\tau_1}^{\tau_2}  \left\{\frac{m(v,Z_1;\hat F') - \theta_0\hat F'(v|Z_1)}{1-\hat F'(v|Z_1)} - \frac{m(v,Z_1;F_0) - \theta_0 F_0(v|Z_1)}{1-F_0(v|Z_1)}\right\}  \frac{d\bar M_{Q,1}(v; G_0)}{G_0(v|Z_1)}\right|O'\right)\\
	& = \frac{1}{m} \E\left(\left.\left|\int_{\tau_1}^{\tau_2}  \left\{\frac{m(v,Z_1;\hat F') - \theta_0\hat F'(v|Z_1)}{1-\hat F'(v|Z_1)} - \frac{m(v,Z_1;F_0) - \theta_0 F_0(v|Z_1)}{1-F_0(v|Z_1)}\right\}  \frac{d\bar M_{Q,1}(v; G_0)}{G_0(v|Z_1)}\right|^2\right|O'\right) \\
 &\leq \frac{1}{m} \E\left(\left.\left[\sup_{t\in[\tau_1, \tau_2]}\left|\frac{m(v,Z_1;\hat F') - \theta_0\hat F'(v|Z_1)}{1-\hat F'(v|Z_1)} - \frac{m(v,Z_1;F_0) - \theta_0 F_0(v|Z_1)}{1-F_0(v|Z_1)}\right| \right.\right.\right.\\ 
 &\qquad\qquad\qquad\qquad\qquad\qquad \cdot \left.\left.\left.\TV\left\{\frac{\bar M_{Q,1}(t; G_0)}{G_0(t|Z_1)}; \tau_1, \tau_2\right\}\right]^2\right|O'\right).
\end{align}
Since $G_0$ satisfied Assumption \ref{ass:overlap}, we have
\begin{align}
    \TV\left\{\frac{\bar M_{Q,1}(t; G_0)}{G_0(t|Z_1)}; \tau_1, \tau_2\right\} 
    & \leq \frac{1}{\delta_2}\cdot \TV\left\{\bar M_{Q,1}(t; G_0); \tau_1, \tau_2\right\}
    \leq \frac{1}{\delta_2}\left(1+\frac{1}{\delta_2}\right) <\infty.
\end{align}
So
\begin{align}
    \var(\B_3'|O') \lesssim \frac{1}{m} \E\left\{\left.\sup_{t\in[\tau_1, \tau_2]}\left|\hat F'(t|Z_1) - F_0(t|Z_1)\right|^2\right|O'\right\}. \label{eq:ANproof_2}
\end{align}
	Therefore,  
	\begin{align}
	\var(\B_3')
	& =  \E\{\var(\B_3'|O')\} + 0 \\
	&\lesssim  \frac{1}{m} \E\left(\sup_{t\in[\tau_1, \tau_2]}\left|\hat F'(t|Z_1) - F_0(t|Z_1)\right|^2\right)\\
	& =\frac{1}{m}  \|\hat F' - F_0\|_{\dagger,\sup, 2}^2 \\
	& =  o(m^{-1}).
	\end{align}
	By Chebyshev's inequality, for any $\epsilon>0$, we have 
 \begin{align}
     \PP\left(m^{1/2}|\B_3'| > \epsilon \right) = \PP(|\B_3'| > m^{-1/2}\epsilon) \leq \frac{\var(\B_3')}{m^{-1}\epsilon^2}  = \frac{o(m^{-1})}{m^{-1}\epsilon^2}\to 0,
 \end{align}
 as $m\to 0$. Therefore, 
	$\B_3' = o_p(m^{-1/2})$.
	
	Likewise,  using integration by parts, we can show that 
	$\B_4' = o_p(m^{-1/2})$.
	
	Putting the above together, we have
	\begin{align}
	\frac{1}{m}\sum_{i=1}^m U_i(\theta_0, \hat F',\hat G') 
	= \frac{1}{m}\sum_{i=1}^m U_i(\theta_0, F_0, G_0) + o_p(m^{-1/2}).
	\end{align}
	Therefore, the conclusion follows.
	
	(3) For the variance estimator, it can be shown that $\hat\sigma_{cf}^2 \convp \beta_0^2 \cdot \E\{U(\theta_0, F_0, G_0)\}$ in a similar way to part (3) in the proof of Theorem \ref{thm:mdr}.
\end{proof}

\newpage
\subsection{Verifying Assumption \ref{assump:if}  under the Cox model}


We consider estimating $F$ by 
the \citet{cox1972regression}  proportional hazards model:
\begin{align}
    \lambda(t|Z) = \lambda_0(t) \exp\{\psi' Z\},
\end{align}
where $\psi$ is the {log} hazard ratio, and $\lambda_0$  the baseline hazard. Denote 
$\Lambda(t) = \int_0^t \lambda_0(u) du$  the cumulative baseline hazard. 
Suppose that the above model is correctly specified for $F$.  
Let $\psi_0$ denote the true value of $\psi$, and $\Lambda_0$  the true value of $\Lambda$. We make the following assumption.

 
\begin{assumption} \label{ass:regularity_cox}
(i) The covariates $Z$ is bounded a.s.; \
(ii) The parameter space for $\psi$ is compact.
\end{assumption}

Following \citet{tsiatis1981large}, let $N(t) = \mathbbm{1}(Q<T\leq t)$ denote the counting process for $T$. Denote $Y(t) = \mathbbm{1}(Q\leq t \leq T)$, and denote 
\begin{align}
    &B(t) = \E\left\{N(t)\right\} , \quad \hat B(t) = \frac{1}{n} \sum_{i=1}^n dN_i(t),\\
    &\E(e^{\psi'z}, t)   = \E\left\{e^{\psi'Z}Y(t)\right\}, \quad 
    \hat \E(e^{\psi'z}, t)  = \frac{1}{n}\sum_{j=1}^n Y_j(t) \exp\{\psi'Z_j\}.
\end{align}
Note that $B(t)$ and $\E(e^{\psi'z}, t)$ are deterministic functions of $t$, while $\hat B(t)$ and $\hat \E(e^{\psi'z}, t)$ are random functions of $t$. 
Consider the partial likelihood \citep{cox1972regression, cox1975partial} 
estimator $\hat \psi$ for $\psi$ with risk set adjustment for left truncation, and 
the Breslow's estimator for $ \Lambda(t)$ which is defined as: 
\begin{align}
    \tilde \Lambda(\hat\psi, t) = \int_0^t \frac{\sum_{i=1}^n dN_i(x)}{\sum_{j=1}^n Y_i(x) \exp\{\hat \psi'Z_j\}} = \int_0^t \frac{d \hat B(t)}{\hat \E(e^{\hat\psi'z}, t)}. \label{eq:brewslow}
\end{align}
We consider the plug-in estimator for $F$: 
\begin{align}
    \hat F(t|Z) = 1-\exp\{-\tilde \Lambda(\hat\psi, t)e^{\hat\psi' Z}\}
\end{align}

The partial likelihood estimator $\hat\psi$ is known to be asymptotically linear. Denote $\phi(O)$ its influence function. Then 
\begin{align}
    \hat \psi - \psi_0 = \frac{1}{n}\sum_{i=1}^n \phi(O_i) + o_p(n^{-1/2}). \label{eq:beta_AL}
\end{align}
In order to verify Assumption \ref{assump:if}, under Assumption \ref{ass:regularity_cox} it suffices to show that $\tilde \Lambda(\hat\psi, t)$ is asymptotically linear, and {its} residual term is $o(n^{-1/2})$ under $\|\cdot\|_{\sup,2}$ and is  $o(1)$ under $\|\cdot\|_{\TV,2}$.

The estimator $\tilde \Lambda(\hat\psi, t)$ has been shown to have an asymptotic linear representation and the corresponding residual term is uniformly $o_p(n^{-1/2})$ on $[\tau_1,\tau_2]$ \citep{kosorok2008introduction, lopuhaa2013asymptotic}. 
However, their asymptotic linear representation  is slightly different from what is  in Assumption \ref{assump:if}. 
In the following, we build upon their conclusions and show that $\tilde \Lambda(\hat\psi, t)$ has an asymptotic linear representation that is aligned with Assumption \ref{assump:if}, and the corresponding remainder term is uniformly $o_p(n^{-1/2})$ on $[\tau_1,\tau_2]$. Moreover, we will show that total variation of the remainder term is $o_p(1)$. 
Then under the regularity assumption that the remainder term is bounded almost surely, we have that it is  $o(n^{-1/2})$ in terms of $\|\cdot\|_{\sup,2}$ and  $o(1)$ in terms of $\|\cdot\|_{\TV,2}$.

We first derive the asymptotic linear representation of $\tilde \Lambda(\hat\psi, t)$ by looking at the decomposition of $\tilde \Lambda(\hat\psi, t) - \Lambda_0(t)$. 
From \eqref{eq:brewslow}, we have 
\begin{align}
    \tilde \Lambda(\hat\psi, t) - \Lambda_0(t) =  A_1(t) + A_2(t), \label{eq:Lambda_A1_A2}
\end{align}
where
\begin{align}
    A_1(t) = \tilde \Lambda(\hat\psi, t) - \tilde \Lambda(\psi_0, t),\quad
    A_2(t) = \tilde \Lambda(\psi_0, t) - \Lambda_0(t).
\end{align}
Following \citet{tsiatis1981large}, it can be shown that
\begin{align}
    \Lambda_0(t) & = \int_0^t \frac{dB(x)}{\E(e^{{\psi_0}'z}, x)}. \label{eq:expr_breslow}
\end{align}


To apply the conclusions in \citet{lopuhaa2013asymptotic}, we consider the following decompositions of $A_1$: 
\begin{align}
    A_1(t) = (\hat\psi - \psi_0)'A_0(t) + R_{10}(t), \label{eq:A_1}
\end{align}
where $A_0(t)$ is the deterministic function given in \citet{lopuhaa2013asymptotic} equation (8), and
$$R_{10}(t) = A_1(t) - (\hat\psi - \psi_0)'A_0(t)$$ 
It can be verified that $\sup_{t\in[\tau_1,\tau_2]} A_0(t) = O(1)$ and $\TV(A_0) = O(1)$.
Denote 
\begin{align}
    R_{11}(t) = \left\{\hat \psi - \psi_0 - \frac{1}{n}\sum_{i=1}^n \phi(O_i)\right\}'A_0(t). \label{eq:R_11}
\end{align}
Then we have
\begin{align}
     A_1(t) = \frac{1}{n}\sum_{i=1}^n \phi(O_i)'A_0(t) + R_{11}(t)  + R_{10}(t). \label{eq:A1_IF}
\end{align}

On the other hand,
\begin{align}
    A_2(t)
    &= \int_0^t \frac{d\hat B(x)}{\hat \E(e^{\psi_0'z}, x)} - \int_0^t \frac{dB(x)}{\E(e^{\psi_0'z}, x)} \\
    &= \int_0^t \frac{d\hat B(x)}{\hat \E(e^{\psi_0'z}, x)} - \int_0^t \frac{d\hat B(x)}{\E(e^{\psi_0'z}, x)} +\int_0^t \frac{d\hat B(x)}{\E(e^{\psi_0'z}, x)} - \int_0^t \frac{dB(x)}{\E(e^{\psi_0'z}, x)}\\
    &= \int_0^t \frac{d\hat B(x)}{\E(e^{\psi_0'z}, x)}\left\{ \frac{\E(e^{\psi_0'z}, x)}{\hat \E(e^{\psi_0'z}, x)} - 1\right\} + \int_0^t \frac{d\left\{\hat B(x) - B(x)\right\}}{\E(e^{\psi_0'z}, x)} \\
    & = A_{21}(t) + A_{22}(t)+ A_{23}(t), \label{eq:A_2}
\end{align}
where 
\begin{align}
    A_{21}(t) & = \int_0^t \frac{d\left\{\hat B(x) - B(x)\right\}}{\E(e^{\psi_0'z}, x)} ,\\
    A_{22}(t) & = \int_0^t \frac{dB(x)}{\E(e^{\psi_0'z}, x)^2}\left\{ \hat \E(e^{\psi_0'z}, x)- \E(e^{\psi_0'z}, x)\right\},\\
    A_{23}(t) & = \int_0^t \frac{d\hat B(x)}{\E(e^{\psi_0'z}, x)\hat \E(e^{\psi_0'z}, x)}\left\{ \hat \E(e^{\psi_0'z}, x)- \E(e^{\psi_0'z}, x)\right\} \\
    &\quad \quad - \int_0^t \frac{dB(x)}{\E(e^{\psi_0'z}, x)^2}\left\{ \hat \E(e^{\psi_0'z}, x)- \E(e^{\psi_0'z}, x)\right\}.
\end{align}
Denote 
\begin{align}
    \zeta_i(t) &= \phi(O_i)'A_0(t) +  \int_0^t \frac{d\left\{N_i(x) - B(x)\right\}}{\E(e^{\psi_0'z}, x)} \\
    &\quad\quad  + \int_0^t \frac{dB(x)}{\E(e^{\psi_0'z}, x)^2}\left\{Y_i (x)\exp\left(\psi_0'Z_i\right)- \E(e^{\psi_0'z}, x)\right\}, \\
    R(t) &= R_{11}(t) + R_{10}(t) + A_{23}(t).
\end{align}
Then combining \eqref{eq:Lambda_A1_A2} \eqref{eq:A1_IF} \eqref{eq:A_2}, we have 
\begin{align}
    \tilde \Lambda(\hat\psi, t) - \Lambda_0(t) = \frac{1}{n}\sum_{i=1}^n \zeta_i(t) + R(t),
\end{align}
and $\zeta_i(t)$'s are i.i.d random functions of $t$ with mean zero and finite second moment for each $t$.

Next we show that (a) $\sup_{t\in[\tau_1,\tau_2]} |R(t)|= o_p(n^{-1/2})$ and (b) $\TV(R) = o_p(1)$. 
Before proving these, we will first bound the supremum of $1/\hat\E(e^{\hat\psi'z},t)$ and $1/\E(e^{\hat\psi'z},t)$ on $[\tau_1,\tau_2]$, which will be useful in the proof of (a) and (b).
Recall from \eqref{eq:p(q,t|z)} that the conditional density of $(Q,T)$ given $Z$ is
\begin{align}
    p_{Q,T|Z}(q,t|z) & = \frac{\mathbbm{1}(q<t)}{\beta(z)} f(t|z)g(q|z),
\end{align}
where $\beta(z)$ is defined in \eqref{eq:beta(z)}. By definition, $\beta(z)\in[0,1]$ almost surely.
So by Assumption \ref{ass:overlap}, we have
\begin{align}
    \E\{Y(t)\} &= \E\left[\E\left\{\left. \mathbbm{1}(Q\leq t\leq T) \right| Z \right\}\right] = \E\left[ \frac{G(t|Z)\{1-F(t|Z)\}}{\beta(Z)} \right] \\
    &\geq G(\tau_1|Z)\{1-F(\tau_2|Z)\} \geq \delta_1\delta_2>0,
\end{align}
for all $t\in[\tau_1,\tau_2]$.
This together with Assumption \ref{ass:regularity_cox} imply that there exist $\delta_0>0$ such that $\inf_{t\in[\tau_1,\tau_2]}\E(e^{\psi_0'z},t) \geq \delta_0$.
In addition, by Assumption \ref{ass:regularity_cox} and the central limit theorem, we have
\begin{align}
    \sup_{t\in[\tau_1,\tau_2]}|\hat\E(e^{\psi_0'z},t) - \E(e^{\psi_0'z},t)| 
    = O_p(n^{-1/2})
    = o_p(1). \label{eq:E_hat-E}
\end{align}
Also $\hat\psi - \psi = O_p(n^{-1/2})$, so 
\begin{align}
    &\quad \sup_{t\in[\tau_1,\tau_2]}|\hat\E(e^{\hat\psi'z},t) - \E(e^{\psi_0'z},t)|\\
    & \leq \sup_{t\in[\tau_1,\tau_2]}|\hat\E(e^{\hat\psi'z},t) - \hat E(e^{\psi_0'z},t)| + \sup_{t\in[\tau_1,\tau_2]}|\hat\E(e^{\psi_0'z},t) - \E(e^{\psi_0'z},t)| \\
    & \leq \frac{1}{n}\sum_{j=1}^n \left|\exp(\hat \psi'Z_j) - \exp(\psi_0'Z_j)\right| + O_p(n^{-1/2}) \\
    & = O_p(n^{-1/2}) + O_p(n^{-1/2}) \\
    &= o_p(1).
\end{align}
Therefore, we have
\begin{align}
    \sup_{t\in[\tau_1,\tau_2]}\left|\frac{1}{\hat \E(e^{\psi_0'z},t)}\right| = O_p(1), \quad \sup_{t\in[\tau_1,\tau_2]}\left|\frac{1}{\hat\E(e^{\hat\psi'z},t)}\right| = O_p(1). \label{eq:Et_bdd}
\end{align}

(a) We now show that $\sup_{t\in[\tau_1,\tau_2]} |R(t)|= o_p(n^{-1/2})$. We have
\begin{align}
    \sup_{t\in[\tau_1,\tau_2]} |R(t)| \leq \sup_{t\in[\tau_1,\tau_2]} |R_{10}(t)+A_{23}(t)| + \sup_{t\in[\tau_1,\tau_2]} |R_{11}(t)|.
\end{align}
\citet{lopuhaa2013asymptotic} showed that 
$$\sup_{t\in[\tau_1,\tau_2]} |R_{10}(t)+A_{23}(t)| = o_p(n^{-1/2}).$$ 
In addition, we have $\hat \psi - \psi_0 - n^{-1}\sum_{i=1}^n \phi(O_i) = o_p(n^{-1/2})$ by \eqref{eq:beta_AL}, and $\sup_{t\in[\tau_1,\tau_2]} |A_0(t)| = O(1)$. So from \eqref{eq:R_11},
$$\sup_{t\in[\tau_1,\tau_2]} |R_{11}(t)| = o_p(n^{-1/2}).$$ 
Therefore, 
$\sup_{t\in[\tau_1,\tau_2]} |R(t)| = o_p(n^{-1/2})$.

\medskip
(b) We next show that $\TV(R) = o_p(1)$. We have
\begin{align}
    \TV(R) \leq \TV(R_{11}+R_{10}) + \TV(A_{23}).
\end{align}
We first consider $R_{11}+R_{10}$. Denote $\tilde A_1(t) = n^{-1}\sum_{i=1}^n\phi(O_i)' A_0(t)$. Then from  \eqref{eq:A1_IF}, $R_{11}+R_{10} = A_1(t) - \tilde A_1(t)$, so
\begin{align}
    \TV(R_{11}+R_{10}) \leq \TV(A_1) + \TV(\tilde A_1).
\end{align}
For $A_1(t)$, we have
\begin{align}
    A_1(t) & =  \int_0^t \frac{d\hat B(x)}{\hat \E(e^{\hat\psi'z},x)}- \int_0^t \frac{d\hat B(x)}{\hat \E(e^{\psi_0'z},x)} \\
    & = \int_0^t \frac{d\hat B(x)}{\hat \E(e^{\hat\psi'z},x)\hat \E(e^{\psi_0'z},x)} \left\{ \hat \E(e^{\hat\psi'z},x) - \hat \E(e^{\psi_0'z},x) \right\} \\
    & = \int_0^t \frac{d\hat B(x)}{\hat \E(e^{\hat\psi'z},x)\hat \E(e^{\psi_0'z},x)} \left[ \frac{1}{n}\sum_{j=1}^n Y_j(x) \left\{\exp(\hat\psi'Z_j) - \exp(\psi_0'Z_j)\right\} \right].
\end{align}
Recall that $\hat\psi - \psi_0 = O_p(n^{-1/2})$, so by Assumption \ref{ass:regularity_cox}, 
\begin{align}
    \sup_{j=1,...,n} \left|\exp(\hat\psi'Z_j) - \exp(\psi_0'Z_j)\right| = O_p(n^{-1/2}) = o_p(1).
\end{align}
Since $\hat B$ and $B$ are both non-decreasing functions taking values in $[0,1]$, we have $\TV(\hat B)\leq 1$ and $\TV(B)\leq 1$.
So by \eqref{eq:Et_bdd}, 
\begin{align}
    \TV(A_1) 
    & \leq \left[\sup_{x\in[\tau_1,\tau_2]} \left|\frac{ n^{-1} \sum_{j=1}^n Y_j(x) \left\{\exp(\hat\psi'Z_j) - \exp(\psi_0'Z_j) \right\}}{\hat \E(e^{\hat\psi'z},x)\hat \E(e^{\psi_0'z},x)} \right| \right] \TV(\hat B)\\
    &\lesssim \left\{ \sup_{j=1,...,n} \left|\exp(\hat\psi'Z_j) - \exp(\psi_0'Z_j)\right|  + o_p(1)\right\} \\
    &= o_p(1).
\end{align}
In addition, since $n^{-1}\sum_{i=1}^n \phi(O_i) = O_p(n^{-1/2})$ by central limit theorem and $\TV(A_0) = O_p(1)$, we have
\begin{align}
    \TV(\tilde A_1) & = \left\{\frac{1}{n}\sum_{i=1}^n \phi(O_i) \right\}' \TV(A_0) = O_p(n^{-1/2}) = o_p(1).
\end{align}
Putting the above together, we have $\TV(R_{11}+R_{10}) = o_p(1)$.

For $A_{23}$, we consider the following decomposition: 
\begin{align}
    A_{23}(t) & = A_{231}(t) - A_{232}(t),
\end{align}
where
\begin{align}
    A_{231}(t) & = \int_0^t \frac{d\left\{\hat B(x)-B(x)\right\}}{\E(e^{\psi_0'z}, x)\hat \E(e^{\psi_0'z}, x)}\left\{ \hat \E(e^{\psi_0'z}, x)- \E(e^{\psi_0'z}, x)\right\} ,
    \\
    A_{232}(t) &=  \int_0^t \frac{dB(x)}{\E(e^{\psi_0'Z}, x)^2 \hat \E(e^{\psi_0'z}, x)}\left\{ \hat \E(e^{\psi_0'z}, x)- \E(e^{\psi_0'z}, x)\right\}^2 .
\end{align}
By \eqref{eq:Et_bdd} and the fact that $\TV(\hat B)\leq 1$ and $\TV(B)\leq 1$, we have 
\begin{align}
    \TV(A_{231}) & \leq \left\{\sup_{x\in[\tau_1,\tau_2]} \left|\frac{\hat \E(e^{\psi_0'z}, x)- \E(e^{\psi_0'z}, x)}{\E(e^{\psi_0'z}, x)\hat \E(e^{\psi_0'z}, x)} \right|\right\}  \TV(\hat B-B) \\
    & \lesssim \left\{\sup_{x\in[\tau_1,\tau_2]} \left|\hat \E(e^{\psi_0'z}, x)- \E(e^{\psi_0'z}, x)\right|+ o_p(1)\right\} \left\{ \TV(\hat B) + \TV(B) \right\} \\
    & = o_p(1), \\
    \TV(A_{232}) & \leq \left[\sup_{x\in[\tau_1,\tau_2]}\left|\frac{\left\{ \hat \E(e^{\psi_0'z}, x)- \E(e^{\psi_0'z}, x)\right\}^2}{\E(e^{\psi_0'z}, x)^2 \hat \E(e^{\psi_0'z}, x)} \right| \right]  \TV(B) \\
    & \lesssim \left\{\sup_{x\in[\tau_1,\tau_2]} \left| \hat \E(e^{\psi_0'z}, x)- \E(e^{\psi_0'z}, x) \right|^2 + o_p(1)\right\} \  \TV(B)  \\
    & = o_p(1).
\end{align}
Therefore, $\TV(A_{23})= o_p(1)$.

Putting the above together, we have $\TV(R)= o_p(1)$.

\newpage
\section{Estimation under right censoring}

This section contains proofs of double robustness of the  estimating functions $U_{c1}$ and $U_{c2}$, and expressions of the doubly robust estimators as well as the extended `IPW.Q' and regression-based estimators.

\subsection{Censoring before truncation}

We first show that $U_{c1}$ is doubly robust in the sense that $\E\left\{U_{c1}(\theta_0; F_x,G, S_{0c})\right\} = 0$ if $F_x = F_{0x}$ or $G = G_0$.
By Assumption \ref{ass:quasi-indpendent} and $C\bigCI(Q,T,Z)$, we have that $Q$ and $X$ are conditionally quasi-independent given $(Z,C)$. Recall that $\Delta = \mathbbm{1}(T<C)$. With a very similar proof as that of Theorem \ref{thm:DR}, we can show that when $F_x = F_{0x}$ or $G = G_0$, 
\begin{align}
    \E\left\{U_{c1}(\theta; F_x,G, S_{0c})\right\} 
    = \E^*\left[\frac{\Delta\{\nu(X) - \theta\}}{S_{0c}(X)}\right].\label{eq:c1_DR_1}
\end{align}
Since $\Delta = \mathbbm{1}(T<C)$, $\Delta \neq 0$ implies $X = T$. So we have
\begin{align}
    \E^*\left[\frac{\Delta\{\nu(X) - \theta\}}{S_{0c}(X)}\right]
    = \E^*\left[\frac{\Delta\{\nu(T) - \theta\}}{S_{0c}(T)}\right]
    = \E^*\left[\E^*\left\{\left.\frac{\Delta}{S_{0c}(T)}\right|T\right\}\cdot \{\nu(T) - \theta\}\right]. \label{eq:c1_DR_2}
\end{align}
Since $C\bigCI (Q,T,Z)$ in the full data, we have 
\begin{align}
    \E^*\left\{\left.\frac{\Delta}{S_{0c}(T)}\right|T\right\} 
     = \E^*\left\{\left.\frac{\mathbbm{1}(C>T)}{S_{0c}(T)}\right|T\right\} = 1. \label{eq:c1_DR_3}
\end{align}
Combining \eqref{eq:c1_DR_1} \eqref{eq:c1_DR_2} \eqref{eq:c1_DR_3}, we have
\begin{align}
    \E\left\{U_{c1}(\theta_0; F_x,G, S_{0c})\right\} 
    = \E^*\{\nu(T) - \theta_0\} = 0,
\end{align}
which ends the proof.

\bigskip
For this censoring scenario, estimators of $\theta$ can be constructed by first estimating the nuisance parameters $F_x$, $G$ and $S_c$, and then solve 
\begin{align}
    \sum_{i=1}^n U_{c1,i}(\theta; \hat F_x, \hat G, \hat S_c) = 0 \label{eq:ee_c1}
\end{align}
for $\theta$. Since \eqref{eq:ee_c1} is linear in $\theta$, we have an explicit solution
\begin{align}
    \hat\theta 
    &= \left[\sum_{i=1}^n \left\{\frac{\Delta_i}{\hat S_c(X_i) \hat G(X_i|Z_i)} - \int_0^\infty \frac{\int_0^v \Delta_i/\hat S_c(x) d \hat F_x(x|Z_i)}{ 1-\hat F_x(v|Z_i) }\cdot \frac{d\tilde M_{Q,i}(v;\hat G)}{\hat G(v|Z_i)}\right\}\right]^{-1} \\
    &\quad\quad\quad \times \left[\sum_{i=1}^n \left\{\frac{\Delta_i\nu(X_i)}{\hat S_c(X_i)\hat G(X_i|Z_i)} - \int_0^\infty \frac{\int_0^v \Delta_i\nu(x)/\hat S_c(x) d\hat F_x(x|Z_i)}{ 1-\hat F_x(v|Z_i) } \cdot \frac{d\tilde M_{Q,i}(v;\hat G)}{G(v|Z_i)} \right\} \right]
\end{align}

When plugging constant estimates $\hat F_x\equiv 0$ or $\hat G \equiv 1$ into $\hat\theta$, we get extensions for the `IPW.Q' and `Reg.T1' estimators:
\begin{align}
    \hat\theta_{\text{IPW.Q}} &= \left.\left\{\sum_{i=1}^n\frac{\Delta_i\nu(X_i)}{\hat S_c(X_i|Z_i) \hat G(X_i|Z_i)}\right\} \right/ \left\{\sum_{i=1}^n\frac{\Delta_i}{\hat S_c(X_i|Z_i)\hat G(X_i|Z_i)}\right\}, \label{eq:IPW.Q_c1}\\
    \hat\theta_{\text{Reg.T1}}
& = \left[\sum_{i=1}^n \Delta_i \left\{\frac{1}{\hat S_c(X_i)}+\frac{\int_0^{Q_i} 1/\hat S_c(t) d\hat F_x(t|Z_i)}{1-\hat F_x(Q_i|Z_i)}\right\}\right]^{-1} \\
&\quad\quad\quad \times\left[\sum_{i=1}^n \Delta_i \left\{\frac{\nu(X_i)}{\hat S_c(X_i)}+\frac{\int_0^{Q_i} \nu(t)/\hat S_c(t) d\hat F_x(t|Z_i)}{1-\hat F_x(Q_i|Z_i)}\right\}\right] . 
\end{align}
Likewise, we have an extension for the `Reg.T2' estimator:
\begin{align}
\hat\theta_{\text{Reg.T2}}
& = \left. \left\{\sum_{i=1}^n \Delta_i\frac{\int_0^{\infty} \nu(t)/\hat S_c(t) d\hat F_x(t|Z_i)}{1-\hat F_x(Q_i|Z_i)}\right\} \right/ \left\{\sum_{i=1}^n \Delta_i \frac{\int_0^{\infty} 1/\hat S_c(t) d\hat F_x(t|Z_i)}{1-\hat F_x(Q_i|Z_i)}\right\} . 
\end{align}

\vskip .15in
\subsection{Censoring after truncation}

We first show that $U_{c2}$ is doubly robust in the sense that 
$\E\left\{U_{c2}(\theta_0;F,G,S_{0D})\right\} = 0$ if $F = F_0$ or $G = G_0$.
Since $\Delta \neq 0$ implies $X = T$, we have
\begin{align}
    U_{c2}(\theta;F,G,S_{0D}) 
    &= \frac{\Delta}{S_{0D}(X-Q)} \left[\frac{\nu(X)-\theta}{G(X|Z)} - \int_0^\infty \frac{\int_0^v \{\nu(t) -\theta\} dF(t|Z) }{ 1-F(v|Z) } \cdot \frac{d\tilde M_Q(v;G)}{G(v|Z)} \right]\\
    & = \frac{\Delta}{S_{0D}(X-Q)} \left[\frac{\nu(T)-\theta}{G(T|Z)} - \int_0^\infty \frac{\int_0^v \{\nu(t) -\theta\} dF(t|Z) }{ 1-F(v|Z) } \cdot \frac{d\bar M_Q(v;G)}{G(v|Z)} \right].
\end{align}
Therefore, 
\begin{align}
    &\quad \E\{U_{c2}(\theta;F,G,S_{0D})\} \\
    & = \E\left[\E\left\{U_{c2}(\theta;F,G,S_{0D})|Q,T,Z\right\}\right] \\
    & = \E\left[\E\left\{\left.\frac{\Delta}{S_{0D}(X-Q)} \right|Q,T,Z\right\} \cdot \left[\frac{\nu(T)-\theta}{G(T|Z)} - \int_0^\infty \frac{\int_0^v \{\nu(t) -\theta\} dF(t|Z) }{ 1-F(v|Z) } \cdot \frac{d\bar M_Q(v;G)}{G(v|Z)} \right]\right] \\
\end{align}
Since $D\bigCI(Q,T,Z)$ in the observed data, we have
\begin{align}
    \E\left\{\left.\frac{\Delta}{S_{0D}(X-Q)} \right|Q,T,Z\right\} 
    & = \E\left\{\left.\frac{\mathbbm{1}(D>X-Q)}{S_{D}(X-Q)} \right|Q,T,Z\right\} = 1.
\end{align}
So
\begin{align}
    \E\{U_{c2}(\theta;F,G,S_{0D})\} 
    &= \E\left[\frac{\nu(T)-\theta}{G(T|Z)} - \int_0^\infty \frac{\int_0^v \{\nu(t) -\theta\} dF(t|Z) }{ 1-F(v|Z) } \cdot \frac{d\bar M_Q(v;G)}{G(v|Z)} \right] \\
    & = \E\{U(\theta;F,G)\}.
\end{align}
Therefore, the conclusion follows from Theorem \ref{thm:DR}.

\bigskip
For this censoring scenario, estimators of $\theta$ can be constructed by first estimating the nuisance parameters $F$, $G$ and $S_D$, and then solve 
\begin{align}
    \sum_{i=1}^n U_{c2,i}(\theta; \hat F,\hat G, \hat S_D) = 0 \label{eq:ee_c2}
\end{align}
for $\theta$. Since \eqref{eq:ee_c2} is linear in $\theta$, we have an explicit solution
\begin{align}
    \hat\theta
    &= \left[ \sum_{i=1}^n \frac{\Delta_i}{\hat S_{D}(X_i-Q_i)} \left\{\frac{1}{\hat G(X_i|Z_i)} - \int_0^\infty \frac{\hat F(v|Z_i)}{1-\hat F(v|Z_i)} \cdot \frac{d\tilde M_{Q,i}(v;\hat G)}{\hat G(v|Z_i)} \right\} \right]^{-1} \\
    &\quad\quad\quad \times \left[\sum_{i=1}^n \frac{\Delta_i}{\hat S_{D}(X_i-Q_i)} \left\{\frac{\nu(X_i)}{\hat G(X_i|Z_i)} - \int_0^\infty \frac{\int_0^v \nu(t)  d\hat F(t|Z_i) }{ 1-\hat F(v|Z_i) } \cdot \frac{d\tilde M_{Q,i}(v;\hat G)}{\hat G(v|Z_i)} \right\} \right]
\end{align}

Again, when plugging constant estimates $\hat F \equiv 0$ or $\hat G\equiv 1$ into $\hat\theta$, we get extensions of the IPW.Q and Reg.T1 estimator for this censoring scenario:
\begin{align}
    \hat\theta_{\text{IPW.Q}} 
    &=  \left. \left\{\sum_{i=1}^n \frac{\Delta_i \ \nu(X_i)}{\hat S_{D}(X_i-Q_i)\hat G(X_i|Z_i)}\right\} \right/ \left\{ \sum_{i=1}^n \frac{\Delta_i}{\hat S_{D}(X_i-Q_i)\hat G(X_i|Z_i)} \right\}, \\
\hat\theta_{\text{Reg.T1}}
& = \left\{\sum_{i=1}^n \frac{\Delta_i}{\hat S_c(X_i)\{1-\hat F(Q_i|Z_i)\}}\right\}^{-1} \\
&\quad\quad\quad \times \left[\frac{1}{n}\sum_{i=1}^n \Delta_i \frac{\nu(T_i)\{1-\hat F(Q_i|Z_i)\} + \int_0^{Q_i} \nu(t) d\hat F(t|Z_i) }{\hat S_c(X_i)\{1-\hat F(Q_i|Z_i)\}}\right].  
\end{align}
Likewise, we have an extension for the Reg.T2 estimator
\begin{align}
\hat\theta_{\text{Reg.T2}} = \left.\left[ \sum_{i=1}^n \frac{\Delta_i \int_0^{\infty} \nu(t) d\hat F(t|Z_i) }{\hat S_c(X_i)\{1-\hat F(Q_i|Z_i)\}} \right] \right/ \left[\sum_{i=1}^n \frac{\Delta_{i}}{\hat S_c(X_i)\{1- \hat F(Q_i|Z_i)\}}\right].
\end{align}


\newpage
\section{Additional materials for simulations}

This section contains details of the data generating mechanisms for the seven scenarios in Section \ref{sec:simu_mdr}, plots for visualizing the simulation results in the main paper, as well as additional simulation results under right censoring.

\subsection{Details of the data generating mechanisms in Scenarios 1 - 7}

\begin{itemize}
    \item Scenario 1: the full data distribution of $T$ given $Z$ and $(\tau-Q)$ given $Z$ follow the following Cox models, respectively.
    \begin{align}
\lambda_1(t|Z_1,Z_2) & = \lambda_{01}(t) e^{0.3 Z_1+0.5 Z_2}, \label{eq:T_cox1}\\
\lambda_2(t|Z_1,Z_2) &= \lambda_{02}(t) e^{0.3 Z_1+0.5 Z_2},\label{eq:Q_cox1}
\end{align}
where the baseline hazards
\begin{align}
    \lambda_{01}(t) &= \left\{
    \begin{array}{ll}
         0, & \text{if  }0<t<\tau_1, \\
         2e^{-1}(t-\tau_1), &  \text{if  } t\geq \tau_1,
    \end{array}\right. \label{eq:T_baseline1} \\
    \lambda_{02}(t) &= \left\{
    \begin{array}{ll}
         0, & \text{if  } 0\leq t < \tau-\tau_2, \\
         2e^{-1}\{t-(\tau-\tau_2)\}, & \text{if  } t\geq \tau-\tau_2.
    \end{array}\right. \label{eq:Q_baseline1}
\end{align}
We note that in this case, $(T-\tau_1)$ follows a Weibull distribution with shape parameter $2$ and scale parameter $e^{(1-0.3Z_1-0.5Z_2)/2}$, and $(\tau_2-Q)$ follows a Weibull distribution with shape parameter $2$ and scale parameter $e^{(1-0.3Z_1-0.5Z_2)/2}$.
    
    \item Scenario 2: $T$ given Z follows the Cox model in \eqref{eq:T_cox1} and \eqref{eq:T_baseline1}, and $(\tau-Q)$ given $Z$ follows the following Cox model with quadratic and interaction terms: 
    \begin{align}
        \lambda_2(t|Z_1,Z_2) 
        &= \lambda_{02}(t) e^{0.3Z_1+0.5Z_2 +0.6(Z_1^2-1/3)+0.5Z_1Z_2},\label{eq:Q_cox2}
    \end{align}
    where the baseline hazard $\lambda_{02}(t)$ is given in \eqref{eq:Q_baseline1}. Under this model, $(\tau_2-Q)$ follows a Weibull distribution with shape parameter $2$ and scale parameter $e^{\{1-0.3Z_1-0.5Z_2 -0.6(Z_1^2-1/3)-0.5Z_1Z_2\}/2}$.
    
    \item Scenario 3: $T$ given Z follows the Cox model in \eqref{eq:T_cox1} and \eqref{eq:T_baseline1}, and $(\tau-Q)$ follows a mixture model of a Cox model with quadratic and interaction terms and an AFT model with quadratic and interaction terms. Specifically, if $0.3Z_1+0.5Z_2 <0$, $(\tau-Q)$ given $Z$ follows the Cox proportional hazards model with quadratic and interaction terms given in \eqref{eq:Q_cox2} and \eqref{eq:Q_baseline1}, otherwise, $(\tau_2-Q)$ follows the following AFT model:
    \begin{equation}
    \begin{split}
        \log(\tau_2-Q) &= -1 +  0.3Z_1+0.5Z_2 +0.6(Z_1^2-1/3)+0.5Z_1Z_2 + \epsilon_2, \\
        \epsilon_2 &\sim \text{Weibull}(1.5,1) - \Gamma(5/3). \label{eq:Q_AFT}
    \end{split}
    \end{equation}
    
    \item Scenario 4: $(\tau-Q)$ given Z follows the Cox model in \eqref{eq:Q_cox1} and \eqref{eq:Q_baseline1}, and $T$ given $Z$ follows a Cox model with quadratic and interaction terms stated as follows. 
    \begin{align}
        \lambda_{1}(t|Z_1,Z_2) 
        &= \lambda_{01}(t) e^{0.3Z_1+0.5Z_2 +0.6(Z_1^2-1/3)+0.5Z_1Z_2},\label{eq:T_cox2}
    \end{align}
    where the baseline hazard $\lambda_{01}(t)$ is given in \eqref{eq:T_baseline1}. Under this model, $T$ follows a Weibull distribution with shape parameter $2$ and scale parameter $e^{(1-0.3Z_1-0.5Z_2 -0.6(Z_1^2-1/3)-0.5Z_1Z_2)/2}$.
    
    \item Scenario 5: $(\tau-Q)$ given Z follows the Cox model in \eqref{eq:Q_cox1} and \eqref{eq:Q_baseline1}, and $T$ given $Z$ follows a mixture model of a Cox model with quadratic and interaction terms and an AFT model with quadratic and interaction terms. Specifically, if $0.3Z_1+0.5Z_2 \geq 0$, $T$ given $Z$ follows the Cox model with quadratic and interaction terms given in \eqref{eq:T_cox2} and \eqref{eq:T_baseline1}; otherwise, $(T-\tau_1)$ follows the following AFT model: 
    \begin{align}
        \log(T-\tau_1) = -1 + 0.3Z_1+0.5Z_2 +0.6(Z_1^2-1/3)+0.5Z_1Z_2 + \epsilon_1, \quad
        \epsilon_1 \sim N(0,1). \label{eq:T_AFT}
    \end{align}
    
    \item Scenario 6: $T$ given Z follows the Cox model with quadratic and interaction terms in \eqref{eq:T_cox2} and \eqref{eq:T_baseline1}, and $(\tau-Q)$ given $Z$ follows the Cox model with quadratic and interaction terms in \eqref{eq:Q_cox2} and \eqref{eq:Q_baseline1}.
    
    \item Scenario 7: $T$ given Z follows the mixture model stated in Scenario 5, and $(\tau-Q)$ follows the mixture model stated in Scenario 3.
    
\end{itemize}

\subsection{Visualization of simulation results in the main paper}

\clearpage
\begin{figure}[H]
	\centering
		\includegraphics[width=1\textwidth]{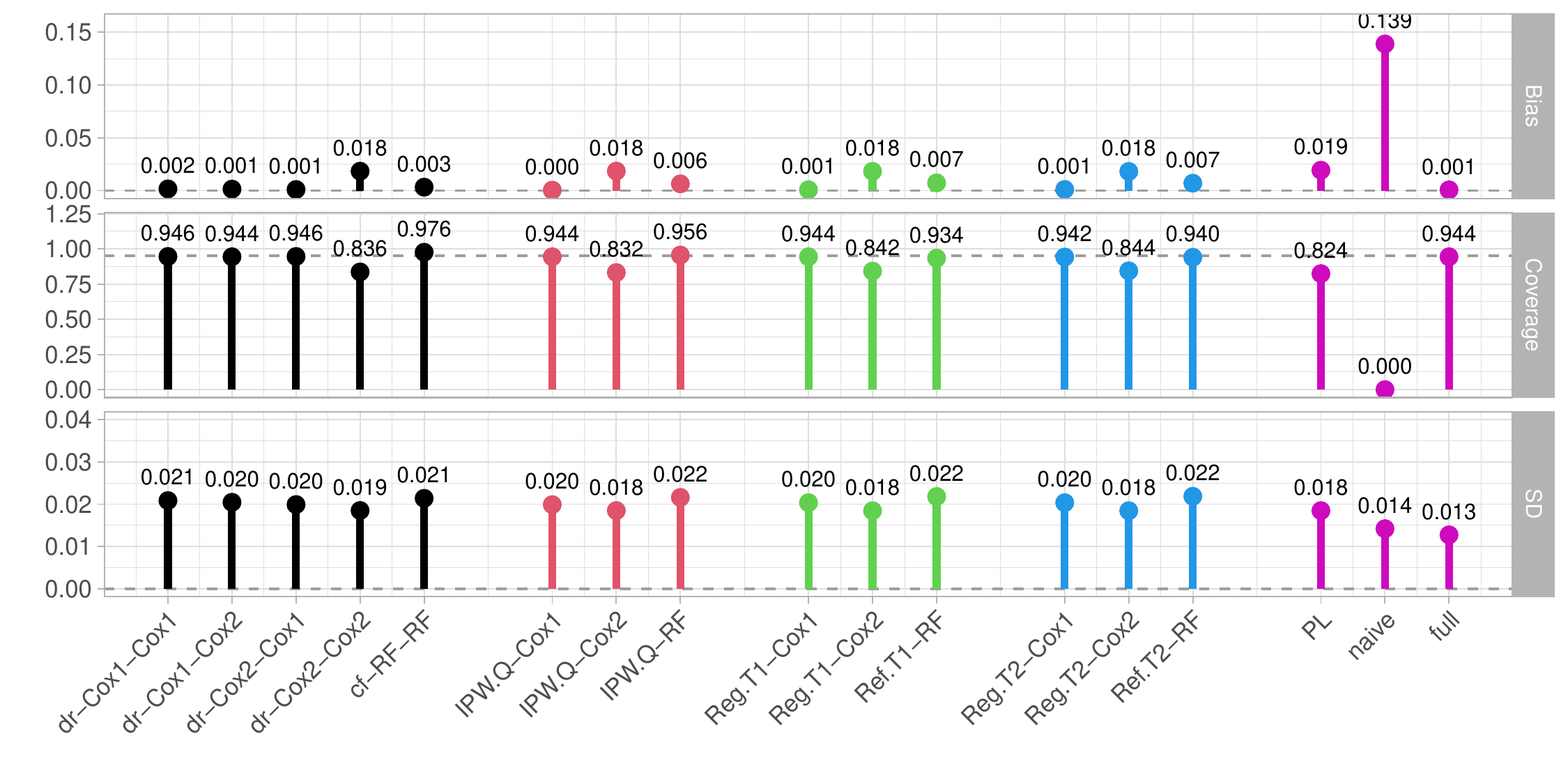}
	\caption{Absolute bias, coverage rate of  95\% confidence intervals  with bootstrap standard errors, and empirical standard deviation for different estimator, corresponding to the simulation results in Section \ref{sec:simu_rdr} of the main paper.}
	\label{fig:simu}
\end{figure}

\subsection{Simulation under right censoring}

\vskip .2in
\noindent \underline{\it Censoring before truncation}
\medskip

We generate the full data $(Q,T,C,Z)$, where the variables $Z= (Z_1,Z_2)$ and $Q$ are generated in the same way as in the simulation section of the main paper. 
The event time variable $T$ is generated such that $(T -1)$ follows Weibull distribution with shape parameter 2 and scale parameter $\{\exp(-2+0.3Z_1+0.5Z_2) - 7^{-2}\}^{-1/2}$, 
and the censoring variable $C$ is generated such that $(C-1)$ follows Weibull distribution with shape parameter $2$ and scale parameter $7$. The censored event time $X = \min(T,C)$, and the event indicator $\Delta = \mathbbm{1}(T<C)$. 
Then subjects with $Q < X$ are included in the observed sample. 
The resulting left truncation rate is 29.5\%, censoring rate is 16.5\% in the observed data, and $\PP^*(C<Q) = 0.061$.

Our estimand is $\theta = \PP^*(T>3)=0.624$, computed from a simulated full data sample with sample size $10^7$. The nuisance parameter $F_x$ and $G$ are estimated from models `Cox1', `Cox2' or `RF' with details given in the main paper. 
It can be shown that under this data generating mechanism, $X$ in the full data follows a Cox proportional hazards model
\begin{align}
    \lambda_{X}(t) = \lambda_{0X}(t) e^{0.3Z_1+0.5Z_2},
\end{align}
where the baseline hazard $\lambda_{0X}(t) = 2e^{-1}(t-1)$ if $t\geq 1$, and 0 otherwise. So `Cox1' is the correct model for $F$ and $G$ while `Cox2' is misspecified for $F$ and $G$. We use the same notation for different estimators and report the same quantities for simulation results as in the main paper. 
We consider $\hat\theta_{dr}$, $\hat\theta_{cf}$, as well as the IPW and the regression-based estimators.
We also consider the product-limit (PL) estimator \citep{wang1991nonparametric} that assumes random left truncation, the naive Kaplan-Meier estimator that ignores left truncation, and the oracle full data estimator. 

The simulation results are summarized in Table \ref{tab:c1_Xcox} and visualized in Figure \ref{fig:c1_X_simu}. 
As expected, $\hat\theta_{dr}$ has small bias and close to nominal coverage rate with boot SE when the model for $G$ is `Cox1'. We also see that `cf-RF-RF' has a relatively small bias and close to nominal coverage rates with boot SE. In addition, $\hat\theta_{\text{IPW.Q}}$ has a large bias when the model for $G$ is misspecified, and the regression-based estimators have large biases when the model for $F_x$ is misspecified. Furthermore, the PL estimator that assumes random left truncation has a large bias, and the naive estimator that ignores left truncation substantially over-estimates the parameter of interest.

\begin{table}[H]
 \centering
 \renewcommand{\arraystretch}{0.6}
	\caption{Simulation results for different estimators under the  `truncation before censoring' scenario. Each observed data set has sample size 1000, and 500 data sets are simulated. SD: standard deviation, SE: standard error, CP: coverage probability.}\label{tab:c1_Xcox}
		\begin{tabular}{lrcrr}
        \\
        \toprule
		Estimator & bias  & SD & SE/boot SE & CP/boot CP \\ 
		\midrule
        dr-Cox1-Cox1 & -0.0004  & 0.022 & 0.021/0.021 & 0.954/0.952 \\ 
        dr-Cox1-{\color{red}Cox2} & -0.0002  & 0.021 & 0.020/0.020 & 0.934/0.950 \\
        dr-{\color{red}Cox2}-Cox1 & -0.0002 & 0.021 & 0.020/0.020 & 0.946/0.950 \\ 
        dr-{\color{red}Cox2}-{\color{red}Cox2} & 0.0194  & 0.019 & 0.019/0.019 & 0.812/0.806 \\ 
        cf-RF-RF & 0.0060  & 0.023 & 0.025/0.026 & 0.952/0.956 \\
        \midrule
    IPW.Q-Cox1 & 0.0006  & 0.021 & 0.019/0.020 & 0.932/0.948 \\ 
    IPW.Q-{\color{red}Cox2} & 0.0212  & 0.019 & 0.018/0.018 & 0.758/0.768 \\ 
    IPW.Q-RF & -0.0053  & 0.023 & 0.020/0.023 & 0.902/0.952 \\ 
  \midrule
    Reg.T1-Cox1 & 0.0002  & 0.020 & - /0.020 & - /0.948 \\ 
    Reg.T1-{\color{red}Cox2} & 0.0208  & 0.019 &  - /0.018 & - /0.762 \\
    Reg.T1-RF & -0.0083  & 0.022 & - /0.022 & - /0.926 \\ 
  \midrule
    Reg.T2-Cox1 & 0.0015 & 0.020 & - /0.019 & - /0.950 \\ 
    Reg.T2-{\color{red}Cox2} & 0.0252 & 0.018 & - /0.017 & - /0.658 \\
    Reg.T2-RF & -0.0054 & 0.022 & - /0.022 & - /0.940 \\
  \midrule
  PL & 0.0223  & 0.019 & - /0.018 & - /0.756 \\ 
  naive & 0.1293  & 0.014 & - /0.014 & - /0.000 \\  
  full data & 0.0002 & 0.013 & 0.013/0.013 & 0.958/0.952 \\ 
  \bottomrule
  \end{tabular}
\end{table}

\begin{figure}[H]
	\centering
		\includegraphics[width=1\textwidth]{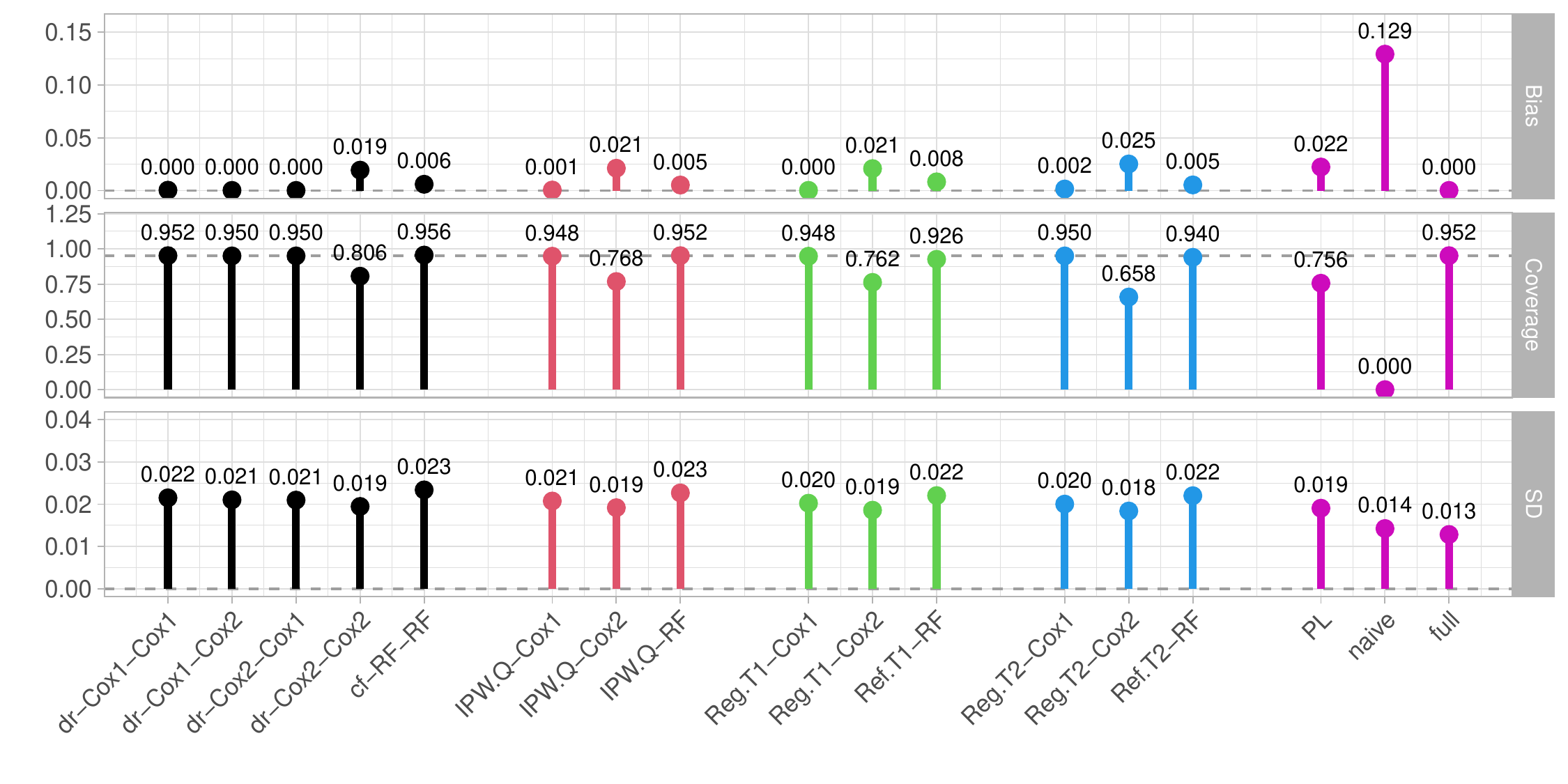}
	\caption{Absolute bias, coverage rate of  95\% confidence intervals  with bootstrap standard errors, and empirical standard deviation for different estimator, corresponding to the `censoring before truncation' scenario.}
	\label{fig:c1_X_simu}
\end{figure}

\vskip .2in
\noindent \underline{\it Censoring after truncation}
\medskip


We generate the full data $(Q,T,C,Z)$, where the variables $Z= (Z_1,Z_2)$, $Q$ and $T$ are generated in the same way as in the simulation section of the main paper. The residual censoring variable $D$ is generated from a Weibull distribution with shape parameter 2 and scale parameter 4. 
The subjects are included in the observed sample only if $Q < T$, and then $T$ is right censored by $(Q+D)$, so the censored event time $X = \min(T,Q+D)$ and the event indicator $\Delta = \mathbbm{1}(T<Q+D)$.
Under this data generating mechanism, the left truncation rate is 29.5\%, and the censoring rate in the observed data is 27.1\%. 

Our estimand is $\theta = \PP^*(T>3)=0.577$, computed from a simulated full data sample with sample size $10^7$. The nuisance parameter $F$ and $G$ are estimated from models `Cox1' or `Cox2' with details given in the main paper.
It is straight forward to see that `Cox1' is the correct model while `Cox2' is the wrong model for both $F$ and $G$ under the data generating mechanism. We use the same notation for the different estimators and report the same quantities for simulation results as in the main paper.
We consider $\hat\theta_{dr}$ as well as the IPW and the regression-based estimators.
We do not implement $\hat\theta_{cf}$ because there is no existing R package, to the best of our knowledge, that incorporates case weights for random forests that handle left truncated data.
We also consider the PL estimator that assumes random left truncation, the naive Kaplan-Meier estimator that ignores left truncation, and the oracle full data estimator. 

The simulation results are summarized in Table \ref{tab:c2} and visualized in Figure \ref{fig:c2_simu}. 
As expected, $\hat\theta_{dr}$ has small bias and close to nominal coverage rate with boot SE when at least one of the models for $F$ and $G$ is correctly specified. In addition, $\hat\theta_{\text{IPW.Q}}$ has a large bias when the model for $G$ is misspecified, and the regression-based estimators have large biases when the model for $F$ is misspecified. Furthermore, the PL estimator that assumes random left truncation has a large bias, and the naive estimator that ignores left truncation substantially over-estimates the parameter of interest.

\begin{table}[H]
    \centering
    \renewcommand{\arraystretch}{0.6}
	\caption{Simulation results for different estimators under the `censoring after truncation' scenario. Each observed data set has sample size 1000, and 500 data sets are simulated. SD: standard deviation, SE: standard error, CP: coverage probability.}
 \label{tab:c2}
		\begin{tabular}{lrcrr}
  \\
  \toprule
		Estimator & bias  & SD & SE/boot SE & CP/boot CP \\ 
		\midrule
        dr-Cox1-Cox1 & -0.0057 & 0.023 & 0.025/0.023 & 0.964/0.946 \\ 
        dr-Cox1-{\color{red}Cox2} & -0.0050 & 0.022 & 0.024/0.022 & 0.956/0.948 \\
        dr-{\color{red}Cox2}-Cox1 & -0.0054 & 0.022 & 0.024/0.022 & 0.960/0.952 \\
        dr-{\color{red}Cox2}-{\color{red}Cox2} & 0.0140  & 0.021 & 0.024/0.021 & 0.938/0.882 \\ 
        \midrule
		IPW.Q-Cox1 & -0.0048 & 0.022 & 0.021/0.022 & 0.934/0.952 \\ 
		IPW.Q-{\color{red}Cox2} & 0.0142 & 0.021 & 0.020/0.021 & 0.872/0.882 \\
		\midrule
        Reg.T1-Cox1 & -0.0030 & 0.022 & - /0.022 & - /0.944 \\
        Reg.T1-{\color{red}Cox2} & 0.0154 & 0.020 & - /0.020 & - /0.876 \\
        \midrule
        Reg.T2-Cox1 & -0.0053 & 0.022 & - /0.021 & - /0.942 \\ 
        Reg.T2-{\color{red}Cox2} & 0.0125 & 0.020 & - /0.020 & - /0.898 \\
		\midrule
        PL & 0.0165 & 0.020 & - /0.020 & - /0.852 \\ 
        naive & 0.1466  & 0.015 & - /0.015 & - /0.000 \\
		full data & -0.0008 & 0.013 & 0.013/0.013 & 0.936/0.934 \\
  \bottomrule
	\end{tabular}
\end{table}

\begin{figure}[H]
	\centering
    \includegraphics[width=1\textwidth]{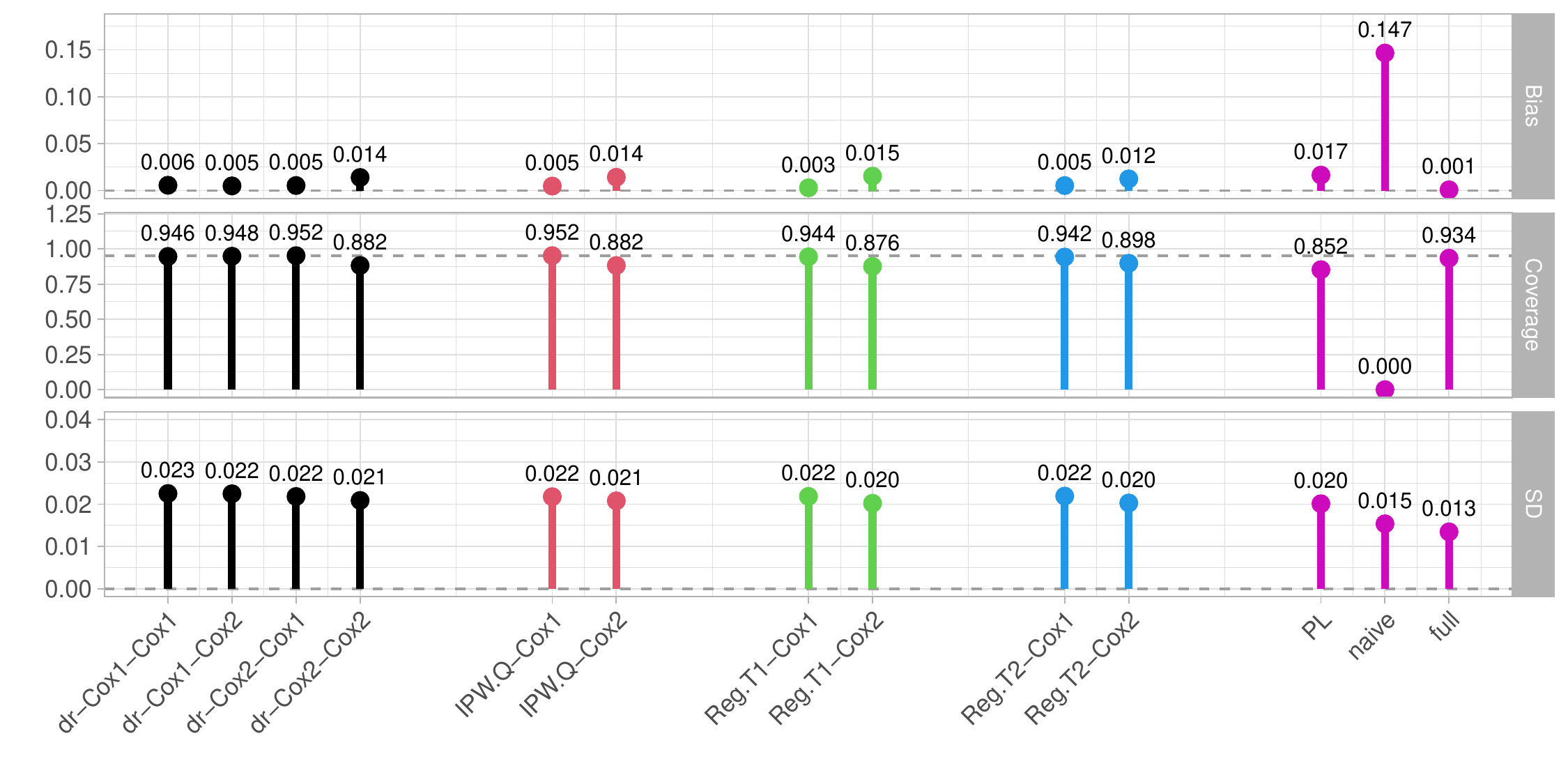}
	\caption{Absolute bias, coverage rate of  95\% confidence intervals  with bootstrap standard errors, and empirical standard deviation for different estimator, corresponding to the `censoring after truncation' scenario.}
	\label{fig:c2_simu}
\end{figure}

\clearpage
\section{Applications}

This section contains  plots for examining the overlap assumption, 
confidence interval plots and tables of the data analysis results, as well as model diagnostics for the two data applications. 
For CNS lymphoma data percentile method with 2000 bootstrap replications is applied due to the small sample size of the data.  

To examine the overlap assumption for the CNS lymphoma data, we look at $\hat F_x$ and $\hat G$ estimated from the Cox models. We do not use the estimates from the relative risk forest \citep{yao2020ensemble} because they may be unstable due to the small sample size of the data. We take $\tau_1 = \min_{i=1}^n X_i$ and $\tau_2 = \max_{i=1}^n Q_i$, which are the minimum jump time of $\hat F_x$ and the maximum jump time of $\hat G$, respectively.
For this data, we denote $\eta_{1}$ the probability $1-\hat F_x(\tau_2|Z)$ or $1-\hat F_x(X|Z)$, whichever is larger; denote $\eta_{2}$ the probability $\hat G(\tau_1|Z)$ or $\hat G(Q|Z)$, whichever is larger; and denote $\eta_3 = \hat S_c(X)$. 
The boxplots of $\eta_1$, $\eta_2$ and $\eta_3$ and their minimum values are shown.
We see that $\eta_1$, $\eta_2$ and $\eta_3$ are bounded away from zero among the uncensored subjects, confirming the overlap condition needed in practice. 

For the HAAS data, we look at $\hat F$ and $\hat G$ estimated from the Cox models. We take $\tau_1 = \min_{i:\Delta_i=1} X_i$and $\tau_2 = \max_{i:\Delta_i=1} Q_i$, which are the minimum jump time of $\hat F$ and the maximum jump time of $\hat G$, respectively. 
For this data, again denote $\eta_{1}$ the probability $1-\hat F(\tau_2|Z)$ or $1-\hat F(X|Z)$, whichever is larger; denote $\eta_{2}$ the probability $\hat G(\tau_1|Z)$ or $\hat G(Q|Z)$, whichever is larger; and denote $\eta_3 = \hat S_D(X-Q)$. 
The boxplots of $\eta_1$, $\eta_2$ and $\eta_3$ and their minimum values are shown. 
We see that $\eta_1$, $\eta_2$ and $\eta_3$ are bounded away from zero among the uncensored subjects, confirming the overlap condition needed in practice. 

\begin{figure}[H]
	\centering
	\includegraphics[width=0.7\textwidth]{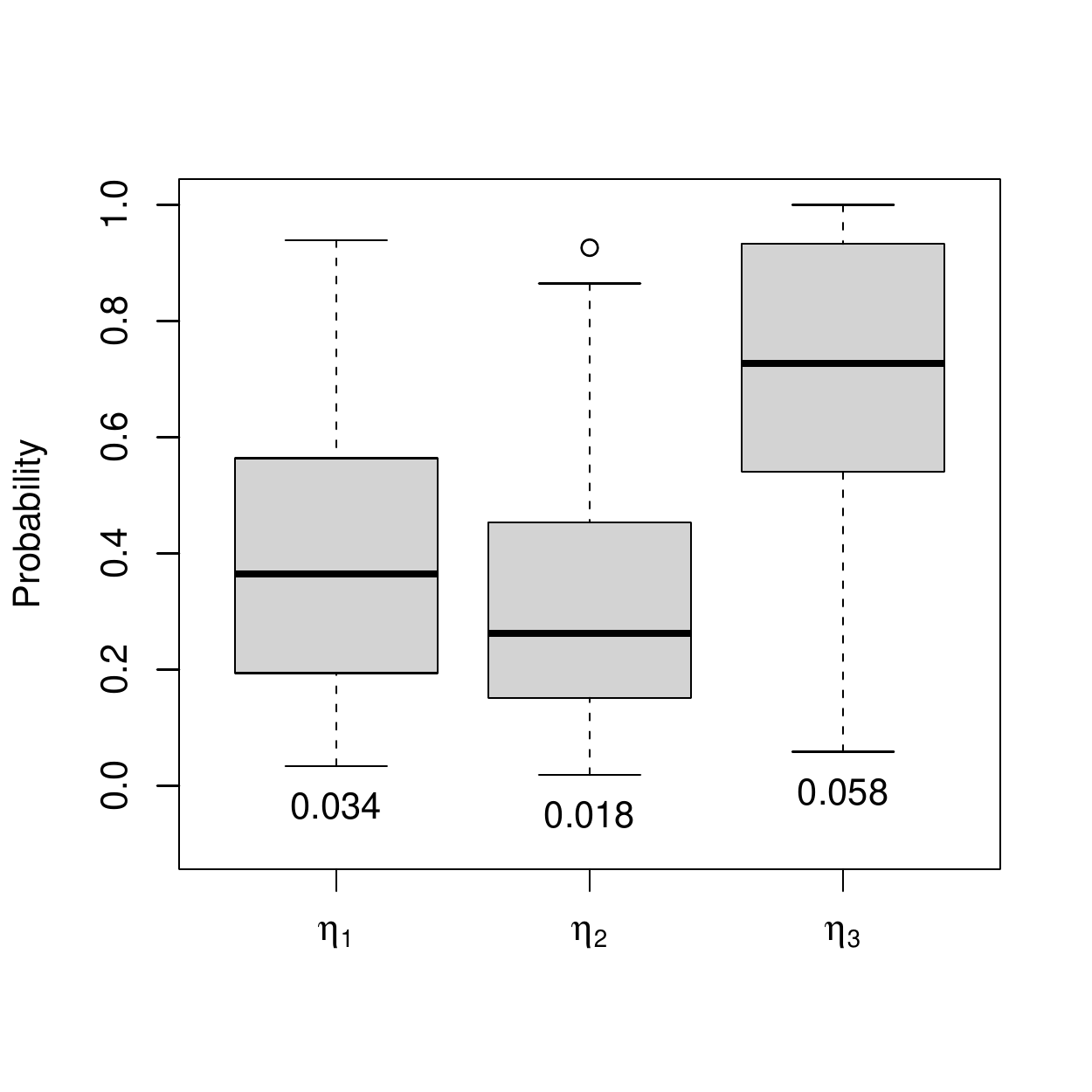}
	\caption{Boxplot of $\eta_1$, $\eta_2$, and $\eta_3$ among the uncensored subjects for the CNS lymphoma data; the minimum value of each quantity is marked on the plot.}
	\label{fig:CNS_overlap_practice.}
\end{figure}

\begin{figure}[H]
	\centering
	\includegraphics[width=1\textwidth]{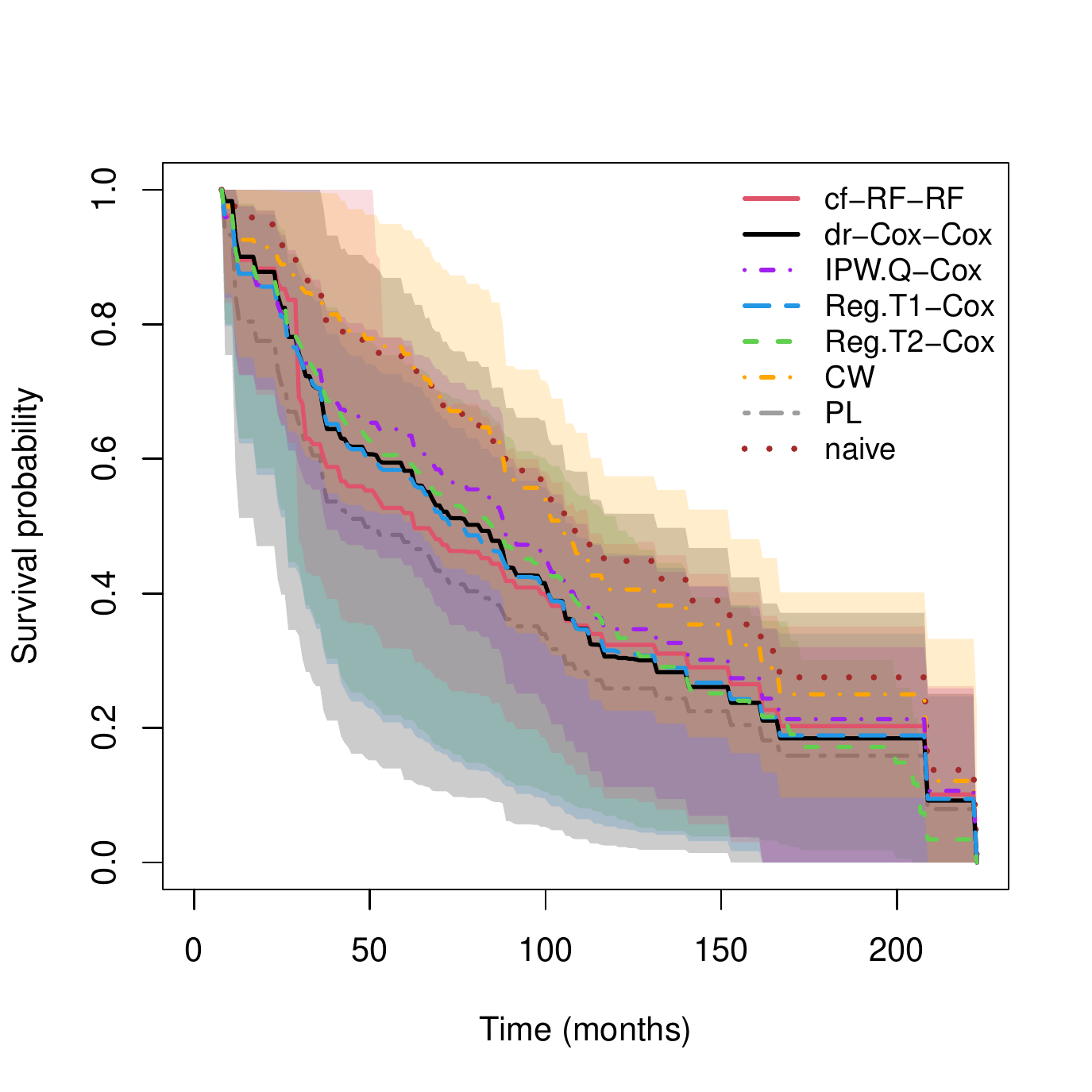}
	\caption{Estimated  overall survival probabilities and their 95\% bootstrap confidence intervals (shaded, except for PL and naive) for the CNS lymphoma data.}
	\label{fig:CNS_survCI_3}
\end{figure}

\begin{table}[H]
	\centering
    \renewcommand{\arraystretch}{0.6}
	\caption{
 Estimates of the overall survival probabilities at 36, 60, 120, and 180 months and their 95\% bootstrap confidence intervals for the CNS lymphoma data. }
 \resizebox{\textwidth}{!}{%
		\begin{tabular}{lcccc}
  \\
  \toprule
    & \multicolumn{4}{c}{Estimate (95\% bootstrap CI)}\\
   \midrule
    Time (months) & 36 & 60 & 120 & 180 \\ 
			\midrule
  cf-RF-RF & 0.616 (0.426, 1.000) & 0.524 (0.318, 0.767) & 0.314 (0.095, 0.473) & 0.203 (0.000, 0.351) \\ 
  dr-Cox-Cox & 0.705 (0.262, 1.000) & 0.582 (0.123, 0.846) & 0.304 (0.023, 0.519) & 0.185 (0.000, 0.371) \\ 
  IPW.Q & 0.731 (0.554, 0.866) & 0.635 (0.442, 0.776) & 0.347 (0.112, 0.455) & 0.213 (0.000, 0.320) \\ 
  Reg.T1 & 0.705 (0.347, 0.834) & 0.570 (0.199, 0.713) & 0.312 (0.050, 0.459) & 0.189 (0.000, 0.340) \\ 
  Reg.T2 & 0.726 (0.355, 0.854) & 0.593 (0.212, 0.741) & 0.334 (0.064, 0.493) & 0.172 (0.018, 0.301) \\ 
  CW & 0.841 (0.584, 0.997) & 0.756 (0.490, 0.943) & 0.406 (0.230, 0.570) & 0.250 (0.088, 0.414) \\ 
  PL & 0.605 (0.433, 0.775) & 0.476 (0.327, 0.628) & 0.259 (0.157, 0.380) & 0.159 (0.072, 0.270) \\ 
  naive & 0.854 (0.781, 0.926) & 0.743 (0.654, 0.830) & 0.448 (0.335, 0.568) & 0.275 (0.136, 0.424) \\ 
  \bottomrule
	\end{tabular}}
	\label{tab:CNS}
\end{table}

\begin{figure}
\centering
 \begin{subfigure}[H]{0.45\textwidth}
		\includegraphics[width=1\textwidth]{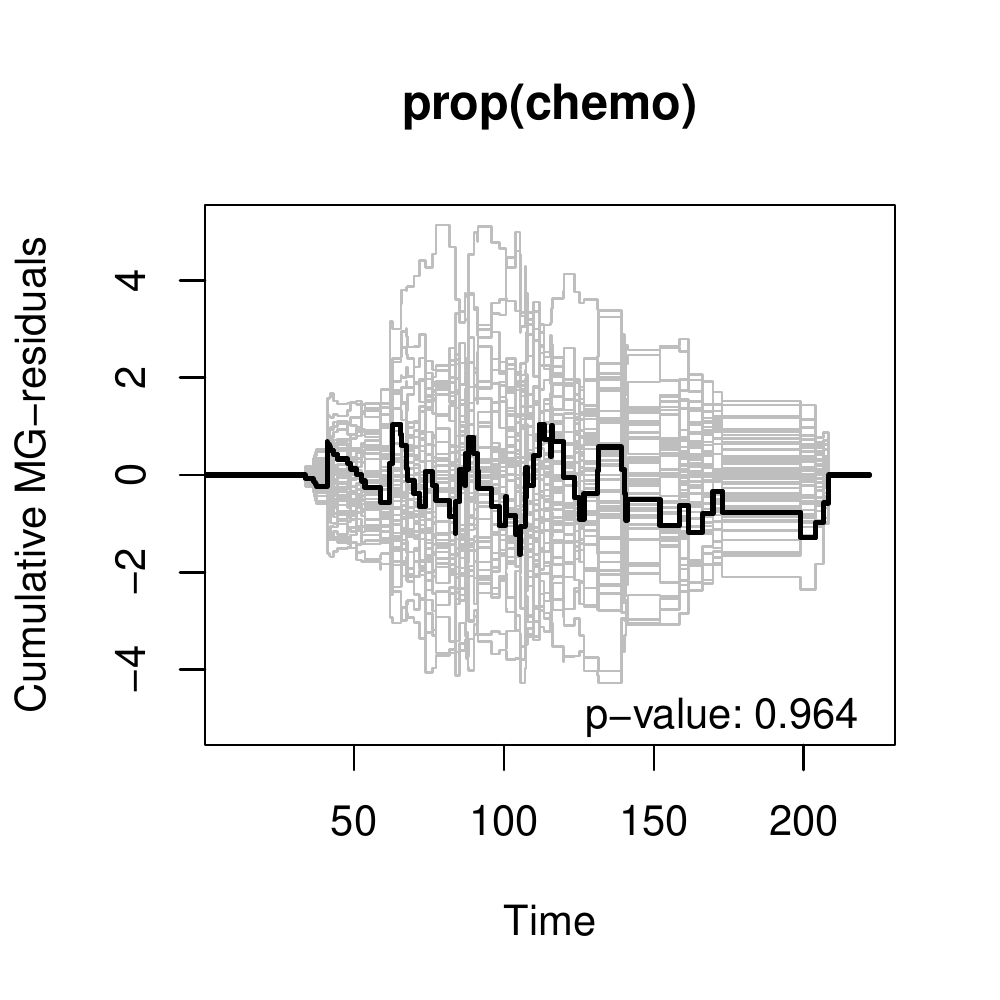}
	\end{subfigure}
 \begin{subfigure}[H]{0.45\textwidth}
		\includegraphics[width=1\textwidth]{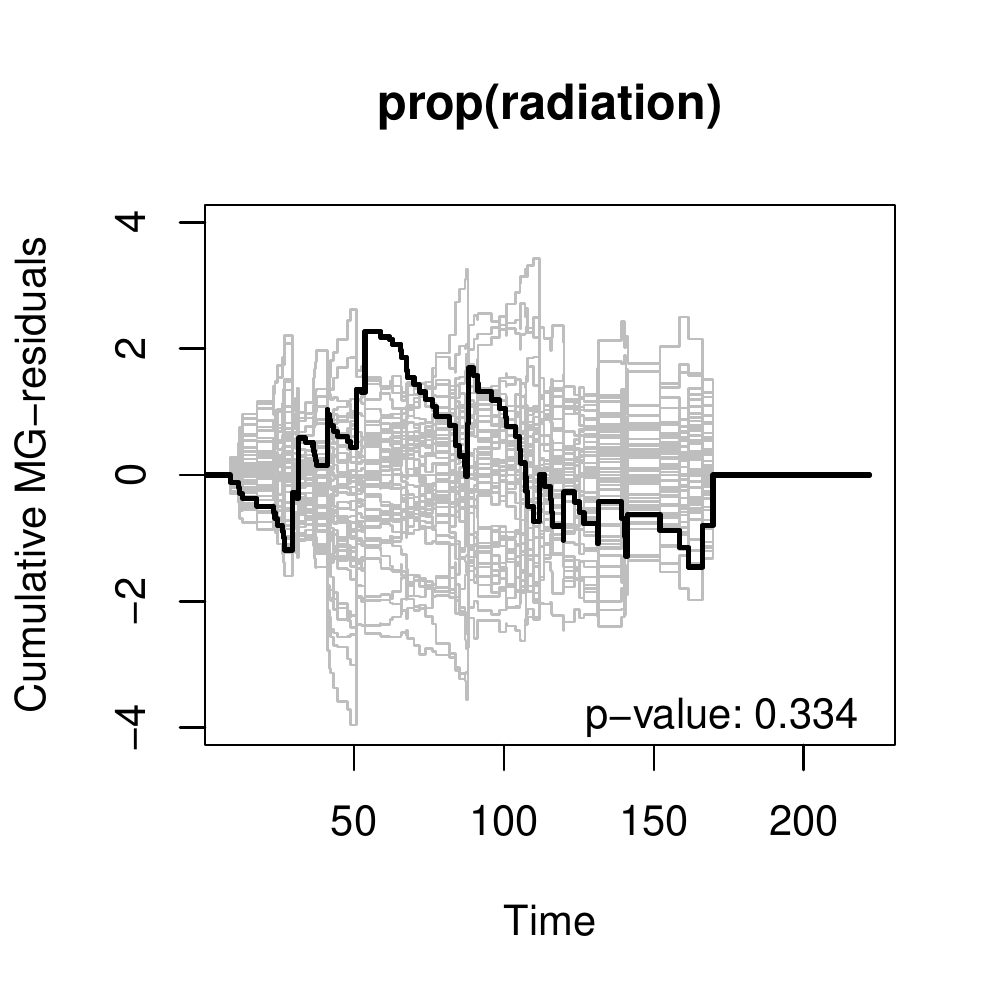}
\end{subfigure}
	\hfill
 \begin{subfigure}[H]{0.45\textwidth}
		\includegraphics[width=1\textwidth]{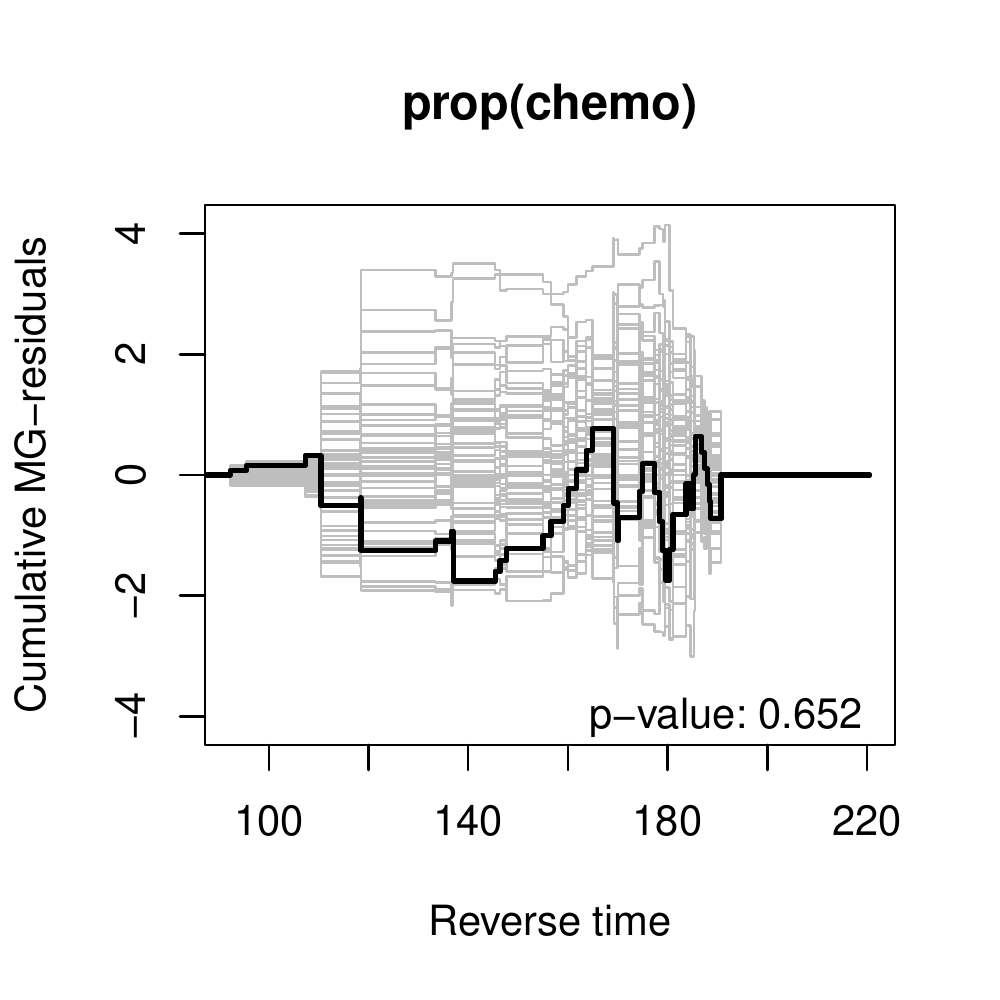}
	\end{subfigure}
 \begin{subfigure}[H]{0.45\textwidth}
		\includegraphics[width=1\textwidth]{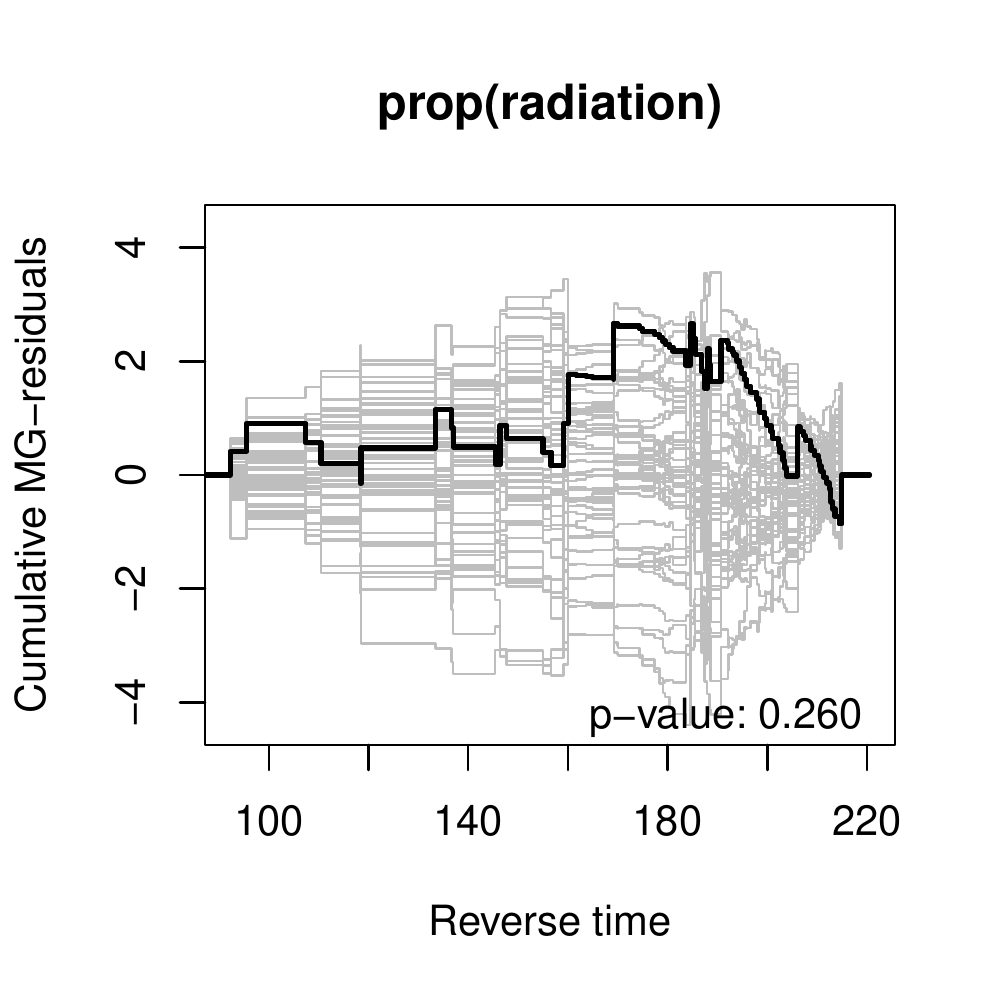}
\end{subfigure}
	\caption{Cumulative martingale residual plots for checking the proportional hazards assumption in the Cox model for estimating $F_x$ (top) and $G$ (bottom), respectively, for the CNS lymphoma data.}
	\label{fig:c1.CNS_Mresidual}
\end{figure}

\clearpage

\begin{figure}[H]
	\centering
	\includegraphics[width=0.7\textwidth]{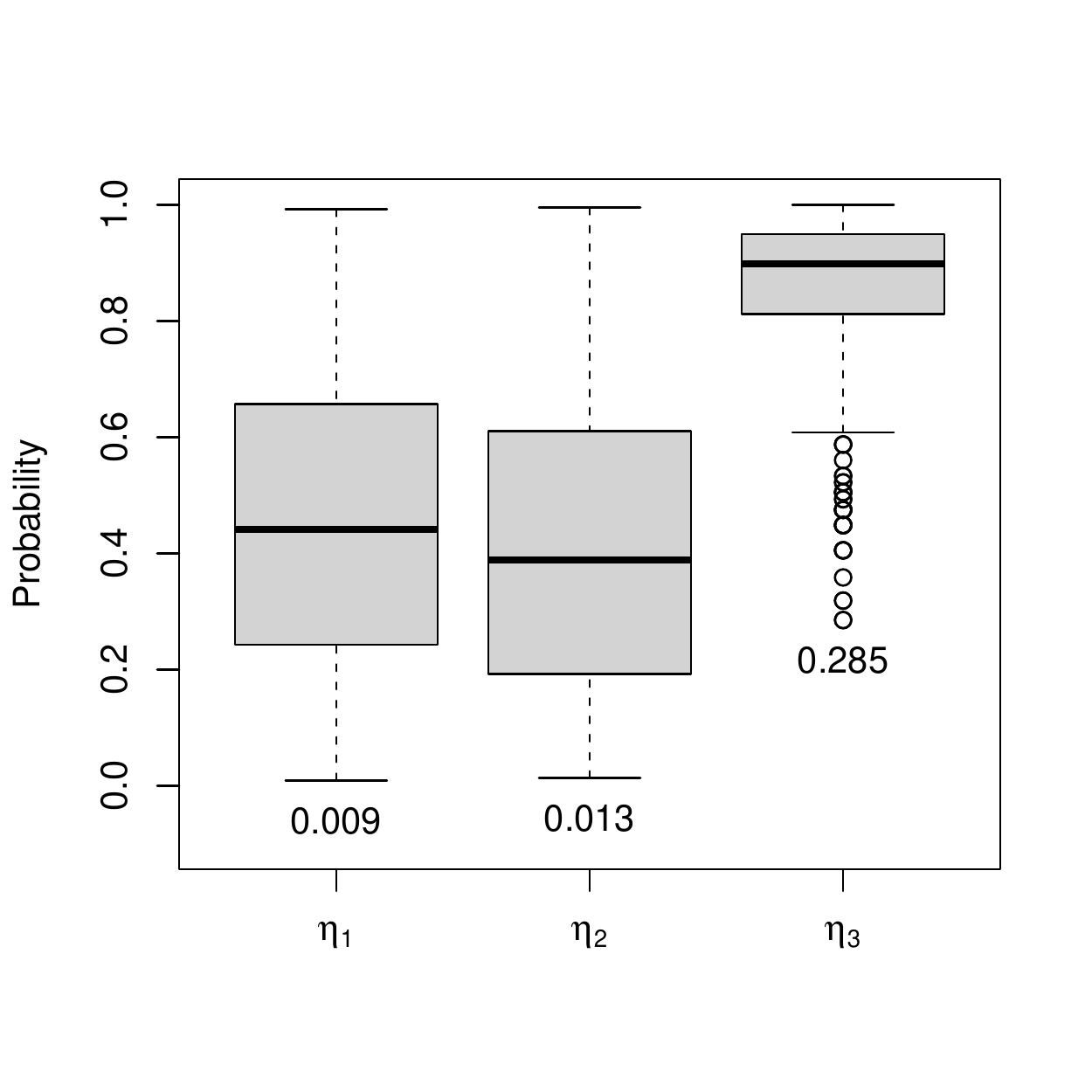}
	\caption{Boxplot of $\eta_1$, $\eta_2$, and $\eta_3$ among the uncensored subjects for the HAAS data; the minimum value of each quantity is marked on the plot.}
	\label{fig:HAAS_overlap_practice.}
\end{figure}

\begin{figure}[H]
	\centering
	\includegraphics[width=01\textwidth]{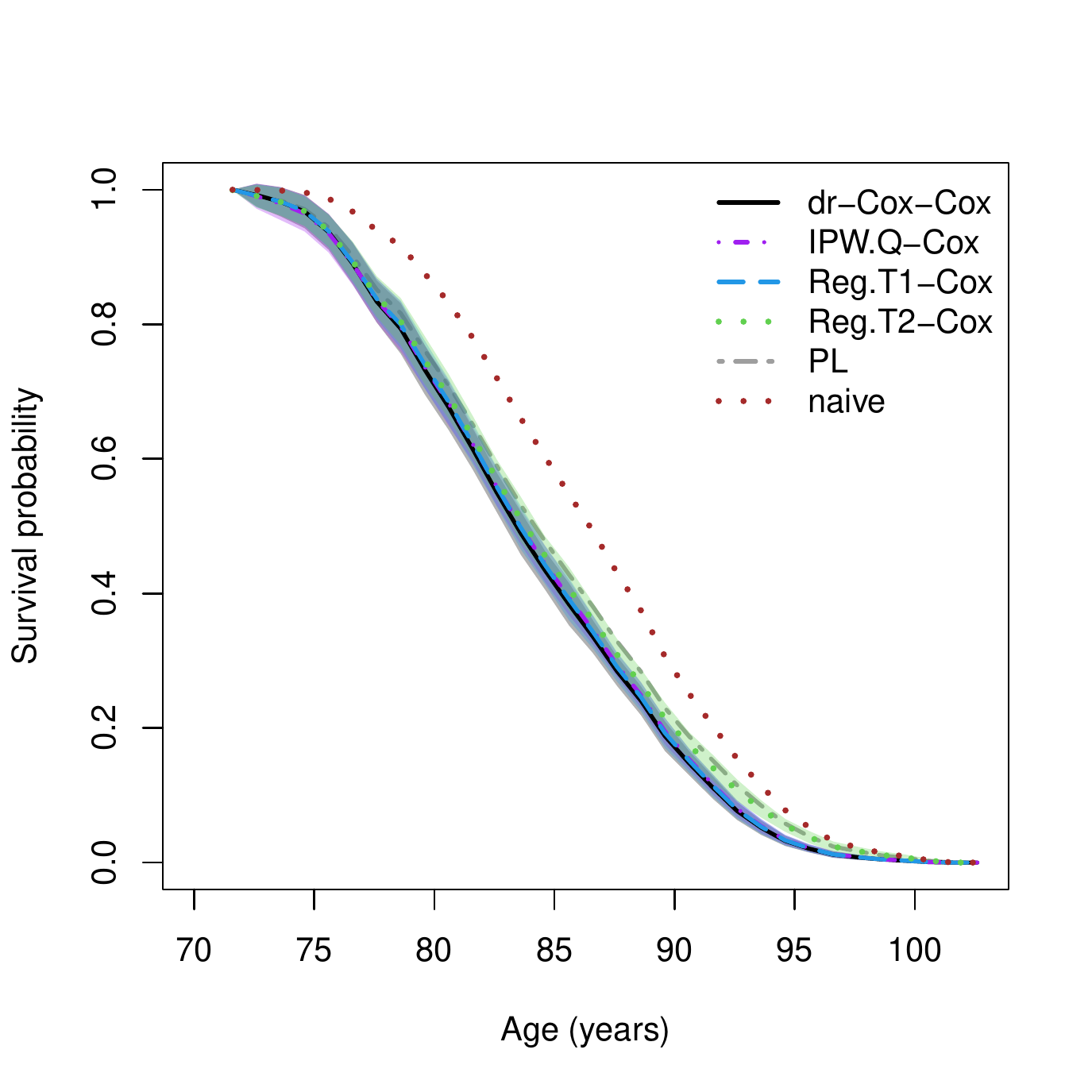}
	\caption{Estimates of the cognitive impairment-free survival and their 95\% bootstrap confidence intervals (shaded, except for PL and naive) for the HAAS data}
	\label{fig:HAAS_survCI}
\end{figure}

\begin{table}[H]
	\centering
    \renewcommand{\arraystretch}{0.6}
	\caption{Estimates of the impairment-free survival at different ages and their 95\% bootstrap confidence intervals for the HAAS data. }
 \resizebox{\textwidth}{!}{%
		\begin{tabular}{lcccc}
  \\
  \toprule
    & \multicolumn{4}{c}{Estimates (95\% bootstrap CIs)}\\
   \midrule
    Age (years) & 80 & 85 & 90 & 95 \\
	\midrule
 dr-Cox-Cox & 0.710 (0.678, 0.742) & 0.418 (0.388, 0.447) & 0.172 (0.155, 0.189) & 0.026 (0.020, 0.033) \\ 
  IPW.Q-Cox & 0.715 (0.683, 0.746) & 0.427 (0.401, 0.454) & 0.179 (0.163, 0.194) & 0.028 (0.021, 0.034) \\ 
  Reg.T1-Cox & 0.717 (0.686, 0.749) & 0.425 (0.396, 0.455) & 0.176 (0.159, 0.193) & 0.027 (0.021, 0.034) \\ 
  Reg.T2-Cox & 0.725 (0.694, 0.756) & 0.440 (0.410, 0.470) & 0.198 (0.180, 0.217) & 0.047 (0.038, 0.056) \\ 
  PL & 0.740 (0.715, 0.765) & 0.462 (0.440, 0.483) & 0.213 (0.197, 0.229) & 0.049 (0.041, 0.056) \\ 
  naive & 0.858 (0.845, 0.872) & 0.585 (0.565, 0.606) & 0.285 (0.264, 0.306) & 0.065 (0.053, 0.078) \\ 
  \bottomrule
	\end{tabular}}
	\label{tab:HAAS}
\end{table}

\begin{figure}
	\centering
 \begin{subfigure}[H]{0.4\textwidth}
		\includegraphics[width=1\textwidth]{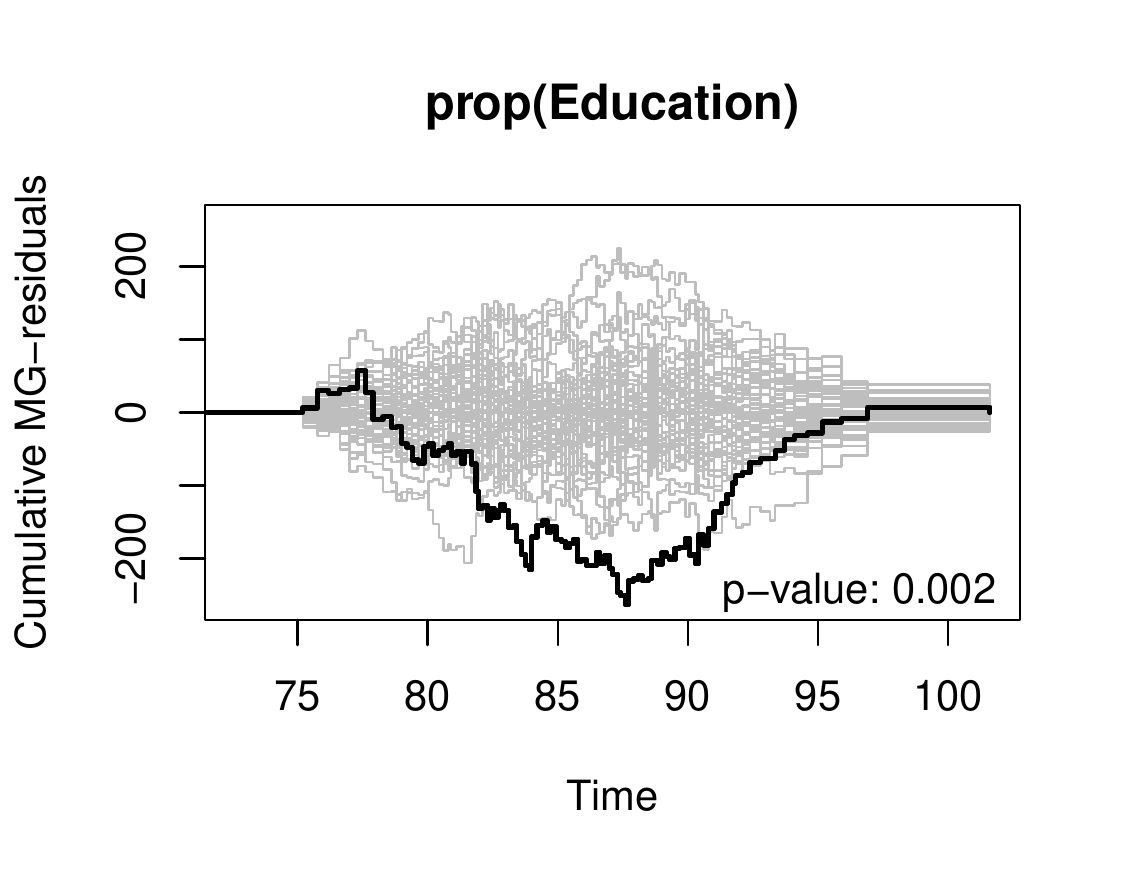}
	\end{subfigure}
 \begin{subfigure}[H]{0.4\textwidth}
		\includegraphics[width=1\textwidth]{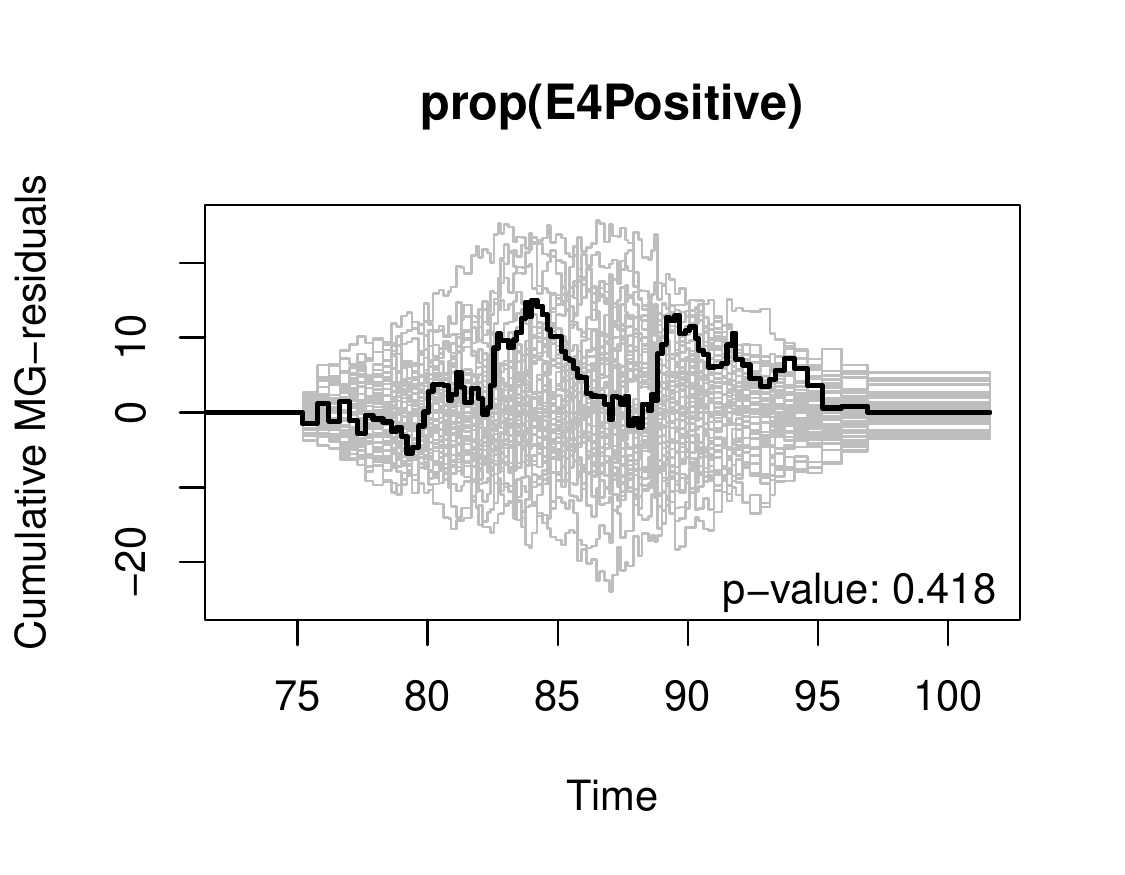}
\end{subfigure}
 \begin{subfigure}[H]{0.4\textwidth}
		\includegraphics[width=1\textwidth]{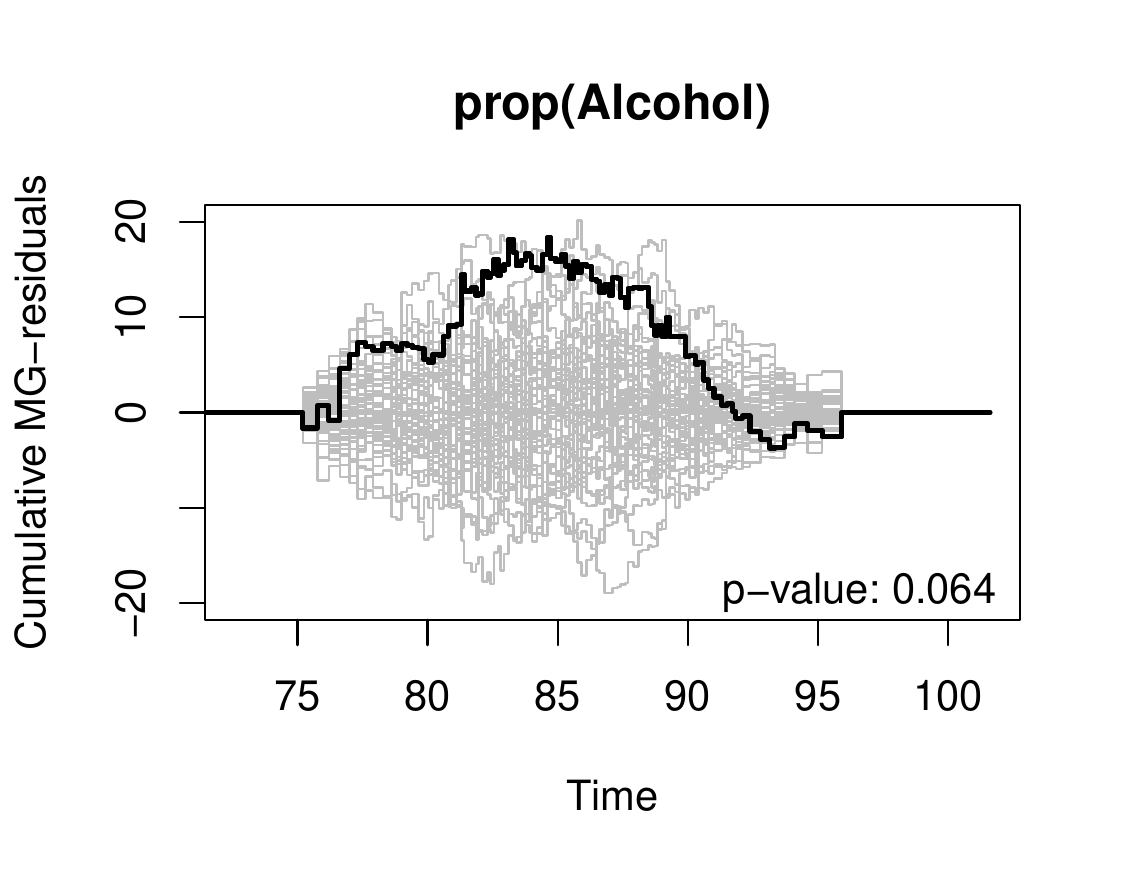}
	\end{subfigure}
 \begin{subfigure}[H]{0.4\textwidth}
		\includegraphics[width=1\textwidth]{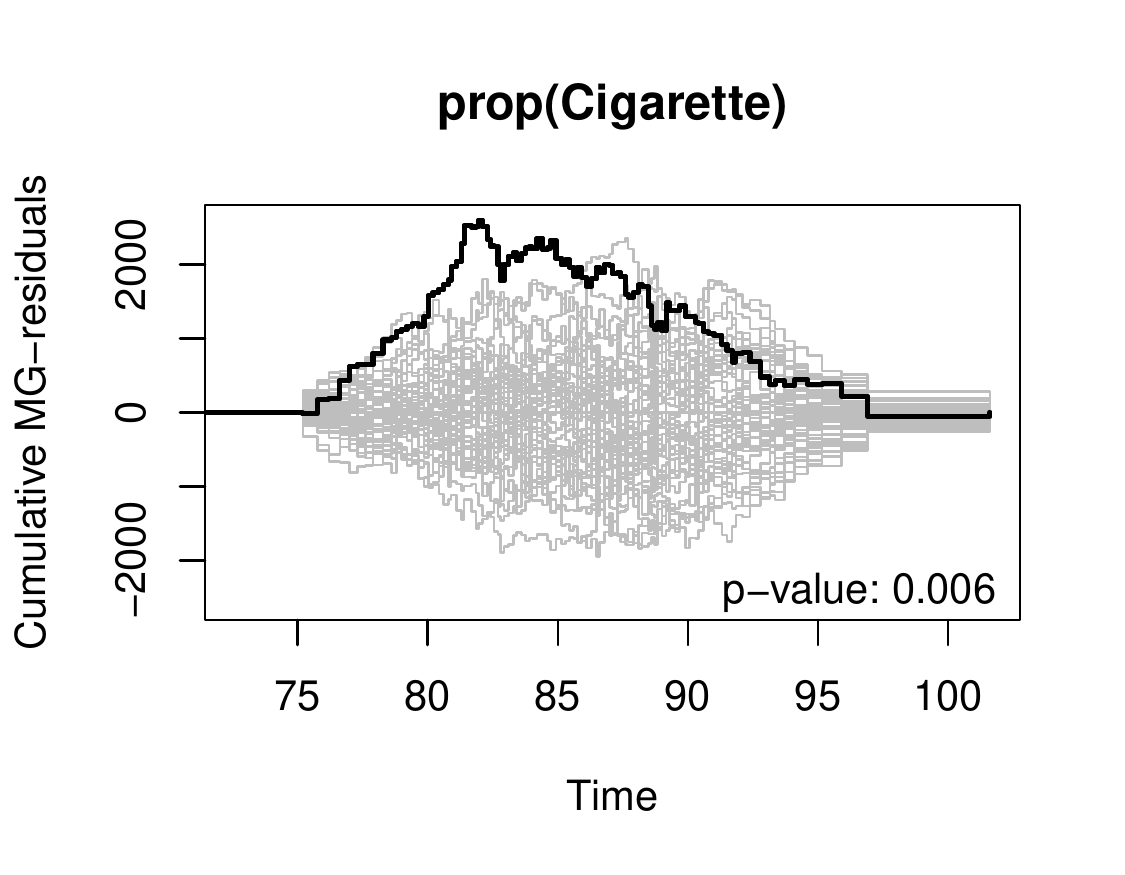}
\end{subfigure}
	\hfill
  \begin{subfigure}[H]{0.4\textwidth}
		\includegraphics[width=1\textwidth]{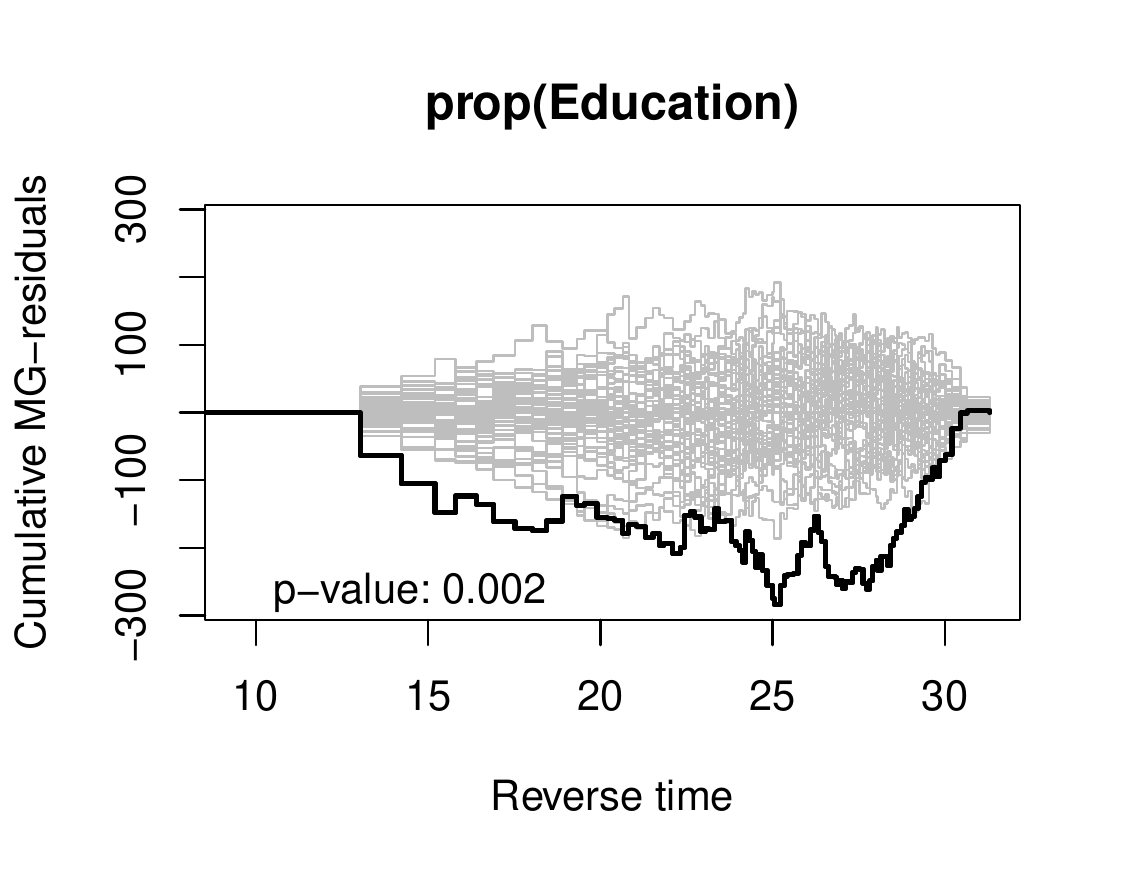}
	\end{subfigure}
 \begin{subfigure}[H]{0.4\textwidth}
		\includegraphics[width=1\textwidth]{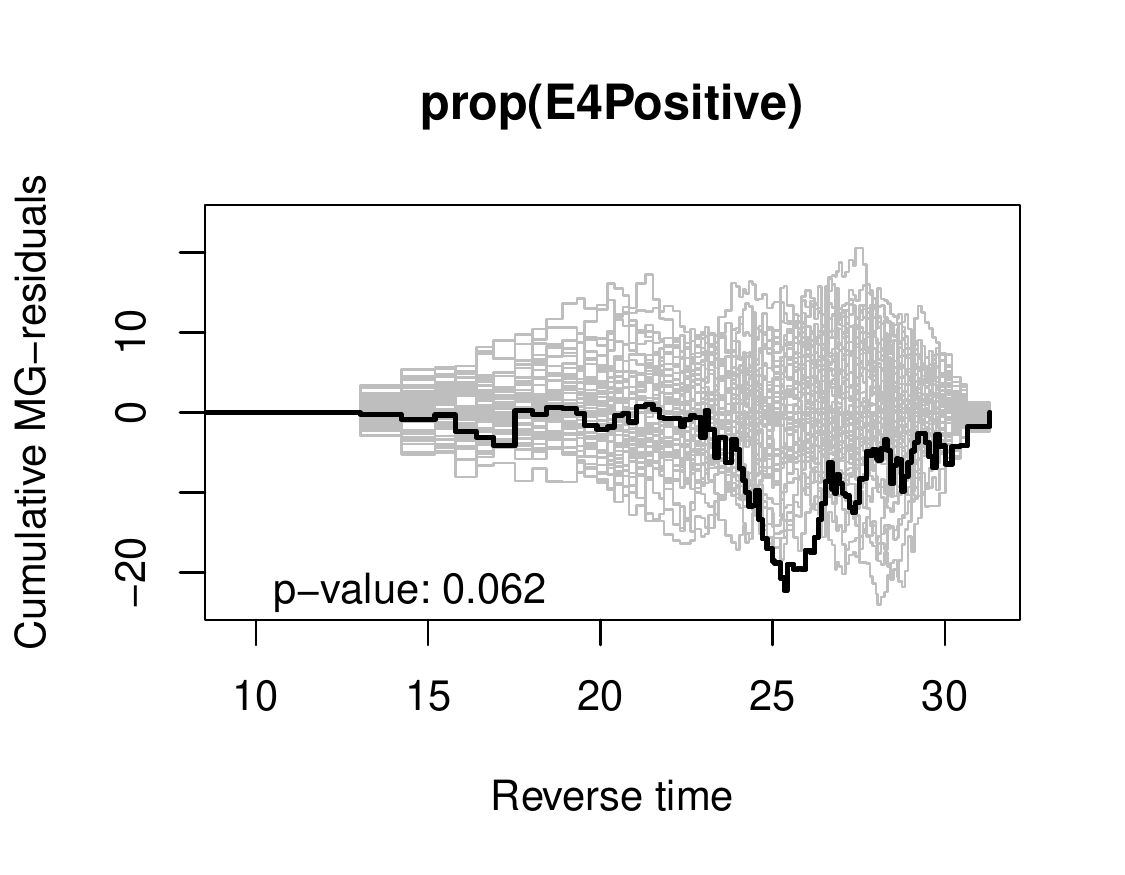}
\end{subfigure}
 \begin{subfigure}[H]{0.4\textwidth}
		\includegraphics[width=1\textwidth]{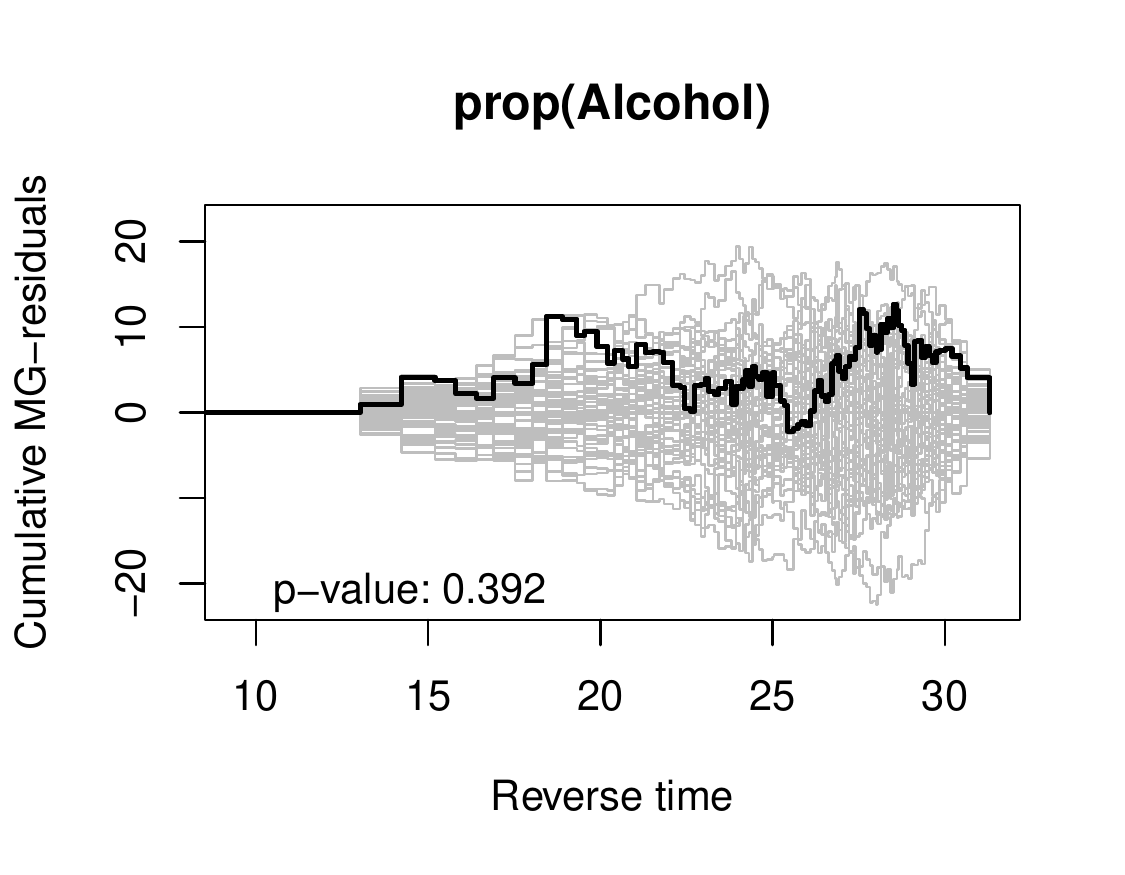}
	\end{subfigure}
 \begin{subfigure}[H]{0.4\textwidth}
		\includegraphics[width=1\textwidth]{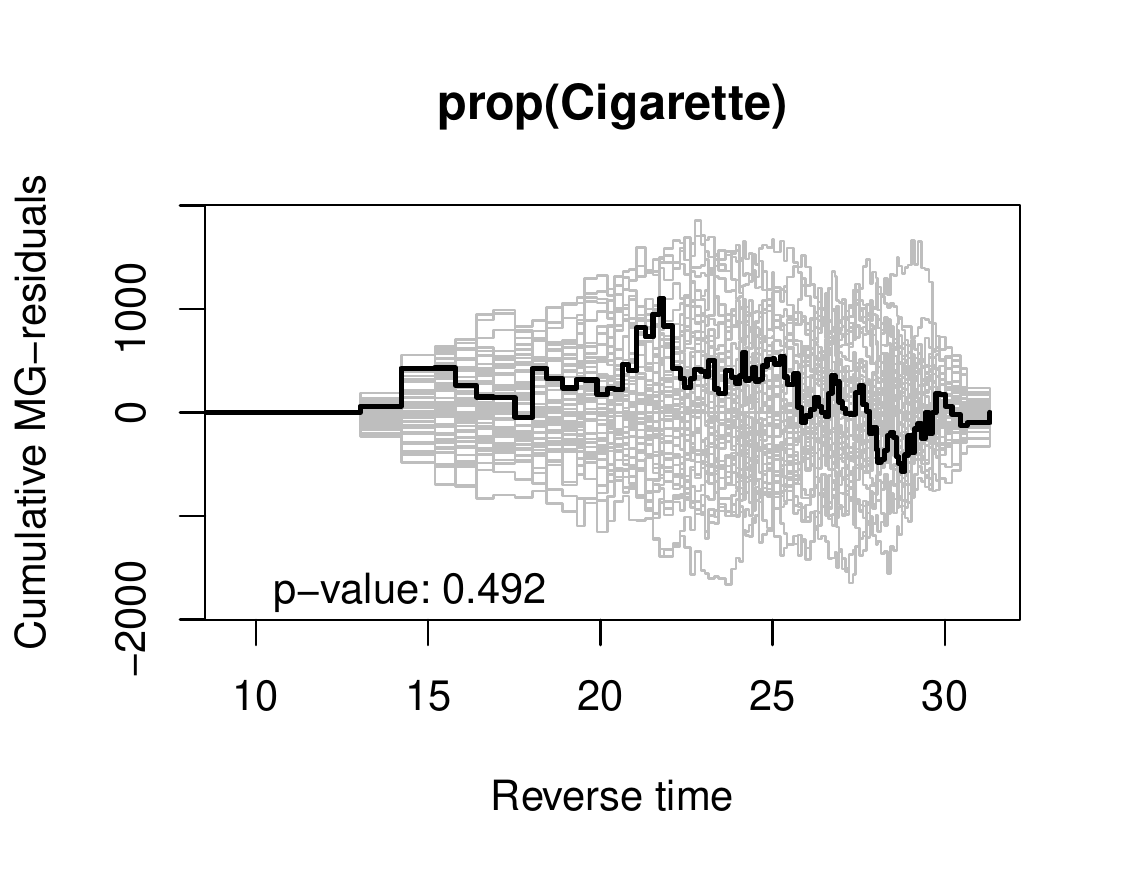}
 \end{subfigure}
	\caption{Cumulative martingale residual plots for checking the proportional hazards assumption in the Cox model for estimating $F$ (top four) and $G$ (bottom four), respectively, for the HAAS data.}
	\label{fig:c2.HAAS_Mresidual}
\end{figure}

\newpage
\section{Equivalence with the IF in Chao (1987)}\label{sec:chao1987}

In this section, we compare the influence function in \cite{chao1987influence} with the EIC derived in this paper. \citet{chao1987influence} considered the setting without covariates, i.e., $Z=\emptyset$ in our notation; and they assumed the random left truncation assumption that $Q$ and $T$ are independent in full data. They derived the influence function for the product-limit estimator of the survival function.
We will show in this section that in their setting, our EIC coincides with the influence function in their paper when estimating the probability of $T$ greater than $t_0$ in full data for some $t_0>0$. This is the parameter of interest $\theta$ defined in \eqref{eq:theta} with $\nu(t) = \mathbbm{1}(t<t_0)$. We note that the random left truncation implies the quasi-independence assumption \citep{tsai1990testing}, which is our Assunmption \ref{ass:quasi-indpendent} when $Z = \emptyset$. We follow the notation in this paper, and denote $F$ and $G$ be the CDF's of $T$ and $Q$ respectively. 

By equation (1.2) in \citet{chao1987influence}, the influence function for the product-limit estimator of $\theta$ is 
\begin{align}
\tilde \varphi(Q,T) & =  \{1-F(t_0)\}\left[ \beta \int_0^{t_0} \frac{\mathbbm{1}(Q\geq s) - \mathbbm{1}(T\geq s)}{G(s) \{1-F(s)\}^2} dF(s) + \frac{\beta  \mathbbm{1}(T\leq t_0)}{G(T)\{1-F(T)\}}\right]\\
& =  \beta \{1-F(t_0)\}\left[ \int_0^\infty \mathbbm{1}(s\leq t_0)  \frac{\mathbbm{1}(Q\geq s) - \mathbbm{1}(T\geq s)}{G(s) \{1-F(s)\}^2} dF(s) + \frac{\mathbbm{1}(T\leq t_0)}{G(T)\{1-F(T)\}}\right].
\end{align}
For observed data, we have $Q<T$, so
\begin{align}
\tilde \varphi(Q,T) 
& =  \beta\{1-F(t_0)\}\left[ - \int_0^\infty \mathbbm{1}(s\leq t_0)  \frac{\mathbbm{1}(Q<s\leq T)}{G(s) \{1-F(s)\}^2} dF(s) + \frac{\mathbbm{1}(T\leq t_0)}{G(T)\{1-F(T)\}}\right]\\
& =  \beta(1-\theta)\left[ - \int_0^\infty \frac{\mathbbm{1}(s\leq t_0) \mathbbm{1}(Q<s\leq T)}{G(s) \{1-F(s)\}^2} dF(s) + \frac{\mathbbm{1}(T\leq t_0)}{G(T)\{1-F(T)\}}\right].
\end{align}

On the other hand, by Lemma \ref{lem:IF}, the influence function derived in this paper under the setting without covariates is

\begin{align}
\varphi(Q,T) & = \beta\left[\frac{\mathbbm{1}(T\leq t_0) - \theta}{G(T)} - \int_0^\infty \frac{m(v;F)-\theta F(v)}{G(v)\{1-F(v)\}} d\bar M_Q(v;G) \right] \\
& = \beta\left[\frac{\mathbbm{1}(T\leq t_0) - \theta}{G(T)} - \int_0^\infty \frac{m(v;F)-\theta F(v)}{G(v)\{1-F(v)\}} \left\{d\bar N_Q(v;G) + \mathbbm{1}(Q\leq v<T) \frac{dG(v)}{G(v)} \right\}\right] \\
& = \beta\left[\frac{\mathbbm{1}(T\leq t_0) - \theta}{G(T)} 
+ \frac{m(Q;F)-\theta F(Q)}{G(Q)\{1-F(Q)\}} 
- \int_0^\infty \frac{m(v;F)-\theta F(v)}{G(v)\{1-F(v)\}} \mathbbm{1}(Q\leq v<T) \frac{dG(v)}{G(v)} \right],
\end{align}
where 
\begin{align}
\beta & = \int_0^\infty \mathbbm{1}(q<t)\ dF(t)\ dG(q), \\
m(v;F) & = \int_0^v \mathbbm{1}(t\leq t_0)\ dF(t) = F(t_0\wedge v),
\end{align}
so we have 
\begin{align}
\varphi(Q,T) & = \beta\left[\frac{\mathbbm{1}(T\leq t_0) - \theta}{G(T)} 
+ \frac{F(t_0\wedge Q)-\theta F(Q)}{G(Q)\{1-F(Q)\}} 
- \int_0^\infty \frac{F(t_0\wedge v)-\theta F(v)}{G(v)\{1-F(v)\}} \mathbbm{1}(Q\leq v<T) \frac{dG(v)}{G(v)} \right] \\
& = \beta\left[\frac{\mathbbm{1}(T\leq t_0) - \theta}{G(T)} 
+ \frac{F(t_0\wedge Q)-\theta F(Q)}{G(Q)\{1-F(Q)\}} 
- \int_Q^T \frac{F(t_0\wedge v)-\theta F(v)}{G(v)^2\{1-F(v)\}} dG(v) \right].
\end{align}

Since $Q<T$ for observed data, there are three possible situations: $t_0<Q<T$, $Q\leq t_0 \leq T$, and $Q < T <t_0$.
We will verify that $\varphi = \tilde \varphi$ in all three situations. We note that $F(t_0) = \theta$.

When $t_0<Q<T$, 
\begin{align}
\tilde \varphi(Q,T) &= 0,\\
\varphi(Q,T) &= \beta\left[ \frac{-\theta}{G(T)} 
+ \frac{F(t_0)-\theta F(Q)}{G(Q)\{1-F(Q)\}} 
- \int_Q^T \frac{F(t_0)-\theta F(v)}{G(v)^2\{1-F(v)\}} dG(v)  \right] \\
&= \beta\left[ -\frac{\theta}{G(T)} 
+ \frac{\theta\{1- F(Q)\}}{G(Q)\{1-F(Q)\}} 
- \int_Q^T \frac{\theta \{1-F(v)\}}{G(v)^2\{1-F(v)\}} dG(v)  \right]\\
&= \beta\left[ -\frac{\theta}{G(T)} 
+ \frac{\theta}{G(Q)} 
- \int_Q^T \frac{\theta }{G(v)^2} dG(v) \right] = 0.
\end{align}
When $Q\leq t_0 \leq T$,
\begin{align}
\tilde \varphi(Q,T) = - \beta(1-\theta)  \int_Q^{t_0} \frac{d F(v)}{G(v)\{1-F(v)\}^2},
\end{align}
and
\begin{align}
\varphi(Q,T) 
& = \beta\left[\frac{ - \theta}{G(T)} 
+ \frac{F(Q)-\theta F(Q)}{G(Q)\{1-F(Q)\}} 
- \int_Q^T\frac{F(t_0\wedge v)-\theta F(v)}{G(v)^2\{1-F(v)\}} dG(v) \right] \\
& = \beta\left[-\frac{\theta}{G(T)} 
+ \frac{F(Q)-\theta F(Q)}{G(Q)\{1-F(Q)\}} \right. \\
&\quad\quad\quad \left. - \int_Q^{t_0}\frac{F(v)-\theta F(v)}{G(v)^2\{1-F(v)\}} dG(v)
- \int_{t_0}^T\frac{F(t_0)-\theta F(v)}{G(v)^2\{1-F(v)\}} dG(v)\right] \\
& = \beta\left[-\frac{\theta}{G(T)} 
+ \frac{F(Q)-\theta F(Q)}{G(Q)\{1-F(Q)\}} \right. \\
&\quad\quad\quad \left. - \int_Q^{t_0}\frac{ (1-\theta) F(v)}{G(v)^2\{1-F(v)\}} dG(v)
- \int_{t_0}^T\frac{\theta \{1- F(v)\}}{G(v)^2\{1-F(v)\}} dG(v)\right] \\
& = \beta\left[- \frac{ \theta}{G(T)} 
+ \frac{(1-\theta) F(Q)}{G(Q)\{1-F(Q)\}} \right. \\
&\quad\quad\quad \left. - (1-\theta) \int_Q^{t_0}\frac{F(v)}{G(v)^2\{1-F(v)\}} dG(v)
- \int_{t_0}^T\frac{\theta}{G(v)^2} dG(v)\right]. \label{eq:chao_1}
\end{align}
By integration by parts, we have
\begin{align}
&\quad -\int_Q^{t_0} \frac{F(v)}{G(v)^2\{1-F(v)\}}\ dG(v) \\
&  =  \int_Q^{t_0}\frac{F(v)}{1-F(v)}\ d\left\{\frac{1}{G(v)}\right\}\\
& = \left.\frac{F(v)}{G(v)\{1-F(v)\}}\right|_{v=Q}^{v=t_0} - \int_Q^{t_0} \frac{1}{G(v)} \left[\frac{dF(v)}{1-F(v)} + \frac{F(v) dF(v)}{\left\{1-F(v) \right\}^2}\right] \\
& = \frac{F(t_0)}{G(t_0)\left\{1-F(t_0)\right\}} - \frac{F(Q)}{G(Q)\{1-F(Q)\}} - \int_Q^{t_0} \frac{dF(v)}{G(v)\{1-F(v)\}^2}. \\
& = \frac{\theta}{G(t_0)(1-\theta)} - \frac{F(Q)}{G(Q)\{1-F(Q)\}} - \int_Q^{t_0} \frac{dF(v)}{G(v)\{1-F(v)\}^2}. \label{eq:chao_2}
\end{align}
Besides, 
\begin{align}
- \int_{t_0}^T\frac{\theta}{G(v)^2} dG(v)
&  = \frac{\theta}{G(T)} - \frac{\theta}{G(t_0)}. \label{eq:chao_3}
\end{align}
Plugging \eqref{eq:chao_2} and \eqref{eq:chao_3} into \eqref{eq:chao_1}, we have
\begin{align}
\varphi(Q,T) 
& = \beta\left[-\frac{\theta}{G(T)} 
+ \frac{(1-\theta) F(Q)}{G(Q)\{1-F(Q)\}} \right. \\
&\quad\quad\quad  + \frac{\theta}{G(t_0)} - \frac{(1-\theta)F(Q)}{G(Q)\{1-F(Q)\}} - (1-\theta)\int_Q^{t_0} \frac{dF(v)}{G(v)\{1-F(v)\}^2} \\
&\quad\quad\quad \left. +\frac{\theta}{G(T)} - \frac{\theta}{G(t_0)} \right] \\
& = \beta\left[-\frac{\theta}{G(T)} 
+ \frac{(1-\theta) F(Q)}{G(Q)\{1-F(Q)\}} \right. \\
&\quad\quad\quad  + \frac{\theta}{G(t_0)} - \frac{(1-\theta)F(Q)}{G(Q)\{1-F(Q)\}} - (1-\theta)\int_Q^{t_0} \frac{dF(v)}{G(v)\{1-F(v)\}^2} \\
&\quad\quad\quad \left. +\frac{\theta}{G(T)} - \frac{\theta}{G(t_0)} \right] \\
& = - \beta(1-\theta)  \int_Q^{t_0} \frac{dF(v)}{G(v)\{1-F(v)\}^2}
\end{align}
When $Q<T<t_0$, we have
\begin{align}
\tilde \varphi(Q,T) =  \beta(1-\theta)  \left\{ - \int_Q^T \frac{dF(v)}{G(v)\{1-F(v)\}^2} + \frac{1}{G(T)\{1-F(T)\}}  \right\},
\end{align}
and
\begin{align}
\varphi(Q,T) 
&=\beta\left[\frac{1 - \theta}{G(T)} 
+ \frac{F(Q)-\theta F(Q)}{G(Q)\{1-F(Q)\}} 
- \int_Q^T \frac{F(v) - \theta F(v)}{G(v)^2\{1-F(v)\}} dG(v) \right] \\
&=\beta\left[\frac{1 - \theta}{G(T)} 
+ \frac{(1-\theta) F(Q)}{G(Q)\{1-F(Q)\}} 
- (1-\theta)\int_Q^T \frac{F(v)}{G(v)^2\{1-F(v)\}} dG(v) \right]. \label{eq:chao_4}
\end{align}
Again, by integration by parts, we have
\begin{align}
-\int_Q^{T} \frac{F(v)}{G(v)^2\{1-F(v)\}}\ dG(v) 
= \frac{F(T)}{G(T)\{1-F(T)\}} - \frac{F(Q)}{G(Q)\{1-F(Q)\}} - \int_Q^{T} \frac{dF(v)}{G(v)\{1-F(v)\}^2}. \\
\label{eq:chao_5}
\end{align}
Plugging \eqref{eq:chao_5} into \eqref{eq:chao_4}, we have
\begin{align}
\varphi(Q,T) 
&=\beta\left[\frac{1 - \theta}{G(T)} 
+ \frac{(1-\theta) F(Q)}{G(Q)\{1-F(Q)\}} \right.\\
&\quad\quad\quad \left. + \frac{(1-\theta)F(T)}{G(T)\{1-F(T)\}} - \frac{(1-\theta)F(Q)}{G(Q)\{1-F(Q)\}} - (1-\theta)\int_Q^{T} \frac{dF(v)}{G(v)\{1-F(v)\}^2}
\right] \\
& = \beta(1-\theta) \left[ \frac{1}{G(Q)\{1-F(Q)\}} - \int_Q^{T} \frac{dF(v)}{G(v)\{1-F(v)\}^2} \right]
\end{align}

\medskip
Therefore, combining all three cases, we have $\varphi = \tilde \varphi$.

\end{document}